\documentclass[draft]{article}
\usepackage{amsfonts,amsmath,amssymb,color,enumerate,latexsym,times}
\definecolor{monblu}{cmyk}{0.80,0.00,0.00,0.70}
\definecolor{monmag}{cmyk}{0.00,1.00,0.00,0.00}
\newtheorem{claim}{Claim}
\newtheorem{proposition}{Proposition}
\newenvironment{proof}{{\bf Proof:}}{~$\dashv$\\}
\def\N{\mathbb{N}}
\def\Formula{\mathbf{A}}
\def\Formula{\mathbf{F}}
\def\Rule{\mathbf{R}}
\def\WK{\mathbf{WK}}
\def\LIK{\mathbf{LIK}}
\def\FIK{\mathbf{FIK}}
\def\Fo{\mathbf{Fo}}
\def\SFo{\mathbf{SFo}}
\def\allfra{\mathbf{all}}
\def\bcfra{\mathbf{b}}
\def\fcfra{\mathbf{f}}
\def\fbcfra{\mathbf{fb}}
\def\fdcfra{\mathbf{fd}}
\def\sym{\mathbf{sym}}
\def\deterministic{\mathbf{det}}
\def\Euclidean{\mathbf{Euc}}
\def\reflexive{\mathbf{ref}}
\def\transitive{\mathbf{tra}}
\def\serial{\mathbf{ser}}
\def\fbucfra{\mathbf{fbu}}
\def\bducfra{\mathbf{bdu}}
\def\fbducfra{\mathbf{fbdu}}
\def\dcfra{\mathbf{d}}
\def\ucfra{\mathbf{u}}
\def\Id{\mathtt{Id}}
\def\At{\mathbf{At}}
\def\Mona{\mathbf{Mo}}
\def\GFo{\mathbf{GF}}
\def\IS{\mathbf{IS}}
\def\K{\mathbf{K}}
\def\IK{\mathbf{IK}}
\def\S{\mathbf{S}}
\def\Log{\mathtt{Log}}
\def\IPL{\mathbf{IPL}}
\def\L{\mathbf{L}}
\begin{document}
\title{Four intuitionistic modal connectives}
\author{Philippe Balbiani\footnote{Email address: philippe.balbiani@irit.fr.}
\hspace{0.31cm}
\c{C}i\u{g}dem Gencer\footnote{Email addresses: cigdem.gencer@irit.fr.}}
\date{Institut de recherche en informatique de Toulouse
\\
CNRS~--~Toulouse INP~--~Universit\'e de Toulouse
\\
Toulouse, France}
\maketitle
\begin{abstract}
We introduce the syntax and the semantics of intuitionistic modal logics based on a connective $\lozenge$ {\em \`a la}\/ P\v{r}enosil, its dual connective $\square$, a connective $\blacklozenge$ {\em \`a la}\/ Wijesekera and its dual connective $\blacksquare$.
We analyze the modal definability of some elementary classes of frames.
We study the complete axiomatizability of the sets of valid formulas determined by these classes of frames.
We prove the decidability of the minimal intuitionistic modal logic determined by the class of all frames.
\end{abstract}
Intuitionistic modal logics.
Correspondence.
Axiomatization.
Completeness.
Canonical model construction.
Decidability.
Monadic two-variable guarded fragment.
\section{Introduction}
Dealing with intuitionistic modal logics (IMLs), when authors interpret formulas in relational structures, the considered truth conditions are either those of Fischer Servi~\cite{FischerServi:1984}, or those of Wijesekera~\cite{Wijesekera:1990}.
The structures in question are of the form $(W,{\leq},{R},V)$ where $(W,{\leq})$ is a nonempty preorder, $R$ is a binary relation on $W$ and $V$ is a valuation on $W$.
There, the modal connective of necessity is interpreted as follows: $s$ in $W$ satisfies ${\square}A$ when for all $t$ in $W$, if $s{\leq}t$ then for all $u$ in $W$, if $t{R}u$ then $u$ satisfies $A$.
As for the modal connective of possibility, there are two schools.
The first school is the one of Fischer Servi, who says that $s$ in $W$ satisfies ${\lozenge_{\mathtt{FS}}}A$ if there exists $t$ in $W$ such that $s{R}t$ and $t$ satisfies $A$.
The second school is the one of Wijesekera, who says that $s$ in $W$ satisfies ${\lozenge_{\mathtt{W}}}A$ if for all $t$ in $W$, if $s{\leq}t$ then there exists $u$ in $W$ such that $t{R}u$ and $u$ satisfies $A$.
These approaches have given rise to the intuitionistic modal logics $\IK$ and $\WK$.\footnote{They have also given rise to multifarious variants such as the so-called intuitionistic modal logics considered in~\cite{Bierman:dePaiva:2000,Lin:Ma:2019,Marin:et:al:2021} and the so-called constructive modal logics considered in~\cite{Alechina:et:al:2001,Arisaka:et:al:2015,deGroot:et:al:2026}.
See~\cite{Olivetti:2022} and~\cite{IMLA:2017} for recent surveys about them.}
\\
\\
In a structure $(W,{\leq},{R},V)$ as above, $V$ is such that for all $s,t$ in $W$, if $s{\leq}t$ then $t$ satisfies at least all atoms that $s$ satisfies.
With the aim of verifying this Heredity Property at the level of all formulas as well, Fischer Servi restricted the discussion to the class of all forward and backward confluent models~---~those such that ${\geq}{\circ}{R}{\subseteq}{R}{\circ}{\geq}$ and ${R}{\circ}{\leq}{\subseteq}{\leq}{\circ}{R}$~---~whereas Wijesekera had no such obligation.
Although, in the class of all forward and backward confluent models, the definitions of satisfiability considered by Fischer Servi and Wijesekera coincide, the following question must be asked: what is the justification for only considering forward and backward confluent models?
To this question, many authors provide the same answer: in the class of all forward and backward confluent models, the set of all valid formulas is equal to the set of all formulas whose standard translation in a first-order language is intuitionistically valid.\footnote{See~\cite[Chapter~$5$]{Simpson:1994} for detailed explanations about this correspondence.}
\\
\\
Since the above-mentioned concept of intuitionistic validity of first-order formulas is defined with respect to the class of all intuitionistic models with increasing domains, accepting this answer based on the standard translation amounts to considering~---~as multifarious authors have done since Kripke~\cite{Kripke:1965}~---~that Intuitionistic First-Order Logic is naturally defined as the logic determined by all intuitionistic models with increasing domains.
In reality, the scientific literature proposes many variants to the logic determined by all intuitionistic models with increasing domains: the variants determined by metaframes, the variants determined by varying domains, etc~\cite[Chapters~$5$ and~$6$]{Gabbay:et:al:2009}.
Therefore, considering only the logic determined by all intuitionistic models with increasing domains amounts to giving more importance to one variant than to other variants.
And it is not an obvious choice at all, indeed.
%
%
\\
\\
That is why, in opposition to Fischer Servi and most of her followers, we decide not to limit the discussion about intuitionistic modal logics to a particular class of models and to interpret the modal connective of possibility by saying with P\v{r}enosil~\cite{Prenosil:2014} that in a model $(W,{\leq},{R},V)$, $s$ in $W$ satisfies ${\lozenge_{\mathtt{P}}}A$ if there exists $t$ in $W$ such that $s{\geq}t$ and there exists $u$ in $W$ such that $t{R}u$ and $u$ satisfies $A$.
The aforementioned Heredity Property is verified for the IML determined by P\v{r}enosil's truth condition without restricting, as done by Fischer Servi, the discussion to a specific class of models.
That being said, we can already hear some of the criticisms that the experts of IMLs are preparing.
\\
\\
One criticism relates again to the standard translation, seeing that in the class of all models, the IML determined by P\v{r}enosil's truth condition is not equal to the set of all formulas whose standard translation is intuitionistically valid over the class of all intuitionistic models with increasing domains.
Some formulas such as ${\lozenge_{\mathtt{P}}}(A{\rightharpoonup}B){\rightharpoonup}({\square}A{\rightharpoonup}
$\linebreak$
{\lozenge_{\mathtt{P}}}B)$ and $({\lozenge_{\mathtt{P}}}A{\rightharpoonup}{\square}B){\rightharpoonup}{\square}(A{\rightharpoonup}B)$ are no longer valid, whereas their standard translations interpreted over the class of all intuitionistic models with increasing domains are intuitionistically valid.
Well, this criticism can be refuted by stating that today, although the normal modal logic $\K$ does not contain the formulas ${\square}A{\rightharpoonup}{\square}{\square}A$ and ${\lozenge}A{\rightharpoonup}{\square}{\lozenge}A$ expressing the properties that philosophers usually associate with the modal concepts of necessity and possibility, everyone accepts the fact that $\K$ is the minimal normal modal logic~\cite[Chapter~$2$]{Hughes:Cresswell:1996}.\footnote{See~\cite{Ballarin:2023} for historical developments about the modern origins of normal modal logics.}
\\
\\
Another criticism concerns Gentzen-type proof systems and tableaux-based proof me\-thods.
It says that within the context of the semantics of IMLs where the truth condition of the modal connective $\square$ is based on the preorder $\leq$ and the truth condition of the modal connective $\lozenge_{\mathtt{P}}$ is based on the inverse preorder $\geq$, Gentzen-type proof systems and tableaux-based proof methods become less useful, seeing that it is then more difficult to formulate with them arguments showing the decidability of the considered logics.
That is a weird criticism because in the setting of Fischer Servi, it is already well-known that to solve the membership problem is difficult.
Witness, the fact that the decidability of $\IK$ has been proved using a very elaborated argument~\cite{Grefe:1996} and the decidability of $\IS4$~---~the extension of $\IK$ with $\S4$-like axioms such as ${\square}A{\rightharpoonup}A$, $A{\rightharpoonup}{\lozenge_{\mathtt{FS}}}A$, ${\square}A{\rightharpoonup}{\square}{\square}A$ and ${\lozenge_{\mathtt{FS}}}{\lozenge_{\mathtt{FS}}}A{\rightharpoonup}{\lozenge_{\mathtt{FS}}}A$~---~has only been proved recently~\cite{Girlando:et:al:2023}.\footnote{See also~\cite{Dalmonte:et:al:2021,Girlando:et:al:2024} for further methods deciding the membership problem in IMLs.}
\\
\\
Refusing to adopt the specific connective of possibility {\em \`a la}\/ Fischer Servi and having no good reason to boycott the connective of possibility {\em \`a la}\/ Wijesekera, we therefore base the language of our intuitionistic modal logics (IMLs) on the usual intuitionistic connectives $\rightharpoonup$, $\top$, $\bot$, $\vee$ and $\wedge$ together with the connective of possibility {\em \`a la}\/ P\v{r}enosil (denoted $\lozenge$ from now on), the above-mentioned connective of necessity (denoted $\square$ from now on) and the connective of possibility {\em \`a la}\/ Wijesekera (denoted $\blacklozenge$ from now on).
As we will see in Section~\ref{section:relational:semantics}, $\lozenge$ and $\square$ have truth conditions that are dual to each other.
It is therefore natural to add to the language a connective of necessity (denoted $\blacksquare$ from now on) having a truth condition dual to the truth condition of the connective $\blacklozenge$ of possibility.
In order to provide each connector with its dual, we also add to the language the connective $\leftharpoondown$ of dual implication sometimes called co-implication, exclusion, etc~\cite{Rauszer:1980}.\footnote{See~\cite{Gore:Shillito:2020,Wolter:1998a} for further developments concerning the connective of dual implication.}
\\
\\
%
%
In Sections~\ref{section:syntax} and~\ref{section:relational:semantics}, we introduce the syntax and the semantics of intuitionistic modal logics.
In Section~\ref{section:expressivity}, we consider several classes of models with respect to which the connectives $\lozenge$, $\square$, $\blacklozenge$ and $\blacksquare$ are not interdefinable.
In Sections~\ref{section:miscellaneous} and~\ref{section:correspondence:general:situation}, we study the ability of our language to characterize this or that class of models in the sense of correspondence theory.
In Section~\ref{section:hilbert:style:axiomatization}, we axiomatically present several IMLs.
In Sections~\ref{section:filters:ideals:tips:and:clips} and \ref{section:existence:properties}, we prepare the ground for the proofs of their completeness presented in Section~\ref{section:canonical:model:construction}.
In Section~\ref{section:decidability}, using a monadic two-variable guarded fragment of Classical First-Order Logic, we show that the membership problem in the minimal IML is decidable.
\\
\\
For all $n{\in}\N$, $(n)$ denotes $\{a{\in}\N$: $1{\leq}a{\leq}n\}$.
For all sets $\Sigma$, $\wp(\Sigma)$ denotes the {\em powerset of $\Sigma$.}
For all sets $W$, $\Id_{W}$ denotes the {\em identity relation on $W$,} that is to say the binary relation $R$ on $W$ such that for all $s,t{\in}W$, $s{R}t$ if and only if $s{=}t$.
For all sets $W$ and for all binary relations $R,S$ on $W$, ${R}{\circ}{S}$ denotes the {\em composition of $R$ and $S$,} that is to say the binary relation $T$ on $W$ such that for all $s,t{\in}W$, $s{T}t$ if and only if there exists $u{\in}W$ such that $s{R}u$ and $u{S}t$.
For all sets $W$, a {\em preorder on $W$}\/ is a reflexive and transitive binary relation on $W$.
For all sets $W$ and for all preorders $\leq$ on $W$, $\geq$ denotes the preorder on $W$ such that for all $s,t{\in}W$, $s{\geq}t$ if and only if $t{\leq}s$.
For all sets $W$ and for all preorders $\leq$ on $W$, a subset $U$ of $W$ is {\em $\leq$-closed}\/ if for all $s,t{\in}W$, if $s{\in}U$ and $s{\leq}t$ then $t{\in}U$.
For all sets $W$, a {\em partial order on $W$}\/ is an antisymmetric preorder on $W$.
``$\IPL$'' stands for ``Intuitionistic Propositional Logic''.
Missing proofs are left to the reader.
\section{Propositional syntax}\label{section:syntax}
\subsection{Formulas}
Let $\At$ be a countable set (with typical members called {\em atoms}\/ and denoted $p$, $q$, etc).
\\
\\
Let $(p_{i})_{i{\in}\N}$ be an enumeration without repetition of $\At$.
\\
\\
Let $\Fo$ be the set (with typical members called {\em formulas}\/ and denoted $A$, $B$, etc) of finite words over $\At$, the {\em Heyting-Brouwer connectives}\/ ${\rightharpoonup}$ and ${\leftharpoondown}$, the Boolean connectives ${\top}$ and ${\bot}$, the Boolean connectives ${\vee}$ and ${\wedge}$, the {\em modal connectives}\/ ${\lozenge}$, ${\square}$, ${\blacklozenge}$ and $\blacksquare$ and the parentheses $($ and $)$ defined as follows:
\begin{itemize}
%
%
\item $A\ {:=}\ p{\mid}(A{\triangledown_{\mathtt{hbc}}^{2}}A){\mid}{\triangledown_{\mathtt{bc}}^{0}}{\mid}(A{\triangledown_{\mathtt{bc}}^{2}}A){\mid}{\triangledown_{\mathtt{mc}}^{1}}A$,
\end{itemize}
$p$ ranging over $\At$ and $\triangledown_{\mathtt{hbc}}^{2}$, $\triangledown_{\mathtt{bc}}^{0}$, $\triangledown_{\mathtt{bc}}^{2}$ and $\triangledown_{\mathtt{mc}}^{1}$ respectively ranging over $\{{\rightharpoonup},{\leftharpoondown}\}$, $\{{\top},{\bot}\}$, $\{{\vee},{\wedge}\}$ and $\{{\lozenge},{\square},{\blacklozenge},{\blacksquare}\}$.
\\
\\
For all $A{\in}\Fo$, the {\em length of $A$}\/ (denoted ${\parallel}A{\parallel}$) is the number of symbols in $A$.
\\
\\
We follow the standard rules for omission of the parentheses.
\\
\\
For all $A,B{\in}\Fo$, we write respectively ${\rceil}A$, ${\lfloor}A$, $A{\curlyvee}B$ and $A{\curlywedge}B$ as abbreviations instead of $A{\rightharpoonup}{\bot}$, $A{\leftharpoondown}{\top}$, $(A{\rightharpoonup}B){\rightharpoonup}B$ and $(A{\leftharpoondown}B){\leftharpoondown}B$.
\\
\\
For all $n{\in}\N$ and for all $A_{1},\ldots,A_{n}{\in}\Fo$, we write respectively $A_{1}{\wedge}\ldots{\wedge}A_{n}$ and $A_{1}{\vee}
$\linebreak$
\ldots{\vee}A_{n}$ as abbreviations instead of $(A_{1}{\wedge}\ldots(A_{n}{\wedge}{\top})\ldots)$ and $(A_{1}{\vee}\ldots(A_{n}{\vee}{\bot})\ldots)$.
%
%
\\
\\
For all $A{\in}\Fo$, let $A{\mid}^{+}{=}A$, $A{\mid}^{-}{=}{\rceil}A$, $A{\mid}_{+}{=}{\lfloor}A$ and $A{\mid}_{-}{=}A$.
\\
\\
A {\em substitution}\/ is a function $\sigma\ :\ \Fo{\longrightarrow}\Fo$ preserving the above-menti\-oned connectives.
\\
\\
A formula $B$ is an {\em instance of a formula $A$}\/ if there exists a substitution $\sigma\ :\ \Fo{\longrightarrow}\Fo$ such that $B{=}\sigma(A)$.
\subsection{Sets of formulas}
For all sets $\Gamma,\Sigma$ of formulas, let
\begin{itemize}
\item $\Gamma{+}\Sigma{=}\{B{\in}\Fo:$ there exists $n{\in}\N$ and there exists $A_{1},\ldots,A_{n}{\in}\Sigma$ such that $A_{1}{\wedge}\ldots{\wedge}A_{n}{\rightharpoonup}B{\in}\Gamma\}$,
\item $\Gamma{-}\Sigma{=}\{B{\in}\Fo:$ there exists $n{\in}\N$ and there exists $A_{1},\ldots,A_{n}{\in}\Sigma$ such that $A_{1}{\vee}\ldots{\vee}A_{n}{\leftharpoondown}B{\in}\Gamma\}$.
\end{itemize}
For all sets $\Gamma$ of formulas and for all $A{\in}\Fo$, we write $\Gamma{+}A$ instead of $\Gamma{+}\{A\}$ and $\Gamma{-}A$ instead of $\Gamma{-}\{A\}$.
\\
\\
For all sets $\Gamma$ of formulas, let ${\square}\Gamma{=}\{A{\in}\Fo:\ {\square}A{\in}\Gamma\}$ and ${\lozenge}\Gamma{=}\{A{\in}\Fo:\ {\lozenge}A{\in}\Gamma\}$.
\subsection{Signed formulas}
A {\em signed formula}\/ is a couple $(\alpha,A)$ where $\alpha{\in}\{+,-\}$ and $A{\in}\Fo$.
\\
\\
Let $\SFo$ be the set of all signed formulas.
\\
\\
For all $(\alpha,A){\in}\SFo$, the {\em length of $(\alpha,A)$}\/ (denoted ${\parallel}(\alpha,A){\parallel}$) is the number of symbols in $(\alpha,A)$.
\\
\\
Let $\mathtt{tr}:\ \SFo{\longrightarrow}\SFo$ be the function associating to each signed formula the signed formula obtained from it after having exchanged in it $-$ with $+$, $\leftharpoondown$ with $\rightharpoonup$, $\bot$ with $\top$, $\wedge$ with $\vee$, $\square$ with $\lozenge$ and $\blacksquare$ with $\blacklozenge$.
\begin{proposition}\label{Proposition:for:all:formulas:2}
For all $(\alpha,A){\in}\SFo$, $\mathtt{tr}(\mathtt{tr}(\alpha,A)){=}(\alpha,A)$.
\end{proposition}
\begin{proposition}\label{Proposition:for:all:formulas:1}
For all $(\alpha,A){\in}\SFo$, ${\parallel}\mathtt{tr}(\alpha,A){\parallel}{=}{\parallel}(\alpha,A){\parallel}$.
\end{proposition}
%
%
%
%
%
%
%
%
A signed formula $(\beta,B)$ is an {\em instance of a signed formula $(\alpha,A)$}\/ if $\beta{=}\alpha$ and $B$ is an instance of $A$.
\begin{proposition}\label{Proposition:for:all:formulas:if:is:an:instance:of:then:is:an:instance:of}
For all $(\alpha,A),(\beta,B){\in}\SFo$, if $(\beta,B)$ is an instance of $(\alpha,A)$ then $\mathtt{tr}(\beta,B)$ is an instance of $\mathtt{tr}(\alpha,A)$.
\end{proposition}
\subsection{Sets of signed formulas}
For all sets $G$ of signed formulas, let $\mathtt{tr}(G){=}\{\mathtt{tr}(\alpha,A):\ (\alpha,A){\in}G\}$.
\begin{proposition}\label{Proposition:for:all:sets:of:formulas:2}
For all sets $G$ of signed formulas, $\mathtt{tr}(\mathtt{tr}(G)){=}G$.
\end{proposition}
A set $G$ of signed formulas is {\em dual}\/ if $\mathtt{tr}(G){=}G$.
\\
\\
A set $G$ of signed formulas is {\em closed for uniform substitution}\/ if for all $(\alpha,A){\in}\SFo$ and for all substitutions $\sigma$, if $(\alpha,A){\in}G$ then $(\alpha,\sigma(A)){\in}G$.
\\
\\
For all sets $G$ of signed formulas, let $G^{+}{=}\{A{\in}\Fo:\ (+,A){\in}G\}$ and $G^{-}{=}\{A{\in}\Fo:\ (-,A){\in}G\}$.
\subsection{Inference rules}
An {\em inference rule}\/ is a couple of the form $\frac{(\alpha_{1},A_{1})\ \ldots\ (\alpha_{m},A_{m})}{(\beta,B)}$ where $m{\in}\N$, $(\alpha_{1},A_{1}),
$\linebreak$
\ldots,(\alpha_{m},A_{m})$ are signed formulas and $(\beta,B)$ is a signed formula.
\\
\\
A set $G$ of signed formulas is {\em closed under the inference rule $\frac{(\alpha_{1},A_{1})\ \ldots\ (\alpha_{m},A_{m})}{(\beta,B)}$}\/ if for all substitutions $\sigma$, if $(\alpha_{1},\sigma(A_{1})),\ldots,(\alpha_{m},\sigma(A_{m})){\in}G$ then $(\beta,\sigma(B)){\in}G$.
\section{Relational semantics}\label{section:relational:semantics}
\subsection{Frames}
A {\em frame}\/ is a triple $(W,{\leq},{R})$ where $(W,{\leq})$ is a nonempty preorder and ${R}$ is a binary relation on $W$.
%
%
%
%
%
%
\\
\\
A frame $(W,{\leq},{R})$ is {\em serial}\/ if for all $s{\in}W$, there exists $t{\in}W$ such that $s{R}t$.
%
%
\\
\\
A frame $(W,{\leq},{R})$ is {\em reflexive}\/ if for all $s{\in}W$, $s{R}s$.
%
%
\\
\\
A frame $(W,{\leq},{R})$ is {\em symmetric}\/ if for all $s,t{\in}W$, if $s{R}t$ then $t{R}s$.
%
%
\\
\\
A frame $(W,{\leq},{R})$ is {\em transitive}\/ if for all $s,t,u{\in}W$, if $s{R}t$ and $t{R}u$ then $s{R}u$.
%
%
\\
\\
A frame $(W,{\leq},{R})$ is {\em Euclidean}\/ if for all $s,t,u{\in}W$, if $s{R}t$ and $s{R}u$ then $t{R}u$.
%
%
\\
\\
A frame $(W,{\leq},{R})$ is {\em deterministic}\/ if for all $s,t,u{\in}W$, if $s{R}t$ and $s{R}u$ then $t{=}u$.
%
%
%
%
\\
\\
For all frames $(W,{\leq},{R})$, let $(W,{\leq},{R})^{\mathtt{tr}}$ be the frame $(W,{\geq},{R})$.
\\
\\
For all classes ${\mathcal C}$ of frames, let ${\mathcal C}^{\mathtt{tr}}$ be the class of all frames $(W,{\leq},{R})^{\mathtt{tr}}$, $(W,{\leq},{R})$ ranging over ${\mathcal C}$.
\\
\\
A class ${\mathcal C}$ of frames is {\em dual}\/ if ${\mathcal C}^{\mathtt{tr}}{=}{\mathcal C}$.
\\
\\
Let ${\mathcal C}_{\allfra}$ be the class of all frames.
\\
\\
Let ${\mathcal C}_{\serial}$ be the class of all serial frames, ${\mathcal C}_{\reflexive}$ be the class of all reflexive frames, ${\mathcal C}_{\sym}$ be the class of all symmetric frames, ${\mathcal C}_{\transitive}$ be the class of all transitive frames, ${\mathcal C}_{\Euclidean}$ be the class of all Euclidean frames and ${\mathcal C}_{\deterministic}$ be the class of all deterministic frames.
\subsection{Confluences}
A frame $(W,{\leq},{R})$ is {\em forward confluent}\/ if ${\geq}{\circ}{R}{\subseteq}{R}{\circ}{\geq}$.
\\
\\
A frame $(W,{\leq},{R})$ is {\em backward confluent}\/ if ${R}{\circ}{\leq}{\subseteq}{\leq}{\circ}{R}$.
\\
\\
A frame $(W,{\leq},{R})$ is {\em downward confluent}\/ if ${\leq}{\circ}{R}{\subseteq}{R}{\circ}{\leq}$.
\\
\\
A frame $(W,{\leq},{R})$ is {\em upward confluent}\/ if ${R}{\circ}{\geq}{\subseteq}{\geq}{\circ}{R}$.
%
%
%
%
%
%
%
%
%
%
%
%
%
%
\\
\\
Let ${\mathcal C}_{\fcfra}$ be the class of all forward confluent frames, ${\mathcal C}_{\bcfra}$ be the class of all backward confluent frames, ${\mathcal C}_{\dcfra}$ be the class of all downward confluent frames and ${\mathcal C}_{\ucfra}$ be the class of all upward confluent frames.
\\
\\
We write ${\mathcal C}_{\fbcfra}$ to denote the class of all forward and backward confluent frames, ${\mathcal C}_{\fdcfra}$ to denote the class of all forward and downward confluent frames, etc.
\\
\\
The conditions of forward confluence and backward confluence have been considered in~\cite{FischerServi:1984} where they have been called ``connecting properties'' and in~\cite[Chapter~$3$]{Simpson:1994} where they have been called ``$(\mathbf{F}1)$'' and ``$(\mathbf{F}2)$''.
They have also been considered in~\cite{Marin:et:al:2021,Plotkin:Stirling:1986}.
The conditions of downward confluence and upward confluence have been considered in~\cite{Bozic:Dosen:1984} where they have received no specific name.
Other conditions have been considered as well in the literature~\cite{Sotirov:1984}.
%
%
%
%
%
%
\begin{proposition}\label{Proposition:about:fc:bc:dc:uc:preorder:first}
For all frames $(W,{\leq},{R})$,
\begin{itemize}
\item $(W,{\leq},{R})$ is in ${\mathcal C}_{\fcfra}$ if and only if $(W,{\leq},{R})^{\mathtt{tr}}$ is in ${\mathcal C}_{\dcfra}$,
\item $(W,{\leq},{R})$ is in ${\mathcal C}_{\bcfra}$ if and only if $(W,{\leq},{R})^{\mathtt{tr}}$ is in ${\mathcal C}_{\ucfra}$.
\end{itemize}
\end{proposition}
%
%
%
%
%
%
%
%
%
%
%
%
%
%
%
%
%
%
%
%
%
%
%
%
%
%
%
%
%
%
%
%
%
%
%
%
%
%
%
%
%
%
%
%
%
%
%
%
%
%
%
%
%
%
\subsection{Valuations and models}
%
%
%
%
%
%
For all nonempty preorders $(W,{\leq})$, a {\em valuation on $(W,{\leq})$}\/ is a function $V\ :\ \At{\longrightarrow}
$\linebreak$
\wp(W)$ such that for all atoms $p$, $V(p)$ is $\leq$-closed.
%
%
%
%
%
%
%
%
%
%
%
%
\\
\\
A {\em model}\/ is a $4$-tuple $(W,{\leq},{R},V)$ where $(W,{\leq},{R})$ is a frame and $V$ is a valuation on the nonempty preorder $(W,{\leq})$.
\begin{proposition}\label{Proposition:about:complement:and:closed:subsets}
Let $(W,{\leq})$ be a nonempty preorder.
For all subsets $U$ of $W$, $U$ is $\leq$-closed if and only if $W{\setminus}U$ is $\geq$-closed.
\end{proposition}
For all valuations $V$ on the nonempty preorder $(W,{\leq})$, let $V^{\mathtt{tr}}$ be the valuation on the nonempty preorder $(W,{\geq})$ such that for all atoms $p$, $V^{\mathtt{tr}}(p){=}W{\setminus}V(p)$.
\\
\\
%
%
For all models $(W,{\leq},{R},V)$, let $(W,{\leq},{R},V)^{\mathtt{tr}}$ be the model $(W,{\geq},{R},V^{\mathtt{tr}})$.
\subsection{Satisfiability of formulas}
With respect to a model $(W,{\leq},{R},V)$, for all $s{\in}W$ and for all $A{\in}\Fo$, the {\em satisfiability of $A$ at $s$ in $(W,{\leq},{R},V)$}\/ (in symbols $s{\models}A$) is defined as follows:
%
%
\begin{itemize}
\item $s{\models}p$ if and only if $s{\in}V(p)$,
%
%
%
%
\item $s{\models}A{\rightharpoonup}B$ if and only if for all $t{\in}W$, if $s{\leq}t$ then either $t{\not\models}A$, or $t{\models}B$,
%
%
%
%
\item $s{\models}A{\leftharpoondown}B$ if and only if there exists $t{\in}W$ such that $t{\leq}s$, $t{\not\models}A$ and $t{\models}B$,
%
%
%
%
\item $s{\models}{\top}$,
%
%
%
%
\item $s{\not\models}{\bot}$,
%
%
%
%
\item $s{\models}A{\vee}B$ if and only if either $s{\models}A$, or $s{\models}B$,
%
%
%
%
\item $s{\models}A{\wedge}B$ if and only if $s{\models}A$ and $s{\models}B$,
%
%
%
%
\item $s{\models}{\lozenge}A$ if and only if there exists $t{\in}W$ such that $t{\leq}s$ and there exists $u{\in}W$ such that $t{R}u$ and $u{\models}A$,
%
%
%
%
\item $s{\models}{\square}A$ if and only if for all $t{\in}W$, if $s{\leq}t$ then for all $u{\in}W$, if $t{R}u$ then $u{\models}A$,
%
%
%
%
\item $s{\models}{\blacklozenge}A$ if and only if for all $t{\in}W$, if $s{\leq}t$ then there exists $u{\in}W$ such that $t{R}u$ and $u{\models}A$,
%
%
%
%
\item $s{\models}{\blacksquare}A$ if and only if there exists $t{\in}W$ such that $t{\leq}s$ and for all $u{\in}W$, if $t{R}u$ then $u{\models}A$.
%
%
%
%
\end{itemize}
When the model $(W,{\leq},{R},V)$ is not clear from the context, for all $s{\in}W$ and for all $A{\in}\Fo$, we write $(W,{\leq},{R},V),s{\models}A$ instead of $s{\models}A$.
\\
\\
Our definition of the satisfiability of $\lozenge$-formulas is the one that P\v{r}enosil~\cite{Prenosil:2014} and Simpson~\cite[Page~$49$]{Simpson:1994} consider.
Our definition of the satisfiability of $\square$-formulas is the one that Fischer Servi~\cite{FischerServi:1984} and Plotkin and Stirling~\cite{Plotkin:Stirling:1986} consider.
Our definition of the satisfiability of $\blacklozenge$-formulas is the one that Wijesekera~\cite{Wijesekera:1990} considers.
Our definition of the satisfiability of $\blacksquare$-formulas has never been considered before.
\begin{proposition}[Heredity Property]\label{Proposition:let:be:a:model:for:all:formulas:and:for:all:if:and:then}
Let $(W,{\leq},{R},V)$ be a model.
For all $A{\in}\Fo$ and for all $s,t{\in}W$, if $s{\models}A$ and $s{\leq}t$ then $t{\models}A$.\footnote{Our Heredity Property is reminiscent of the Heredity Property considered as well in the literature.
See~\cite[Proposition~$2.1$]{Chagrov:Zakharyaschev:1997}.}
\end{proposition}
\begin{proof}
By induction on $A$.
\medskip
\end{proof}
Concerning $\lozenge$-formulas and $\square$-formulas, Bo\v{z}i\'{c} and Do\v{s}en~\cite{Bozic:Dosen:1984} write
\begin{itemize}
\item $s{\models_{\mathtt{BD}}}{\lozenge}A$ if and only if there exists $t{\in}W$ such that $s{R}t$ and $t{\models_{\mathtt{BD}}}A$,
%
%
%
%
\item $s{\models_{\mathtt{BD}}}{\square}A$ if and only if for all $t{\in}W$, if $s{R}t$ then $t{\models_{\mathtt{BD}}}A$.
%
%
%
%
\end{itemize}
This definition of the satisfiability of $\lozenge$-formulas and $\square$-formulas necessitates to restrict the discussion to ${\mathcal C}_{\fdcfra}$, otherwise the Heredity Property would not hold.
\\
\\
Concerning $\lozenge$-formulas and $\square$-formulas, Fischer Servi~\cite{FischerServi:1984} write
\begin{itemize}
\item $s{\models_{\mathtt{FS}}}{\lozenge}A$ if and only if there exists $t{\in}W$ such that $s{R}t$ and $t{\models_{\mathtt{FS}}}A$,
%
%
%
%
\item $s{\models_{\mathtt{FS}}}{\square}A$ if and only if for all $t{\in}W$, if $s{\leq}t$ then for all $u{\in}W$, if $t{R}u$ then $u{\models_{\mathtt{FS}}}A$.
%
%
%
%
\end{itemize}
This definition of the satisfiability of $\lozenge$-formulas and $\square$-formulas necessitates to restrict the discussion to ${\mathcal C}_{\fcfra}$, otherwise the Heredity Property would not hold.\footnote{Indeed, Fischer Servi restricted the discussion to ${\mathcal C}_{\fbcfra}$, seeing that~---~this being just my interpretation~---~she also wanted $\IS4$~---~the extension of $\IK$ mentioned in the introduction~---~to be complete with respect to a class of reflexive and transitive frames.}
\\
\\
The reader may easily verify that in ${\mathcal C}_{\fdcfra}$, the definition of the satisfiability of $\lozenge$-formulas and $\square$-formulas given by Bo\v{z}i\'{c} and Do\v{s}en, the definition of the satisfiability of $\lozenge$-formulas and $\square$-formulas given by Fischer Servi and our definition of the satisfiability of $\lozenge$-formulas and $\square$-formulas are equivalent.
\begin{proposition}\label{Proposition:let:be:a:model:for:all:formulas:and:for:all:if:and:only:if:for:all:if:then:if:and:only:if:there:exists:such:that:and}
Let $(W,{\leq},{R},V)$ be a model.
For all $A{\in}\Fo$ and for all $s{\in}W$,\footnote{It follows that when ${\leq}{=}{\Id_{W}}$, $\rceil$ becomes like classical negation and $\lfloor$ becomes like classical negation.}
\begin{itemize}
\item $s{\models}{\rceil}A$ if and only if for all $t{\in}W$, if $s{\leq}t$ then $t{\not\models}A$,
%
%
%
%
\item $s{\models}{\lfloor}A$ if and only if there exists $t{\in}W$ such that $t{\leq}s$ and $t{\not\models}A$.
%
%
%
%
\end{itemize}
\end{proposition}
\begin{proposition}\label{Proposition:about:curlyvee:and:curlywedge:when:we:have:identity:instead:of:preorder}
Let $(W,{\leq},{R},V)$ be a model.
For all $A,B{\in}\Fo$ and for all $s{\in}W$,\footnote{It follows that when ${\leq}{=}{\Id_{W}}$, $\curlyvee$ becomes like classical disjunction and $\curlywedge$ becomes like classical conjunction.}
\begin{itemize}
\item $s{\models}A{\curlyvee}B$ if and only if for all $t{\in}W$, if $s{\leq}t$ then either there exists $u{\in}W$ such that $t{\leq}u$, $u{\models}A$ and $u{\not\models}B$, or $t{\models}B$,
%
%
%
%
\item $s{\models}A{\curlywedge}B$ if and only if there exists $t{\in}W$ such that $t{\leq}s$ and for all $u{\in}W$, if $u{\leq}t$ then either $u{\models}A$, or $u{\not\models}B$ and $t{\models}B$.
%
%
%
%
\end{itemize}
\end{proposition}
\subsection{Satisfiability of signed formulas}
With respect to a model $(W,{\leq},{R},V)$, for all $s{\in}W$ and for all $(\alpha,A){\in}\SFo$, the {\em satisfiability of $(\alpha,A)$ at $s$ in $(W,{\leq},{R},V)$}\/ (in symbols $s{\models}(\alpha,A)$) is defined as follows:
\begin{itemize}
\item $s{\models}(+,A)$ if and only if $s{\models}A$,
\item $s{\models}(-,A)$ if and only if $s{\not\models}A$.
\end{itemize}
When the model $(W,{\leq},{R},V)$ is not clear from the context, for all $s{\in}W$ and for all $(\alpha,A){\in}\SFo$, we write $(W,{\leq},{R},V),s{\models}(\alpha,A)$ instead of $s{\models}(\alpha,A)$.
\begin{proposition}\label{Proposition:about:duality:valuation:complement:valuation:transpose:formulas}
For all models $(W,{\leq},{R},V)$, for all $(\alpha,A){\in}\SFo$ and for all $s{\in}W$, $(W,{\leq},{R},V),s{\models}(\alpha,A)$ if and only if $(W,{\leq},{R},V)^{\mathtt{tr}},s{\models}\mathtt{tr}(\alpha,A)$.
\end{proposition}
\begin{proof}
By induction on $A$.
\medskip
\end{proof}
%
%
%
%
A signed formula $(\alpha,A)$ is {\em true in a model $(W,{\leq},{R},V)$}\/ (in symbols $(W,{\leq},{R},V){\models}
$\linebreak$
(\alpha,A)$) if for all $s{\in}W$, $s{\models}(\alpha,A)$.
\begin{proposition}\label{proposition:about:truth:in:a:model:and:signs:of:signed:formulas}
For all models $(W,{\leq},{R},V)$ and for all $(\alpha,A){\in}\SFo$, the following conditions are equivalent: $\mathbf{1)}$~$(W,{\leq},{R},V){\models}(\alpha,A)$; $\mathbf{2)}$~$(W,{\leq},{R},V){\models}(+,A{\mid}^{\alpha})$;
\linebreak$
\mathbf{3)}$~$(W,{\leq},{R},V){\models}(-,A{\mid}_{\alpha})$.
\end{proposition}
\begin{proposition}\label{Proposition:about:duality:valuation:complement:valuation:transpose:formulas:validity:model}
For all models $(W,{\leq},{R},V)$ and for all $(\alpha,A){\in}\SFo$, $(W,{\leq},{R},V){\models}
$\linebreak$
(\alpha,A)$ if and only if $(W,{\leq},{R},V)^{\mathtt{tr}}{\models}\mathtt{tr}(\alpha,A)$.
\end{proposition}
%
%
%
%
%
%
%
%
A signed formula $(\alpha,A)$ is {\em valid in a frame $(W,{\leq},{R})$}\/ (in symbols $(W,{\leq},{R}){\models}(\alpha,A)$) if for all models $(W,{\leq},{R},V)$ based on $(W,{\leq},{R})$, $(W,{\leq},{R},V){\models}(\alpha,A)$.
In that case, $(W,{\leq},{R})$ {\em validates $(\alpha,A)$.}
\begin{proposition}\label{proposition:about:validity:and:signs:of:signed:formulas}
For all frames $(W,{\leq},{R})$ and for all $(\alpha,A){\in}\SFo$, the following conditions are equivalent: $\mathbf{1)}$~$(W,{\leq},{R}){\models}(\alpha,A)$; $\mathbf{2)}$~$(W,{\leq},{R}){\models}(+,A{\mid}^{\alpha})$; $\mathbf{3)}$~$(W,{\leq},{R})
$\linebreak$
{\models}(-,A{\mid}_{\alpha})$.
\end{proposition}
%
%
%
%
%
%
%
%
%
%
\begin{proposition}\label{Proposition:about:duality:complement:valuation:transpose:formulas:validity:frame}
For all frames $(W,{\leq},{R})$ and for all $(\alpha,A){\in}\SFo$, $(W,{\leq},{R}){\models}(\alpha,A)$ if and only if $(W,{\leq},{R})^{\mathtt{tr}}{\models}\mathtt{tr}(\alpha,A)$.
\end{proposition}
%
%
%
%
%
%
%
%
%
%
A signed formula $(\alpha,A)$ is {\em valid on a class ${\mathcal C}$ of frames}\/ (in symbols ${\mathcal C}{\models}(\alpha,A)$) if for all frames $(W,{\leq},{R})$ in ${\mathcal C}$, $(W,{\leq},{R}){\models}(\alpha,A)$.
\begin{proposition}\label{Proposition:about:duality:complement:valuation:transpose:formulas:validity:class:of:frames}
For all classes ${\mathcal C}$ of frames and for all $(\alpha,A){\in}\SFo$, the following conditions are equivalent: $\mathbf{1)}$~${\mathcal C}{\models}(\alpha,A)$; $\mathbf{2)}$~${\mathcal C}{\models}(+,A{\mid}^{\alpha})$; $\mathbf{3)}$~${\mathcal C}{\models}(-,A{\mid}_{\alpha})$.
\end{proposition}
\begin{proposition}\label{Proposition:about:duality:complement:valuation:transpose:formulas:validity:class:of:frames:bis}
For all classes ${\mathcal C}$ of frames and for all $(\alpha,A){\in}\SFo$, ${\mathcal C}{\models}(\alpha,A)$ if and only if ${\mathcal C}^{\mathtt{tr}}{\models}\mathtt{tr}(\alpha,A)$.
\end{proposition}
%
%
%
%
%
%
%
%
\begin{proposition}\label{about:sat:formulas:forward:downward:frames}
For all atoms $p,q$,
\begin{itemize}
\item ${\mathcal C}_{\fcfra}{\models}(+,{\square}(p{\rightharpoonup}q){\rightharpoonup}({\lozenge}p{\rightharpoonup}{\blacklozenge}q))$,
\item ${\mathcal C}_{\fcfra}{\models}(-,({\lozenge}p{\rightharpoonup}{\blacklozenge}q){\leftharpoondown}{\square}(p{\rightharpoonup}q))$,
\item ${\mathcal C}_{\dcfra}{\models}(+,({\square}p{\leftharpoondown}{\blacksquare}q){\rightharpoonup}{\lozenge}(p{\leftharpoondown}q))$,
\item ${\mathcal C}_{\dcfra}{\models}(-,{\lozenge}(p{\leftharpoondown}q){\leftharpoondown}({\square}p{\leftharpoondown}{\blacksquare}q))$.
\end{itemize}
\end{proposition}
For all classes ${\mathcal C}$ of frames, let $\Log({\mathcal C}){=}\{(\alpha,A){\in}\SFo:\ {\mathcal C}{\models}(\alpha,A)\}$.
\begin{proposition}
For all classes ${\mathcal C}$ of frames, if ${\mathcal C}$ is dual then $\Log({\mathcal C})$ is dual.
\end{proposition}
%
%
%
%
%
%
%
%
\begin{proposition}\footnote{Here, $\FIK$ is the set of all $\{{\leftharpoondown},{\blacklozenge},{\blacksquare}\}$-free formulas valid in ${\mathcal C}_{\fcfra}$~\cite{Balbiani:et:al:2024a}, $\IK$ is the set of all $\{{\leftharpoondown},{\blacklozenge},{\blacksquare}\}$-free formulas valid in ${\mathcal C}_{\fbcfra}$~\cite{FischerServi:1984}, $\LIK$ is the set of all $\{{\leftharpoondown},{\blacklozenge},{\blacksquare}\}$-free formulas valid in ${\mathcal C}_{\fdcfra}$~\cite{Balbiani:et:al:2024b} and $\WK$ is the set of all $\{{\leftharpoondown},{\lozenge},{\blacksquare}\}$-free formulas valid in ${\mathcal C}_{\allfra}$~\cite{Wijesekera:1990}.}
\begin{itemize}
\item $\FIK{=}\{A{\in}\Fo:$ $A$ is $\{{\leftharpoondown},{\blacklozenge},{\blacksquare}\}$-free$\}{\cap}\Log({\mathcal C}_{\fcfra})$,
\item $\IK{=}\{A{\in}\Fo:$ $A$ is $\{{\leftharpoondown},{\blacklozenge},{\blacksquare}\}$-free$\}{\cap}\Log({\mathcal C}_{\fbcfra})$,
\item $\LIK{=}\{A{\in}\Fo:$ $A$ is $\{{\leftharpoondown},{\blacklozenge},{\blacksquare}\}$-free$\}{\cap}\Log({\mathcal C}_{\fdcfra})$,
\item $\WK{=}\{A{\in}\Fo:$ $A$ is $\{{\leftharpoondown},{\lozenge},{\blacksquare}\}$-free$\}{\cap}\Log({\mathcal C}_{\allfra})$.
\end{itemize}
\end{proposition}
%
%
%
%
%
%
%
%
%
%
\section{Eliminability}\label{section:expressivity}
%
%
A modal connective $\triangledown_{\mathtt{mc}}^{1}$ is {\em eliminable with respect to a class ${\mathcal C}$ of frames}\/ if for all atoms $p$, there exists a $\triangledown_{\mathtt{mc}}^{1}$-free formula $A$ such that for all frames $(W,{\leq},{R})$ in ${\mathcal C}$, for all models $(W,{\leq},{R},V)$ based on $(W,{\leq},{R})$ and for all $s{\in}W$, $s{\models}{\triangledown_{\mathtt{mc}}^{1}}p$ if and only if $s{\models}A$.
\\
\\
It follows from Proposition~\ref{interdefinable:1} that $\lozenge$, $\blacklozenge$, $\square$ and $\blacksquare$ are eliminable with respect to ${\mathcal C}_{\fdcfra}$.
\begin{proposition}\label{interdefinable:1}
\begin{enumerate}
\item $\lozenge$ is eliminable with respect to ${\mathcal C}_{\fcfra}$,
\item $\square$ is eliminable with respect to ${\mathcal C}_{\dcfra}$,
\item $\blacklozenge$ is eliminable with respect to ${\mathcal C}_{\fcfra}$,
\item $\blacksquare$ is eliminable with respect to ${\mathcal C}_{\dcfra}$.
\end{enumerate}
\end{proposition}
\begin{proof}
$\mathbf{1)}$~It suffices to prove that for all atoms $p$, ${\mathcal C}_{\fcfra}{\models}(+,{\lozenge}p{\rightharpoonup}{\blacklozenge}p)$ and ${\mathcal C}_{\allfra}{\models}(-,{\lozenge}p
$\linebreak$
{\leftharpoondown}{\blacklozenge}p)$.
\\
\\
$\mathbf{2)}$~Dual to the proof of Item~$\mathbf{1)}$.
\\
\\
$\mathbf{3)}$~It suffices to prove that for all atoms $p$, ${\mathcal C}_{\allfra}{\models}(+,{\blacklozenge}p{\rightharpoonup}{\lozenge}p)$ and ${\mathcal C}_{\fcfra}{\models}(-,{\blacklozenge}p{\leftharpoondown}{\lozenge}p)$.
\\
\\
$\mathbf{4)}$~Dual to the proof of Item~$\mathbf{3)}$.
\medskip
\end{proof}
%
%
%
%
\begin{proposition}\label{interdefinable:3}
\begin{enumerate}
\item $\lozenge$ is not eliminable with respect to ${\mathcal C}_{\bducfra}$,
\item $\square$ is not eliminable with respect to ${\mathcal C}_{\fbucfra}$,
\item $\blacklozenge$ is not eliminable with respect to ${\mathcal C}_{\bducfra}$,
\item $\blacksquare$ is not eliminable with respect to ${\mathcal C}_{\fbucfra}$.
\end{enumerate}
\end{proposition}
\begin{proof}
$\mathbf{1)}$~Let $(W^{\prime},{\leq^{\prime}},{R^{\prime}})$ and $(W^{\prime\prime},{\leq^{\prime\prime}},{R^{\prime\prime}})$ be the frames in ${\mathcal C}_{\bducfra}$ defined by $W^{\prime}{=}\{a,b,c,d,e\}$, $a{\leq^{\prime}}c$, $b{\leq^{\prime}}d$, $c{\leq^{\prime}}e$, $a{R^{\prime}}b$, $a{R^{\prime}}d$, $c{R^{\prime}}d$, $W^{\prime\prime}{=}\{a,b,c,d,e\}$, $a{\leq^{\prime\prime}}c$, $b{\leq^{\prime\prime}}d$, $c{\leq^{\prime\prime}}e$, $a{R^{\prime\prime}}b$ and $c{R^{\prime\prime}}d$.
Let $V^{\prime}$ be a valuation on $(W^{\prime},{\leq^{\prime}},{R^{\prime}})$ such that $V^{\prime}(p){=}
$\linebreak$
\{d\}$ and $V^{\prime\prime}$ be a valuation on $(W^{\prime\prime},{\leq^{\prime\prime}},{R^{\prime\prime}})$ such that $V^{\prime\prime}(p){=}\{d\}$.
The reader may easily verify that $(W^{\prime},{\leq^{\prime}},{R^{\prime}},V^{\prime}),a{\models}{\lozenge}p$ and $(W^{\prime\prime},{\leq^{\prime\prime}},{R^{\prime\prime}},V^{\prime\prime}),a{\not\models}{\lozenge}p$.
In other respect, by induction on the $\lozenge$-free formula $A$, the reader may easily verify that for all $s{\in}\{a,b,c,d,e\}$, $(W^{\prime},{\leq^{\prime}},{R^{\prime}},V^{\prime}),s{\models}A$ if and only if $(W^{\prime\prime},{\leq^{\prime\prime}},{R^{\prime\prime}},V^{\prime\prime}),s{\models}A$.
\\
\\
$\mathbf{2)}$~Dual to the proof of Item~$\mathbf{1)}$.
%
%
\\
\\
$\mathbf{3)}$~Let $(W^{\prime},{\leq^{\prime}},{R^{\prime}})$ and $(W^{\prime\prime},{\leq^{\prime\prime}},{R^{\prime\prime}})$ be the frames in ${\mathcal C}_{\bducfra}$ defined by $W^{\prime}{=}\{a,b,c,
$\linebreak$
d,e,f\}$, $a{\leq^{\prime}}c$, $b{\leq^{\prime}}d$, $c{\leq^{\prime}}e$, $d{\leq^{\prime}}f$, $a{R^{\prime}}b$, $a{R^{\prime}}d$, $a{R^{\prime}}f$, $c{R^{\prime}}d$, $c{R^{\prime}}f$, $e{R^{\prime}}f$, $W^{\prime\prime}{=}\{a,b,c,
$\linebreak$
d,e,f\}$, $a{\leq^{\prime\prime}}c$, $b{\leq^{\prime\prime}}d$, $c{\leq^{\prime\prime}}e$, $d{\leq^{\prime\prime}}f$, $a{R^{\prime\prime}}b$, $a{R^{\prime\prime}}d$, $a{R^{\prime\prime}}f$, $c{R^{\prime\prime}}d$ and $e{R^{\prime\prime}}f$.
Let $V^{\prime}$ be a valuation on $(W^{\prime},{\leq^{\prime}},{R^{\prime}})$ such that $V^{\prime}(p){=}\{f\}$ and $V^{\prime\prime}$ be a valuation on $(W^{\prime\prime},{\leq^{\prime\prime}},
$\linebreak$
{R^{\prime\prime}})$ such that $V^{\prime\prime}(p){=}\{f\}$.
The reader may easily verify that $(W^{\prime},{\leq^{\prime}},{R^{\prime}},V^{\prime}),a{\models}{\blacklozenge}p$ and $(W^{\prime\prime},{\leq^{\prime\prime}},{R^{\prime\prime}},V^{\prime\prime}),a{\not\models}{\blacklozenge}p$.
In other respect, by induction on the $\blacklozenge$-free formula $A$, the reader may easily verify that for all $s{\in}\{a,b,c,d,e,f\}$, $(W^{\prime},{\leq^{\prime}},{R^{\prime}},V^{\prime}),s{\models}A$ if and only if $(W^{\prime\prime},{\leq^{\prime\prime}},{R^{\prime\prime}},V^{\prime\prime}),s{\models}A$.
\\
\\
$\mathbf{4)}$~Dual to the proof of Item~$\mathbf{3)}$.
%
%
\medskip
\end{proof}
\section{Correspondence: some examples}\label{section:miscellaneous}
\begin{table}[ht]
\begin{center}
\begin{tabular}{|c|c|c|}
\hline
%
%
%
%
%
%
%
%
$1$&$\forall s,\exists t,s{\geq}{\circ}{R}t$&$(+,{\lozenge}{\top})$
\\
\hline
$2$&$\forall s,\exists t,s{\leq}{\circ}{R}t$&$(-,{\square}{\bot})$
\\
\hline
$3$&$\forall s,s{\geq}{\circ}{R}{\circ}{\geq}s$&$(+,p{\rightharpoonup}{\lozenge}p)$
\\
\hline
$4$&$\forall s,s{\leq}{\circ}{R}{\circ}{\leq}s$&$(-,p{\leftharpoondown}{\square}p)$
\\
\hline
$5$&$\forall s,s{R}{\circ}{\geq}s$&$(+,p{\rightharpoonup}{\blacklozenge}p)$
\\
\hline
$6$&$\forall s,s{R}{\circ}{\leq}s$&$(-,p{\leftharpoondown}{\blacksquare}p)$
\\
\hline
$7$&$\forall s,t,s{R}t\Rightarrow t{\geq}{\circ}{R}{\circ}{\geq}s$&$(+,p{\rightharpoonup}{\square}{\lozenge}p)$
\\
\hline
$8$&$\forall s,t,s{R}t\Rightarrow t{\leq}{\circ}{R}{\circ}{\leq}s$&$(-,p{\leftharpoondown}{\lozenge}{\square}p)$
\\
\hline
$9$&$\forall s,t,s{R}{\circ}{\leq}t\Rightarrow t{R}{\circ}{\geq}s$&$(+,p{\rightharpoonup}{\square}{\blacklozenge}p)$
\\
\hline
$10$&$\forall s,t,s{R}{\circ}{\geq}t\Rightarrow t{R}{\circ}{\leq}s$&$(-,p{\leftharpoondown}{\lozenge}{\blacksquare}p)$
\\
\hline
$11$&$\forall s,\exists t,t{\leq}s\&\forall u,t{R}u\Rightarrow u{\geq}{\circ}{R}{\circ}{\geq}s$&$(+,p{\rightharpoonup}{\blacksquare}{\lozenge}p)$
\\
\hline
$12$&$\forall s,\exists t,t{\geq}s\&\forall u,t{R}u\Rightarrow u{\leq}{\circ}{R}{\circ}{\leq}s$&$(-,p{\leftharpoondown}{\blacklozenge}{\square}p)$
\\
\hline
$13$&$\forall s,\exists t,t{\leq}s\&\forall u,t{R}{\circ}{\leq}u\Rightarrow u{R}{\circ}{\geq}s$&$(+,p{\rightharpoonup}{\blacksquare}{\blacklozenge}p)$
\\
\hline
$14$&$\forall s,\exists t,t{\geq}s\&\forall u,t{R}{\circ}{\geq}u\Rightarrow u{R}{\circ}{\leq}s$&$(-,p{\leftharpoondown}{\blacklozenge}{\blacksquare}p)$
\\
\hline
\end{tabular}
\vspace{+0.25cm}
\caption{}\label{table:modal:definability:results:a}
\end{center}
\end{table}
\vspace{-0.50cm}
\subsection{Definability}
A frame $(W,{\leq},{R})$ {\em respects an inference rule $\frac{(\alpha_{1},A_{1}),\ldots,(\alpha_{m},A_{m})}{(\beta,B)}$}\/ if and only if if for all $i{\in}(m)$, $(W,{\leq},{R}){\models}(\alpha_{i},A_{i})$ then $(W,{\leq},{R}){\models}(\beta,B)$.
\\
\\
A class ${\mathcal C}$ of frames {\em corresponds to an inference rule $\frac{(\alpha_{1},A_{1}),\ldots,(\alpha_{m},A_{m})}{(\beta,B)}$ with respect to a class ${\mathcal D}$ of frames}\/ if for all frames $(W,{\leq},{R})$, if $(W,{\leq},{R})$ is in ${\mathcal D}$ then $(W,{\leq},{R})$ respects $\frac{(\alpha_{1},A_{1}),\ldots,(\alpha_{m},A_{m})}{(\beta,B)}$ if and only if $(W,{\leq},{R})$ is in ${\mathcal C}$.
\\
\\
A class ${\mathcal C}$ of frames is {\em modally definable with respect to a class ${\mathcal D}$ of frames}\/ if there exists an inference rule $\frac{(\alpha_{1},A_{1}),\ldots,(\alpha_{m},A_{m})}{(\beta,B)}$ such that ${\mathcal C}$ corresponds to $\frac{(\alpha_{1},A_{1}),\ldots,(\alpha_{m},A_{m})}{(\beta,B)}$ with respect to ${\mathcal D}$.
In that case, $\frac{(\alpha_{1},A_{1}),\ldots,(\alpha_{m},A_{m})}{(\beta,B)}$ constitutes a {\em modal definition of ${\mathcal C}$ with respect to ${\mathcal D}$.}
\\
\\
A class ${\mathcal C}$ of frames is {\em positively definable with respect to a class ${\mathcal D}$ of frames}\/ if there exists $(\beta,B){\in}\SFo$ such that ${\mathcal C}$ corresponds to $\frac{\cdot}{(\beta,B)}$ with respect ${\mathcal D}$.\footnote{Obviously, if ${\mathcal C}$ is positively definable with respect to ${\mathcal D}$ then ${\mathcal C}$ is modally definable with respect to ${\mathcal D}$.}
In that case, $(\beta,B)$ constitutes a {\em positive definition of ${\mathcal C}$ with respect to ${\mathcal D}$.}
\\
\\
A class ${\mathcal C}$ of frames is {\em negatively definable with respect to a class ${\mathcal D}$ of frames}\/ if there exists $m{\in}\N$ and there exists $(\alpha_{1},A_{1}),\ldots,(\alpha_{m},A_{m}){\in}\SFo$ such that either ${\mathcal C}$ corresponds to $\frac{(\alpha_{1},A_{1}),\ldots,(\alpha_{m},A_{m})}{(-,\top)}$ with respect to ${\mathcal D}$, or ${\mathcal C}$ corresponds to $\frac{(\alpha_{1},A_{1}),\ldots,(\alpha_{m},A_{m})}{(+,\bot)}$ with respect to ${\mathcal D}$.\footnote{Obviously, if ${\mathcal C}$ is negatively definable with respect to ${\mathcal D}$ then ${\mathcal C}$ is modally definable with respect to ${\mathcal D}$.}
In that case, $(\alpha_{1},A_{1}),\ldots,(\alpha_{m},A_{m})$ constitute a {\em negative definition of ${\mathcal C}$ with respect to ${\mathcal D}$.}
\subsection{About positive definability}
In Proposition~\ref{Proposition:about:correpsondence:table:a}, we show that the elementary conditions considered in Table~\ref{table:modal:definability:results:a} determine positively definable classes of frames.
\begin{proposition}\label{Proposition:about:correpsondence:table:a}
For each condition $\varphi$ considered in Table~\ref{table:modal:definability:results:a} and for each corresponding signed formula $(\beta,B)$, with respect to ${\mathcal C}_{\allfra}$, $(\beta,B)$ is a positive definition of the class of all frames where $\varphi$ hold.
\end{proposition}
\begin{proof}\footnote{We only consider the case of the $13$th condition.}
Let $(W,{\leq},{R})$ be a frame.
\\
\\
Suppose there exists $s{\in}W$ such that for all $t{\in}W$, if $t{\leq}s$ then there exists $u{\in}W$ such that $t{R}{\circ}{\leq}u$ and not $u{R}{\circ}{\geq}s$.
In that case, $V$ being a valuation on $(W,{\leq})$ such that $V(p){=}\{v{\in}W:\ s{\leq}v\}$,
the reader may easily verify that $s{\models}p$ and $s{\not\models}{\blacksquare}{\blacklozenge}p$.
Thus, $(W,{\leq},{R}){\not\models}p{\rightharpoonup}{\blacksquare}{\blacklozenge}p$.
\medskip
\end{proof}
Conditions~$1$ and~$2$ in Table~\ref{table:modal:definability:results:a} are somehow related to the condition of seriality.
\begin{proposition}\label{proposition:positive:definition:seriality}
With respect to ${\mathcal C}_{\allfra}$, $(+,{\blacklozenge}{\top})$ and $(-,{\blacksquare}{\bot})$ are positive definitions of ${\mathcal C}_{\serial}$.
\end{proposition}
In Proposition~\ref{Proposition:forward:confluence:is:modally:definable}, we show that ${\mathcal C}_{\fcfra}$ and ${\mathcal C}_{\dcfra}$ are positively definable.
\begin{proposition}\label{Proposition:forward:confluence:is:modally:definable}
With respect to ${\mathcal C}_{\allfra}$,
\begin{enumerate}
\item $(+,{\square}(p{\rightharpoonup}q){\rightharpoonup}({\lozenge}p{\rightharpoonup}{\blacklozenge}q))$, $(-,({\lozenge}p{\rightharpoonup}{\blacklozenge}q){\leftharpoondown}{\square}(p{\rightharpoonup}q))$, $(+,{\lozenge}p{\rightharpoonup}{\blacklozenge}p)$ and $(-,
$\linebreak$
{\blacklozenge}p{\leftharpoondown}{\lozenge}p)$ are positive definitions of ${\mathcal C}_{\fcfra}$,
\item $(+,({\square}p{\leftharpoondown}{\blacksquare}q){\rightharpoonup}{\lozenge}(p{\leftharpoondown}q))$, $(-,{\lozenge}(p{\leftharpoondown}q){\leftharpoondown}({\square}p{\leftharpoondown}{\blacksquare}q))$, $(+,{\blacksquare}p{\rightharpoonup}{\square}p)$ and $(-,
$\linebreak$
{\square}p{\leftharpoondown}{\blacksquare}p)$ are positive definitions of ${\mathcal C}_{\dcfra}$.
\end{enumerate}
\end{proposition}
\begin{proof}
$\mathbf{1)}$\footnote{We only consider the case of $(+,{\lozenge}p{\rightharpoonup}{\blacklozenge}p)$.}
Let $(W,{\leq},{R})$ be a frame.
\\
\\
%
%
%
%
%
%
%
Suppose $(W,{\leq},{R})$ is not forward confluent, i.e. there exists $s,t{\in}W$ such that $s{\geq}{\circ}{R}t$ and not $s{R}{\circ}{\geq}t$.
In that case, $V$ being a valuation on $(W,{\leq})$ such that $V(p){=}\{u{\in}W:\ t{\leq}u\}$,
the reader may easily verify that $s{\models}{\lozenge}p$ and $s{\not\models}{\blacklozenge}p$.
Thus, $(W,{\leq},{R}){\not\models}{\lozenge}p{\rightharpoonup}{\blacklozenge}p$.
\\
\\
$\mathbf{2)}$ Dual to the proof of Item~$\mathbf{1)}$.
\medskip
\end{proof}
%
%
%
%
%
%
%
%
%
%
%
%
%
%
%
%
%
%
\subsection{About negative definability}
Regarding negative definability, there is nothing more to say than what we have already said.
\begin{proposition}
For all classes ${\mathcal C},{\mathcal D}$ of frames, the following conditions are equivalent:
\begin{enumerate}
\item ${\mathcal C}$ is negatively definable with respect to ${\mathcal D}$,
\item ${\mathcal D}{\setminus}{\mathcal C}$ is positively definable with respect to ${\mathcal D}$.
\end{enumerate}
\end{proposition}
\begin{proof}
$\mathbf{1){\Rightarrow}2)}$:
Suppose ${\mathcal C}$ is negatively definable with respect to ${\mathcal D}$.
Hence, there exists $m{\in}\N$ and there exists $(\alpha_{1},A_{1}),\ldots,(\alpha_{m},A_{m}){\in}\SFo$ such that for all frames $(W,{\leq},{R})$, if $(W,{\leq},{R})$ is in ${\mathcal D}$ then there exists $i{\in}(m)$ such that $(W,{\leq},{R}){\not\models}(\alpha_{i},A_{i})$ if and only if $(W,{\leq},{R})$ is in ${\mathcal C}$.
Thus, for all frames $(W,{\leq},{R})$, if $(W,{\leq},{R})$ is in ${\mathcal D}$ then for all $i{\in}(m)$, $(W,{\leq},{R}){\models}(\alpha_{i},A_{i})$ if and only if $(W,{\leq},{R})$ is in ${\mathcal D}{\setminus}{\mathcal C}$.
Consequently, by Proposition~\ref{proposition:about:validity:and:signs:of:signed:formulas}, ${\mathcal D}{\setminus}{\mathcal C}$ corresponds to $\frac{\cdot}{(+,A_{1}{\mid}^{\alpha_{1}}{\wedge}\ldots{\wedge}A_{m}{\mid}^{\alpha_{m}})}$ with respect to ${\mathcal D}$ and ${\mathcal D}{\setminus}{\mathcal C}$ corresponds to $\frac{\cdot}{(-,A_{1}{\mid}_{\alpha_{1}}{\vee}\ldots{\vee}A_{m}{\mid}_{\alpha_{m}})}$ with respect to ${\mathcal D}$.
\\
\\
$\mathbf{2){\Rightarrow}1)}$:
Suppose ${\mathcal D}{\setminus}{\mathcal C}$ is positively definable with respect to ${\mathcal D}$.
Hence, there exists $(\beta,B){\in}\SFo$ such that for all frames $(W,{\leq},{R})$, if $(W,{\leq},{R})$ is in ${\mathcal D}$ then $(W,{\leq},{R}){\models}
$\linebreak$
(\beta,B)$ if and only if $(W,{\leq},{R})$ is in ${\mathcal D}{\setminus}{\mathcal C}$.
Thus, for all frames $(W,{\leq},{R})$, if $(W,{\leq},{R})$ is in ${\mathcal D}$ then $(W,{\leq},{R}){\not\models}(\beta,B)$ if and only if $(W,{\leq},{R})$ is in ${\mathcal C}$.
Consequently, ${\mathcal C}$ corresponds to $\frac{(\beta,B)}{(+,{\bot})}$ with respect ${\mathcal D}$ and ${\mathcal C}$ corresponds to $\frac{(\beta,B)}{(-,{\top})}$ with respect ${\mathcal D}$.
%
%
\medskip
\end{proof}
\subsection{Some modally undefinable classes of frames}
Conditions~$3$--$6$ in Table~\ref{table:modal:definability:results:a} are somehow related to the condition of reflexivity.
\begin{proposition}\label{Proposition:there:is:no:formula:such:that:for:all:frames:is:reflexive:if:and:only:if:and:ref}
With respect to ${\mathcal C}_{\fbducfra}$, ${\mathcal C}_{\reflexive}$ is not modally definable.
\end{proposition}
\begin{proof}
For the sake of the contradiction, suppose there exists an inference rule
\linebreak$
\frac{(\alpha_{1},A_{1}),\ldots,(\alpha_{m},A_{m})}{(\beta,B)}$ such that $(\ast)$~for all frames $(W,{\leq},{R})$ in ${\mathcal C}_{\fbducfra}$, $(W,{\leq},{R})$ is reflexive if and only if $(W,{\leq},{R})$ respects $\frac{(\alpha_{1},A_{1}),\ldots,(\alpha_{m},A_{m})}{(\beta,B)}$.
Let $(W^{\prime},{\leq^{\prime}},{R^{\prime}})$ and $(W^{\prime\prime},{\leq^{\prime\prime}},{R^{\prime\prime}})$ be the frames in ${\mathcal C}_{\fbducfra}$ such that $W^{\prime}{=}\{a,b\}$, $a{\leq^{\prime}}b$, $b{\leq^{\prime}}a$, $a{R^{\prime}}a$, $b{R^{\prime}}b$, $W^{\prime\prime}{=}\{a,b\}$, $a{\leq}^{\prime\prime}b$, $b{\leq}^{\prime\prime}a$, $a{R^{\prime\prime}}b$ and $b{R^{\prime\prime}}a$.
Obviously, $(W^{\prime},{\leq^{\prime}},{R^{\prime}})$ is reflexive and $(W^{\prime\prime},{\leq^{\prime\prime}},{R^{\prime\prime}})$ is not reflexive.
Hence, by $(\ast)$, $(W^{\prime},{\leq^{\prime}},{R^{\prime}})$ respects $\frac{(\alpha_{1},A_{1}),\ldots,(\alpha_{m},A_{m})}{(\beta,B)}$ and $(W^{\prime\prime},{\leq^{\prime\prime}},{R^{\prime\prime}})$ does not respect $\frac{(\alpha_{1},A_{1}),\ldots,(\alpha_{m},A_{m})}{(\beta,B)}$.
Thus, for all $i{\in}(m)$, $(W^{\prime\prime},{\leq^{\prime\prime}},
$\linebreak$
{R^{\prime\prime}}){\models}(\alpha_{i},A_{i})$ and $(W^{\prime\prime},{\leq^{\prime\prime}},{R^{\prime\prime}}){\not\models}(\beta,B)$.
\begin{claim}
For all $(\gamma,C){\in}\SFo$, $(W^{\prime},{\leq^{\prime}},{R^{\prime}}){\models}(\gamma,C)$ if and only if $(W^{\prime\prime},{\leq^{\prime\prime}},{R^{\prime\prime}}){\models}(\gamma,
$\linebreak$
C)$.
\end{claim}
\begin{proof}
By induction on $C$.
\medskip
\end{proof}
Since for all $i{\in}(m)$, $(W^{\prime\prime},{\leq^{\prime\prime}},{R^{\prime\prime}}){\models}(\alpha_{i},A_{i})$ and $(W^{\prime\prime},{\leq^{\prime\prime}},{R^{\prime\prime}}){\not\models}(\beta,B)$, therefore for all $i{\in}(m)$, $(W^{\prime},{\leq^{\prime}},{R^{\prime}}){\models}(\alpha_{i},A_{i})$ and $(W^{\prime},{\leq^{\prime}},{R^{\prime}}){\not\models}(\beta,B)$.
Consequently, $(W^{\prime},
$\linebreak$
{\leq^{\prime}},{R^{\prime}})$ does not respect $\frac{(\alpha_{1},A_{1}),\ldots,(\alpha_{m},A_{m})}{(\beta,B)}$: a contradiction.
\medskip
\end{proof}
Conditions~$7$--$14$ in Table~\ref{table:modal:definability:results:a} are somehow related to the condition of symmetry.
\begin{proposition}\label{Proposition:there:is:no:formula:such:that:for:all:frames:is:reflexive:if:and:only:if:and:sym}
With respect to ${\mathcal C}_{\fbducfra}$, ${\mathcal C}_{\sym}$ is not modally definable.
\end{proposition}
\begin{proof}
Similar to the proof of Proposition~\ref{Proposition:there:is:no:formula:such:that:for:all:frames:is:reflexive:if:and:only:if:and:ref}, this time considering the frames $(W^{\prime},{\leq^{\prime}},
$\linebreak$
{R^{\prime}})$ and $(W^{\prime\prime},{\leq^{\prime\prime}},{R^{\prime\prime}})$ in ${\mathcal C}_{\fbducfra}$ such that $W^{\prime}{=}\{a,b,c,d,e,f\}$, $a{\leq^{\prime}}c$, $b{\leq^{\prime}}d$, $c{\leq^{\prime}}e$, $d{\leq^{\prime}}f$, $a{R^{\prime}}b$, $a{R^{\prime}}d$, $b{R^{\prime}}a$, $b{R^{\prime}}c$, $c{R^{\prime}}b$, $c{R^{\prime}}d$, $c{R^{\prime}}f$, $d{R^{\prime}}a$, $d{R^{\prime}}c$, $d{R^{\prime}}e$, $e{R^{\prime}}d$, $e{R^{\prime}}f$, $f{R^{\prime}}c$, $f{R^{\prime}}e$, $W^{\prime\prime}{=}\{a,b,c,d,e,f\}$, $a{\leq^{\prime\prime}}c$, $b{\leq^{\prime\prime}}d$, $c{\leq^{\prime\prime}}e$, $d{\leq^{\prime\prime}}f$, $a{R^{\prime\prime}}b$, $a{R^{\prime\prime}}d$, $b{R^{\prime\prime}}a$, $b{R^{\prime\prime}}c$, $c{R^{\prime\prime}}b$, $c{R^{\prime\prime}}d$, $c{R^{\prime\prime}}f$, $d{R^{\prime\prime}}a$, $d{R^{\prime\prime}}e$, $e{R^{\prime\prime}}d$, $e{R^{\prime\prime}}f$, $f{R^{\prime\prime}}c$ and $f{R^{\prime\prime}}e$.
\medskip
\end{proof}
Indeed, similar results hold for ${\mathcal C}_{\transitive}$, ${\mathcal C}_{\Euclidean}$ and ${\mathcal C}_{\deterministic}$.
\begin{proposition}\label{Proposition:there:is:no:formula:such:that:for:all:frames:is:transitive:if:and:only:if:and:ref}
With respect to ${\mathcal C}_{\fbducfra}$, neither ${\mathcal C}_{\transitive}$, nor ${\mathcal C}_{\Euclidean}$, nor ${\mathcal C}_{\deterministic}$ are modally definable.
\end{proposition}
\begin{proof}
%
%
``${\mathcal C}_{\transitive}$'':
Similar to the proof of Proposition~\ref{Proposition:there:is:no:formula:such:that:for:all:frames:is:reflexive:if:and:only:if:and:ref}.
\\
\\
%
%
``${\mathcal C}_{\Euclidean}$'':
Similar to the proof of Proposition~\ref{Proposition:there:is:no:formula:such:that:for:all:frames:is:reflexive:if:and:only:if:and:ref}.
\\
\\
%
%
``${\mathcal C}_{\deterministic}$'':
Similar to the proof of Proposition~\ref{Proposition:there:is:no:formula:such:that:for:all:frames:is:reflexive:if:and:only:if:and:ref}, this time considering the frames $(W^{\prime},
$\linebreak$
{\leq^{\prime}},{R^{\prime}})$ and $(W^{\prime\prime},{\leq^{\prime\prime}},{R^{\prime\prime}})$ in ${\mathcal C}_{\fbducfra}$ such that $W^{\prime}{=}\{a,b\}$, $a{\leq^{\prime}}b$, $b{\leq^{\prime}}a$, $a{R^{\prime}}a$, $b{R^{\prime}}b$, $W^{\prime\prime}{=}\{a,b\}$, $a{\leq}^{\prime\prime}b$, $b{\leq}^{\prime\prime}a$, $a{R^{\prime\prime}}a$, $a{R^{\prime\prime}}b$, $b{R^{\prime\prime}}a$ and $b{R^{\prime\prime}}b$.
\medskip
\end{proof}
Now, about ${\mathcal C}_{\bcfra}$ and ${\mathcal C}_{\ucfra}$.
\begin{proposition}\label{Proposition:backward:confluence:is:not:modally:definable}
With respect to ${\mathcal C}_{\fdcfra}$, neither ${\mathcal C}_{\bcfra}$, nor ${\mathcal C}_{\ucfra}$ are modally definable.
\end{proposition}
\begin{proof}
Similar to the proof of Proposition~\ref{Proposition:there:is:no:formula:such:that:for:all:frames:is:reflexive:if:and:only:if:and:ref}, this time considering the frames $(W^{\prime},{\leq^{\prime}},
$\linebreak$
{R^{\prime}})$ and $(W^{\prime\prime},{\leq^{\prime\prime}},{R^{\prime\prime}})$ in ${\mathcal C}_{\fdcfra}$ such that $W^{\prime}{=}\{a,b,c,d\}$, $b{\leq^{\prime}}c$, $c{\leq^{\prime}}d$, $a{R^{\prime}}b$, $a{R^{\prime}}c$, $a{R^{\prime}}d$, $W^{\prime\prime}{=}\{a,b,c,d\}$, $b{\leq}^{\prime\prime}c$, $c{\leq}^{\prime\prime}d$, $a{R^{\prime\prime}}b$ and $a{R^{\prime\prime}}d$.
\medskip
\end{proof}
\section{Correspondence: the general situation}\label{section:correspondence:general:situation}
\subsection{First-order formulas}
Let a countable set (with typical members called {\em individual variables}\/ and denoted $x$, $y$, etc) be given.
\\
\\
Let $\Fo^{\prime}$ be the set (with typical members called {\em first-order formulas}\/ and denoted $A$, $B$, etc) defined by
\begin{itemize}
\item $A::=R(x,y){\mid}x{=}y{\mid}x{\leq}y{\mid}{\bot}{\mid}{\neg}A{\mid}(A{\vee}A){\mid}{\forall x}A$,
\end{itemize}
$x,y$ ranging over the set of all individual variables.
\\
\\
For all $A{\in}\Fo^{\prime}$, the {\em length of $A$}\/ (denoted ${\parallel}A{\parallel}$) is the number of symbols in $A$.
\\
\\
We follow the standard rules for omission of the parentheses.
\\
\\
We define the other Boolean constructs as usual.
\\
\\
For all individual variables $x$ and for all $A{\in}\Fo^{\prime}$, we write ${\exists x}A$ as an abbreviation instead of ${\neg}{\forall x}{\neg}A$.
\\
\\
The first-order formulas of the form $R(x,y)$, $x{=}y$ and $x{\leq}y$ are called {\em atomic formulas.}
\\
\\
For all $A{\in}\Fo^{\prime}$, let $fiv(A)$ be the set of all individual variables freely occurring in $A$.
\\
\\
A first-order formula $A$ is called a {\em sentence}\/ if $fiv(A){=}\emptyset$.
\\
\\
Sometimes, we write $\bar{x}$ for a list $x_{1},\ldots,x_{m}$ of pairwise distinct individual variables.
We leave it to the context to determine the length of such list.
\\
\\
When $\bar{x}$ is a list of pairwise distinct individual variables, we write $A(\bar{x})$ to denote a first-order formula $A$ whose free individual variables belongs to $\bar{x}$.
\\
\\
For all individual variables $x$, let $(\cdot)^{x}\ :\ \Fo^{\prime}{\longrightarrow}\Fo^{\prime}$ be the function defined as follows:
\begin{itemize}
\item $(R(y,z))^{x}$ is $R(y,z)$,
\item $(y{=}z)^{x}$ is $y{=}z$,
\item $(y{\leq}z)^{x}$ is $y{\leq}z$,
\item $({\bot})^{x}$ is $\bot$,
\item $({\neg}C)^{x}$ is ${\neg}(C)^{x}$,
\item $(C{\vee}D)^{x}$ is $(C)^{x}{\vee}(D)^{x}$,
\item $({\forall y}C)^{x}$ is ${\forall y}(R(x,y){\rightarrow}(C)^{x})$.
\end{itemize}
{\bf From now on in this article, when we write $(C)^{x}$, we assume that $x$ does not occur in $C$.}
\begin{proposition}
For all individual variables $x$ and for all $C{\in}\Fo^{\prime}$, $fiv((C)^{x}){\subseteq}\{x\}{\cup}
$\linebreak$
fiv(C)$.
\end{proposition}
\begin{proof}
By induction on $C$.
\medskip
\end{proof}
\subsection{Satisfiability of first-order formulas}
The {\em satisfiability of a first-order formula $A(\bar{x})$ in a frame $(W,{\leq},{R})$ with respect to a list $\bar{s}$ of elements in $W$}\/ (in symbols $(W,{\leq},{R}){\models}A(\bar{x})\ \lbrack\bar{s}\rbrack$) is defined as follows:\footnote{When we write $(W,{\leq},{R}){\models}A(\bar{x})\ \lbrack\bar{s}\rbrack$, we mean that the elements in $\bar{s}$ are the current values of the individual variables in $\bar{x}$.}
%
%
\begin{itemize}
\item $(W,{\leq},{R}){\models}R(x_{i},x_{j})\ \lbrack\bar{s}\rbrack$ if and only if $s_{i}{R}s_{j}$,
\item $(W,{\leq},{R}){\models}x_{i}{=}x_{j}\ \lbrack\bar{s}\rbrack$ if and only if $s_{i}{=}s_{j}$,
\item $(W,{\leq},{R}){\models}x_{i}{\leq}x_{j}\ \lbrack\bar{s}\rbrack$ if and only if $s_{i}{\leq}s_{j}$,
\item $(W,{\leq},{R}){\not\models}{\bot}\ \lbrack\bar{s}\rbrack$,
\item $(W,{\leq},{R}){\models}{\neg}A\ \lbrack\bar{s}\rbrack$ if and only if $(W,{\leq},{R}){\not\models}A\ \lbrack\bar{s}\rbrack$,
\item $(W,{\leq},{R}){\models}A{\vee}B\ \lbrack\bar{s}\rbrack$ if and only if either $(W,{\leq},{R}){\models}A\ \lbrack\bar{s}\rbrack$, or $(W,{\leq},{R}){\models}B
$\linebreak$
\lbrack\bar{s}\rbrack$,
\item $(W,{\leq},{R}){\models}{\forall x}A(\bar{x},x)\ \lbrack\bar{s}\rbrack$ if and only if for all $s{\in}W$, $(W,{\leq},{R}){\models}A(\bar{x},x)\ \lbrack\bar{s},s\rbrack$.
\end{itemize}
When $(W,{\leq},{R}){\models}A(\bar{x})\ \lbrack\bar{s}\rbrack$, we also say that $A(\bar{x})$ {\em holds in $(W,{\leq},{R})$ when $\bar{x}$ is interpreted by $\bar{s}$.}
\\
\\
A first-order formula $A(\bar{x})$ is {\em valid in a frame $(W,{\leq},{R})$}\/ (in symbols $(W,{\leq},{R}){\models}
$\linebreak$
A(\bar{x})$) if $A(\bar{x})$ is satisfied in $(W,{\leq},{R})$ with respect to all lists $\bar{s}$ of elements in $W$.
\\
\\
A first-order formula $A$ is {\em valid in a class ${\mathcal C}$ of frames}\/ (in symbols ${\mathcal C}{\models}A$) if $A$ is valid in all frames in ${\mathcal C}$.
\\
\\
A class ${\mathcal C}$ of frames is {\em elementary}\/ if there exists a sentence $C$ such that for all frames $(W,{\leq},{R})$, $(W,{\leq},{R}){\models}C$ if and only if $(W,{\leq},{R})$ is in ${\mathcal C}$.
\subsection{Relativized reducts}
A frame $(W^{\prime},{\leq^{\prime}},{R^{\prime}})$ is the {\em relativized reduct of a frame $(W,{\leq},{R})$}\/ if there exists $s{\in}W$ such that $(W^{\prime},{\leq^{\prime}},{R^{\prime}})$ is the restriction of $(W,{\leq},{R})$ to the set of all $t{\in}W$ such that $s{R}t$.
In that case, $(W^{\prime},{\leq^{\prime}},{R^{\prime}})$ is the {\em relativized reduct of $(W,{\leq},{R})$ with respect to $s$.}
\begin{proposition}\label{Proposition:condition:suffisante:pour:existence:reduction}
For all frames $(W,{\leq},{R})$ and for all $s{\in}W$, there exists a relativized reduct of $(W,{\leq},{R})$ with respect to $s$ if and only if $(W,{\leq},{R}){\models}{\exists y}R(x,y)\ \lbrack s\rbrack$.
\end{proposition}
The following result~---~a classic result from model theory~\cite[Theorem $5.1.1$]{Hodges:1993}~---~will be useful for proving the undecidability of the modal definability problem.
\begin{proposition}[Relativization Property]\label{Syntactic:Restriction:Proposition}
Let $(W,{\leq},{R}),(W^{\prime},{\leq^{\prime}},{R^{\prime}})$ be frames and $s{\in}W$.
If $(W^{\prime},{\leq^{\prime}},{R^{\prime}})$ is the relativized reduct of $(W,{\leq},{R})$ with respect to $s$ then for all first-order formulas $C(\bar{y})$ and for all lists $\bar{t}$ of elements in $W^{\prime}$, if $\bar{y}$ and $\bar{t}$ have the same length then $(C(\bar{y}))^{x}$ holds in $(W,{\leq},{R})$ when $x$ is interpreted by $s$ and $\bar{y}$ is interpreted by $\bar{t}$ if and only if $C(\bar{y})$ holds in $(W^{\prime},{\leq^{\prime}},{R^{\prime}})$ when $\bar{y}$ is interpreted by $\bar{t}$, i.e. $(W,{\leq},{R}){\models}(C(\bar{y}))^{x}\ \lbrack s,\bar{t}\rbrack$ if and only if $(W^{\prime},{\leq^{\prime}},{R^{\prime}}){\models}C(\bar{y})\ \lbrack\bar{t}\rbrack$.\footnote{When we write $(W,{\leq},{R}){\models}(C(\bar{y}))^{x}\ \lbrack s,\bar{t}\rbrack$ and $(W^{\prime},{\leq^{\prime}},{R^{\prime}}){\models}C(\bar{y})\ \lbrack\bar{t}\rbrack$, we mean that $s$ is the current value of $x$ and the elements in $\bar{t}$ are the current values of the individual variables in $\bar{y}$.}
\end{proposition}
\begin{proof}
By induction on $C$.
\medskip
\end{proof}
\subsection{About definability with respect to ${\mathcal C}_{\allfra}$}
We are now ready to prove the undecidability of the following decision problems:
\begin{description}
\item[$\mathbf{MD}$:] determine whether a given elementary class of frames is modally definable with respect to ${\mathcal C}_{\allfra}$,
\item[$\mathbf{PD}$:] determine whether a given elementary class of frames is positively definable with respect to ${\mathcal C}_{\allfra}$,
\item[$\mathbf{ND}$:] determine whether a given elementary class of frames is negatively definable with respect to ${\mathcal C}_{\allfra}$.
\end{description}
\begin{proposition}\label{proposition:modal:definability:is:undecidable}
The following decision problems are undecidable: $\mathbf{MD}$, $\mathbf{PD}$ and $\mathbf{ND}$.
\end{proposition}
\begin{proof}\footnote{This proof is an adaptation of the proof developed in~\cite{Balbiani:Tinchev:2017} within the framework of modal logics.
We only consider the case of $\mathbf{MD}$.}
As is well-known, the problem of determining the validity of sentences in ${\mathcal C}_{\allfra}$ is undecidable~\cite{Kalmar:1937}.
Hence, it suffices to reduce the problem of determining the validity of sentences in ${\mathcal C}_{\allfra}$ to $\mathbf{MD}$.
Let $B_{1}(y){=}{\forall z}({\neg}R(y,z){\wedge}{\neg}R(z,y))$ and $B_{2}{=}{\forall y^{\prime}}{\forall y^{\prime\prime}}(B_{1}(y^{\prime}){\wedge}B_{1}(y^{\prime\prime}){\rightarrow}y^{\prime}{=}y^{\prime\prime})$.
Suppose there exists a sentence $C$ such that either ${\mathcal C}_{\allfra}{\models}C$ and the class of frames determined by the sentence ${\exists x}({\exists y}R(x,y){\wedge}
$\linebreak$
{\neg}(C)^{x}){\wedge}{\exists y}B_{1}(y){\wedge}B_{2}$ is not positively definable with respect to ${\mathcal C}_{\allfra}$, or ${\mathcal C}_{\allfra}{\not\models}C$ and the class of frames determined by the sentence ${\exists x}({\exists y}R(x,y){\wedge}{\neg}(C)^{x}){\wedge}{\exists y}B_{1}(y){\wedge}B_{2}$ is modally definable with respect to ${\mathcal C}_{\allfra}$.
\\
\\
In the former case, neither $(+,\bot)$, nor $(-,\top)$ are positive definitions of the class of frames determined by ${\exists x}({\exists y}R(x,y){\wedge}{\neg}(C)^{x}){\wedge}{\exists y}B_{1}(y){\wedge}B_{2}$ with respect to ${\mathcal C}_{\allfra}$.
Thus, there exists a frame $(W,{\leq},{R})$ such that $(W,{\leq},{R}){\models}{\exists x}({\exists y}R(x,y){\wedge}{\neg}(C)^{x}){\wedge}
$\linebreak$
{\exists y}B_{1}(y){\wedge}B_{2}$.
Consequently, there exists $s{\in}W$ such that $(W,{\leq},{R}){\models}{\exists y}R(x,y)\ \lbrack s\rbrack$ and $(W,{\leq},{R}){\not\models}(C)^{x}\ \lbrack s\rbrack$.
Hence, by Proposition~\ref{Proposition:condition:suffisante:pour:existence:reduction}, there exists a relativized reduct $(W^{\prime},{\leq^{\prime}},{R^{\prime}})$ of $(W,{\leq},{R})$ with respect to $s$.
Since $(W,{\leq},{R}){\not\models}(C)^{x}\ \lbrack s\rbrack$, therefore by Proposition~\ref{Syntactic:Restriction:Proposition}, $(W^{\prime},{\leq^{\prime}},{R^{\prime}}){\not\models}C$.
Thus, ${\mathcal C}_{\allfra}\not\models C$: a contradiction.
\\
\\
In the latter case, there exists a frame $(W,{\leq},{R})$ such that $(W,{\leq},{R}){\not\models}C$ and there exists a modal definition $\frac{(\alpha_{1},A_{1})\ \ldots\ (\alpha_{m},A_{m})}{(\beta,B)}$ of the class of frames determined by
\linebreak$
{\exists x}({\exists y}R(x,y){\wedge}{\neg}(C)^{x}){\wedge}{\exists y}B_{1}(y){\wedge}B_{2}$ with respect to ${\mathcal C}_{\allfra}$.
Let $(W^{\prime},{\leq^{\prime}},{R^{\prime}}),(W^{\prime\prime},
$\linebreak$
{\leq^{\prime\prime}},{R^{\prime\prime}})$ be the frames defined by $W^{\prime}{=}W\cup\{s^{\prime},t^{\prime}\}$, ${\leq^{\prime}}{=}{\leq}{\cup}\{(s^{\prime},s^{\prime}),(t^{\prime},t^{\prime})\}$, $R^{\prime}{=}R{\cup}
$\linebreak$
(\{s^{\prime}\}\times W)$, $W^{\prime\prime}{=}W\cup\{s^{\prime},t^{\prime},t^{\prime\prime}\}$, ${\leq^{\prime\prime}}{=}{\leq}{\cup}\{(s^{\prime},s^{\prime}),(t^{\prime},t^{\prime}),(t^{\prime\prime},t^{\prime\prime})\}$ and $R^{\prime\prime}{=}R{\cup}
$\linebreak$
(\{s^{\prime}\}\times W)$.
The reader may easily verify that $(W,{\leq},{R})$ is the relativized reduct of $(W^{\prime},{\leq^{\prime}},{R^{\prime}})$ with respect to $s^{\prime}$.
Consequently, by Proposition~\ref{Proposition:condition:suffisante:pour:existence:reduction}, $(W^{\prime},{\leq^{\prime}},{R^{\prime}}){\models}
$\linebreak$
{\exists y}R(x,y)\ \lbrack s^{\prime}\rbrack$.
Moreover, $(W^{\prime},{\leq^{\prime}},{R^{\prime}}){\models}{\exists y}B_{1}(y)$, $(W^{\prime},{\leq^{\prime}},{R^{\prime}}){\models}B_{2}$ and $(W^{\prime\prime},{\leq^{\prime\prime}},
$\linebreak$
{R^{\prime\prime}}){\not\models}B_{2}$.
Hence, $(W^{\prime\prime},{\leq^{\prime\prime}},{R^{\prime\prime}}){\not\models}{\exists x}({\exists y}R(x,y){\wedge}{\neg}(C)^{x}){\wedge}{\exists y}B_{1}(y){\wedge}B_{2}$.
%
%
\begin{claim}
For all $(\gamma,C){\in}\SFo$, $(W^{\prime},{\leq^{\prime}},{R^{\prime}}){\models}(\gamma,C)$ if and only if $(W^{\prime\prime},{\leq^{\prime\prime}},{R^{\prime\prime}}){\models}(\gamma,
$\linebreak$
C)$.
\end{claim}
\begin{proof}
By induction on $C$.
\medskip
\end{proof}
%
%
Since $(W^{\prime\prime},{\leq^{\prime\prime}},{R^{\prime\prime}}){\not\models}{\exists x}({\exists y}R(x,y){\wedge}{\neg}(C)^{x}){\wedge}{\exists y}B_{1}(y){\wedge}B_{2}$ and $\frac{(\alpha_{1},A_{1})\ \ldots\ (\alpha_{m},A_{m})}{(\beta,B)}$ is a modal definition of the class of frames determined by ${\exists x}({\exists y}R(x,y){\wedge}{\neg}(C)^{x}){\wedge}
$\linebreak$
{\exists y}B_{1}(y){\wedge}B_{2}$ with respect to ${\mathcal C}_{\allfra}$, therefore $(W^{\prime},{\leq^{\prime}},{R^{\prime}}){\not\models}{\exists x}({\exists y}R(x,y){\wedge}
$\linebreak$
{\neg}(C)^{x}){\wedge}{\exists y}B_{1}(y){\wedge}B_{2}$.
Since $(W^{\prime},{\leq^{\prime}},{R^{\prime}}){\models}{\exists y}B_{1}(y)$ and $(W^{\prime},{\leq^{\prime}},{R^{\prime}}){\models}B_{2}$, therefore $(W^{\prime},{\leq^{\prime}},{R^{\prime}}){\not\models}{\exists x}({\exists y}R(x,y){\wedge}{\neg}(C)^{x})$.
Thus, either $(W^{\prime},{\leq^{\prime}},{R^{\prime}}){\not\models}{\exists y}R(x,
$\linebreak$
y)\ \lbrack s^{\prime}\rbrack$, or $(W^{\prime},{\leq^{\prime}},{R^{\prime}}){\models}(C)^{x}\ \lbrack s^{\prime}\rbrack$.
Since $(W^{\prime},{\leq^{\prime}},{R^{\prime}}){\models}{\exists y}R(x,y)\ \lbrack s^{\prime}\rbrack$, therefore $(W^{\prime},{\leq^{\prime}},{R^{\prime}}){\models}(C)^{x}\ \lbrack s^{\prime}\rbrack$.
Since $(W,{\leq},{R})$ is the relativized reduct of $(W^{\prime},{\leq^{\prime}},{R^{\prime}})$ with respect to $s^{\prime}$, therefore by Proposition~\ref{Syntactic:Restriction:Proposition}, $(W,{\leq},{R}){\models}C$: a contradiction.
\medskip
\end{proof}
%
%
%
%
%
%
%
%
%
%
%
%
%
%
%
%
\begin{table}[ht]
\begin{center}
\begin{tabular}{|c|c|c|c|}
\hline
$(\Rule1^{+})$&$\frac{(+,p),(+,p{\rightharpoonup}q)}{(+,q)}$&$(\Rule1^{-})$&$\frac{(-,p),(-,p{\leftharpoondown}q)}{(-,q)}$
\\
\hline
$(\Rule2^{+})$&$\frac{(-,p{\leftharpoondown}q)}{(+,q{\rightharpoonup}p)}$&$(\Rule2^{-})$&$\frac{(+,p{\rightharpoonup}q)}{(-,q{\leftharpoondown}p)}$
\\
\hline
$(\Rule3^{+})$&$\frac{(+,p{\rightharpoonup}q)}{(+,{\lozenge}p{\rightharpoonup}{\lozenge}q)}$&$(\Rule3^{-})$&$\frac{(-,p{\leftharpoondown}q)}{(-,{\square}p{\leftharpoondown}{\square}q)}$
\\
\hline
$(\Rule4^{+})$&$\frac{(+,p)}{(+,{\square}p)}$&$(\Rule4^{-})$&$\frac{(-,p)}{(-,{\lozenge}p)}$
\\
\hline
$(\Rule5^{+})$&$\frac{(+,p{\rightharpoonup}q)}{(+,{\blacklozenge}p{\rightharpoonup}{\blacklozenge}q)}$&$(\Rule5^{-})$&$\frac{(-,p{\leftharpoondown}q)}{(-,{\blacksquare}p{\leftharpoondown}{\blacksquare}q)}$
\\
\hline
$(\Rule6^{+})$&$\frac{(+,p)}{(+,{\blacksquare}p)}$&$(\Rule6^{-})$&$\frac{(-,p)}{(-,{\blacklozenge}p)}$
\\
\hline
$(\Rule7^{+})$&$\frac{(+,p{\vee}q)}{(+,{\blacksquare}p{\curlyvee}{\blacklozenge}q)}$&$(\Rule7^{-})$&$\frac{(-,p{\wedge}q)}{(-,{\blacklozenge}p{\curlywedge}{\blacksquare}q)}$
\\
\hline
$(\Rule8^{+})$&$\frac{(+,{\lozenge}p{\rightharpoonup}q{\vee}{\square}r)}{(+,{\lozenge}p{\rightharpoonup}q{\vee}{\lozenge}(p{\wedge}r))}$&$(\Rule8^{-})$&$\frac{(-,{\square}p{\leftharpoondown}q{\wedge}{\lozenge}r)}{(-,{\square}p{\leftharpoondown}q{\wedge}{\square}(p{\vee}r))}$
\\
\hline
%
%
%
%
%
%
%
%
%
%
%
%
%
%
%
%
%
%
%
%
%
%
%
%
%
%
%
%
%
%
%
%
%
%
%
%
%
%
%
%
%
%
%
%
%
%
%
%
%
%
\end{tabular}
\vspace{+0.25cm}
\caption{}\label{table:rules}
\end{center}
\end{table}
\vspace{-0.50cm}
\section{Hilbert-style axiomatization}\label{section:hilbert:style:axiomatization}
An {\em intuitionistic modal logic (IML)}\/ is a set $\L$ of signed formulas closed for uniform substitution, closed under the inference rules described in Table~\ref{table:rules} and containing the signed formulas called {\em axioms}\/ and described in Table~\ref{table:axioms}.\footnote{The reason why we have separated the condition of closure under uniform substitution and the condition of closure under the inference rules described in Table~\ref{table:rules} is simply the following: uniform substitution is not an inference rule.}
\\
\\
The presentation of IMLs as sets of signed formulas is reminiscent of the distinction done by Rauszer~\cite{Rauszer:1980} between on one hand, ``provable formulas'' and ``theorems'' and on the other hand, ``rejected formulas'' and ``antitheorems''.
\\
\\
All the inference rules described in Table~\ref{table:rules} but inference rule $(\Rule7^{+})$ and its dual $(\Rule7^{-})$ and all the axioms described in Table~\ref{table:axioms} but axioms $(\Formula11^{+})$ and $(\Formula12^{+})$ and their duals $(\Formula11^{-})$ and $(\Formula12^{-})$ have been already considered in the above-mentioned literature about intuitionistic modal logics.
In particular, inference rule $(\Rule8^{+})$ and its dual $(\Rule8^{-})$ have been considered in~\cite{Prenosil:2014}.
See also~\cite{Balbiani:Gencer:preliminary:draft:SL}.
\begin{table}[ht]
\begin{center}
\begin{tabular}{|c|c|}
\hline
$(\Formula1^{+})$&$(+,(p{\rightharpoonup}(q{\rightharpoonup}r)){\rightharpoonup}((p{\rightharpoonup}q){\rightharpoonup}(p{\rightharpoonup}r)))$
\\
\hline
$(\Formula2^{+})$&$(+,p{\rightharpoonup}(q{\rightharpoonup}p))$
\\
\hline
$(\Formula3^{+})$&$(+,p{\rightharpoonup}(q{\rightharpoonup}p{\wedge}q)))$, $(+,p{\wedge}q{\rightharpoonup}p)$ and $(+,p{\wedge}q{\rightharpoonup}q)$
\\
\hline
$(\Formula4^{+})$&$(+,(p{\rightharpoonup}r){\rightharpoonup}((q{\rightharpoonup}r){\rightharpoonup}(p{\vee}q{\rightharpoonup}r)))$
\\
\hline
$(\Formula5^{+})$&$(+,{\top})$
\\
\hline
$(\Formula6^{+})$&$(+,(p{\rightharpoonup}q){\wedge}r{\rightharpoonup}(p{\leftharpoondown}r){\vee}q)$
\\
\hline
$(\Formula7^{+})$&$(+,{\rceil}{\lozenge}{\bot})$
\\
\hline
$(\Formula8^{+})$&$(+,{\square}(p{\rightharpoonup}q){\rightharpoonup}({\square}p{\rightharpoonup}{\square}q))$
\\
\hline
$(\Formula9^{+})$&$(+,{\square}p{\rightharpoonup}{\blacksquare}p)$
\\
\hline
$(\Formula10^{+})$&$(+,{\square}(p{\rightharpoonup}q){\rightharpoonup}({\blacklozenge}p{\rightharpoonup}{\blacklozenge}q))$
\\
\hline
$(\Formula11^{+})$&$(+,{\blacklozenge}(p{\vee}q){\rightharpoonup}{\lozenge}p{\curlyvee}{\blacklozenge}q)$
\\
\hline
$(\Formula12^{+})$&$(+,{\square}p{\rightharpoonup}{\blacksquare}(q{\leftharpoondown}p){\curlyvee}{\blacklozenge}q)$
\\
\hline
$(\Formula1^{-})$&$(-,(p{\leftharpoondown}(q{\leftharpoondown}r)){\leftharpoondown}((p{\leftharpoondown}q){\leftharpoondown}(p{\leftharpoondown}r)))$
\\
\hline
$(\Formula2^{-})$&$(-,p{\leftharpoondown}(q{\leftharpoondown}p)$
\\
\hline
$(\Formula3^{-})$&$(-,p{\leftharpoondown}(q{\leftharpoondown}p{\vee}q)$, $(-,p{\vee}q{\leftharpoondown}p)$ and $(-,p{\vee}q{\leftharpoondown}q)$
\\
\hline
$(\Formula4^{-})$&$(-,(p{\leftharpoondown}r){\leftharpoondown}((q{\leftharpoondown}r){\leftharpoondown}(p{\wedge}q{\leftharpoondown}r)))$
\\
\hline
$(\Formula5^{-})$&$(-,{\bot})$
\\
\hline
$(\Formula6^{-})$&$(-,(p{\leftharpoondown}q){\vee}r{\leftharpoondown}(p{\rightharpoonup}r){\wedge}q)$
\\
\hline
$(\Formula7^{-})$&$(-,{\lfloor}{\square}{\top})$
\\
\hline
$(\Formula8^{-})$&$(-,{\lozenge}(p{\leftharpoondown}q){\leftharpoondown}({\lozenge}p{\leftharpoondown}{\lozenge}q))$
\\
\hline
$(\Formula9^{-})$&$(-,{\lozenge}p{\leftharpoondown}{\blacklozenge}p)$
\\
\hline
$(\Formula10^{-})$&$(-,{\lozenge}(p{\leftharpoondown}q){\leftharpoondown}({\blacksquare}p{\leftharpoondown}{\blacksquare}q))$
\\
\hline
$(\Formula11^{-})$&$(-,{\blacksquare}(p{\wedge}q){\leftharpoondown}{\square}p{\curlywedge}{\blacksquare}q)$
\\
\hline
$(\Formula12^{-})$&$(-,{\lozenge}p{\leftharpoondown}{\blacklozenge}(q{\rightharpoonup}p){\curlywedge}{\blacksquare}q)$
\\
\hline
\end{tabular}
\vspace{+0.25cm}
\caption{}\label{table:axioms}
\end{center}
\end{table}
\vspace{-0.50cm}
\begin{proposition}\label{proposition:soundness}
For all classes ${\mathcal C}$ of frames, $\Log({\mathcal C})$ is an IML.
\end{proposition}
\begin{proof}\footnote{We only consider the case of $(\Rule8^{+})$ and $(\Formula12^{+})$.}
%
%
%
%
%
%
%
For the sake of the contradiction, suppose there exists $A,B,C{\in}\Fo$ and there exists a frame $(W,{\leq},{R})$ such that $(W,{\leq},{R}){\models}{\lozenge}A{\rightharpoonup}B{\vee}{\square}C$ and $(W,{\leq},{R}){\not\models}{\lozenge}A{\rightharpoonup}B
$\linebreak$
{\vee}{\lozenge}(A{\wedge}C)$.
Hence, there exists a model $(W,{\leq},{R},V)$ based on $(W,{\leq},{R})$ and there exists $s{\in}W$ such that $s{\models}{\lozenge}A$, $s{\not\models}B$ and $s{\not\models}{\lozenge}(A{\wedge}C)$.
Thus, there exists $t{\in}W$ such that $t{\leq}s$ and there exists $u{\in}W$ such that $t{R}u$ and $u{\models}A$.
Consequently, $t{\models}{\lozenge}A$.
Moreover, since $s{\not\models}B$, therefore $t{\not\models}B$.
Since $(W,{\leq},{R}){\models}{\lozenge}A{\rightharpoonup}B{\vee}{\square}C$, therefore $t{\models}{\square}C$.
Since $t{R}u$, therefore $u{\models}C$.
Since $t{\leq}s$, $t{R}u$ and $u{\models}A$, therefore $s{\models}{\lozenge}(A{\wedge}C)$: a contradiction.
\\
\\
%
%
For the sake of the contradiction, suppose there exists $A,B{\in}\Fo$ and there exists a frame $(W,{\leq},{R})$ such that $(W,{\leq},{R}){\not\models}{\square}A{\rightharpoonup}{\blacksquare}(B{\leftharpoondown}A){\curlyvee}{\blacklozenge}B)$.
Hence, there exists a model $(W,{\leq},{R},V)$ based on $(W,{\leq},{R})$ and there exists $s{\in}W$ such that $s{\models}{\square}A$, $s{\models}{\blacksquare}(B{\leftharpoondown}A){\rightharpoonup}{\blacklozenge}B$ and $s{\not\models}{\blacklozenge}B$.
Thus, there exists $t{\in}W$ such that $s{\leq}t$ and $\mathbf{(\ast)}$~for all $u{\in}W$, if $t{R}u$ then $u{\not\models}B$.
Consequently, $t{\not\models}{\blacklozenge}B$.
Since $s{\models}{\blacksquare}(B{\leftharpoondown}A){\rightharpoonup}{\blacklozenge}B$ and $s{\leq}t$, therefore $t{\not\models}{\blacksquare}(B{\leftharpoondown}A)$.
Hence, there exists $v{\in}W$ such that $t{R}v$ and $v{\not\models}B{\leftharpoondown}A$.
Since~$\mathbf{(\ast)}$, therefore $v{\not\models}B$.
Moreover, since $s{\models}{\square}A$ and $s{\leq}t$, therefore $v{\models}A$.
Thus, $v{\models}B{\leftharpoondown}A$: a contradiction.
\medskip
\end{proof}
Obviously, for all families $(\L_{i})_{i{\in}I}$ of IMLs, $\bigcap\{\L_{i}:\ i{\in}I\}$ is an IML.
As a result, there exists a least IML (denoted $\L_{\min}$).
In other respect, obviously, $\SFo$ is the greatest IML.
\begin{proposition}
$\L_{\min}$ is dual.
\end{proposition}
\begin{proof}
It suffices to notice that for all inference rules $\frac{(\alpha_{1},A_{1}),\ldots,(\alpha_{m},A_{m})}{(\beta,B)}$ in Table~\ref{table:rules}, $\frac{\mathtt{tr}(\alpha_{1},A_{1}),\ldots,\mathtt{tr}(\alpha_{m},A_{m})}{\mathtt{tr}(\beta,B)}$ is in Table~\ref{table:rules} too and for all signed formulas $(\gamma,C)$ in Table~\ref{table:axioms}, $\mathtt{tr}(\gamma,C)$ is in Table~\ref{table:axioms} too.
\medskip
\end{proof}
%
%
%
%
%
%
%
%
\begin{table}[ht]
\begin{center}
\begin{tabular}{|c|c|}
\hline
$(\Formula13^{+})$&$(+,p{\rightharpoonup}p{\vee}q)$ and $(+,q{\rightharpoonup}p{\vee}q)$
\\
\hline
$(\Formula14^{+})$&$(+,\bot{\rightharpoonup}p)$
\\
\hline
$(\Formula15^{+})$&$(+,p{\rightharpoonup}p)$
\\
\hline
$(\Formula16^{+})$&$(+,(p{\rightharpoonup}r){\rightharpoonup}((r{\rightharpoonup}q){\rightharpoonup}(p{\rightharpoonup}q)))$
\\
\hline
%
%
%
%
%
%
%
%
%
%
%
%
%
%
%
%
%
%
%
%
%
%
%
%
%
%
%
%
%
%
%
%
%
%
%
%
%
%
%
%
%
%
%
%
%
%
%
%
%
%
%
%
%
%
%
%
$(\Formula17^{+})$&$(+,{\lozenge}(t_{1}{\wedge}\ldots{\wedge}t_{n}{\wedge}p){\rightharpoonup}{\lozenge}t_{1}{\wedge}\ldots{\wedge}{\lozenge}t_{n})$
\\
\hline
$(\Formula18^{+})$&$(+,(t_{1}{\rightharpoonup}\ldots(t_{n}{\rightharpoonup}t_{1}{\wedge}\ldots{\wedge}t_{n})\ldots))$
\\
\hline
%
%
%
%
%
%
%
%
$(\Formula19^{+})$&$(+,t_{1}{\wedge}\ldots{\wedge}t_{n}{\wedge}(t{\rightharpoonup}u){\wedge}t{\rightharpoonup}u)$
\\
\hline
%
%
%
%
%
%
%
%
$(\Formula20^{+})$&$(+,{\blacklozenge}(({\bot}{\leftharpoondown}{\top}){\rightharpoonup}({\top}{\rightharpoonup}{\bot}){\vee}p){\rightharpoonup}{\lozenge}p)$
\\
\hline
$(\Formula21^{+})$&$(+,{\blacklozenge}((p{\leftharpoondown}{\top}){\rightharpoonup}p{\vee}q){\rightharpoonup}{\blacklozenge}(p{\vee}q))$
\\
\hline
$(\Formula22^{+})$&$(+,{\blacklozenge}(({\bot}{\leftharpoondown}p){\rightharpoonup}(p{\rightharpoonup}{\bot}){\vee}q){\rightharpoonup}{\blacklozenge}(p{\rightharpoonup}q))$
\\
\hline
$(\Formula23^{+})$&$(+,{\lozenge}p{\vee}{\lozenge}q{\rightharpoonup}{\lozenge}(p{\vee}q))$
\\
\hline
$(\Formula24^{+})$&$(+,{\lozenge}(p{\vee}q){\rightharpoonup}{\lozenge}p{\vee}{\lozenge}q)$
\\
\hline
$(\Formula13^{-})$&$(-,p{\leftharpoondown}p{\wedge}q)$ and $(-,q{\leftharpoondown}p{\wedge}q)$
\\
\hline
$(\Formula14^{-})$&$(-,\top{\leftharpoondown}p)$
\\
\hline
$(\Formula15^{-})$&$(-,p{\leftharpoondown}p)$
\\
\hline
$(\Formula16^{-})$&$(-,(p{\leftharpoondown}r){\leftharpoondown}((r{\leftharpoondown}q){\leftharpoondown}(p{\leftharpoondown}q)))$
\\
\hline
%
%
%
%
%
%
%
%
%
%
%
%
%
%
%
%
%
%
%
%
%
%
%
%
%
%
%
%
%
%
%
%
%
%
%
%
%
%
%
%
%
%
%
%
%
%
%
%
%
%
%
%
%
%
%
%
$(\Formula17^{-})$&$(-,{\square}(t_{1}{\vee}\ldots{\vee}t_{n}{\vee}p){\leftharpoondown}{\square}t_{1}{\vee}\ldots{\vee}{\square}t_{n})$
\\
\hline
$(\Formula18^{-})$&$(-,(t_{1}{\leftharpoondown}\ldots(t_{n}{\leftharpoondown}t_{1}{\vee}\ldots{\vee}t_{n})\ldots))$
\\
\hline
%
%
%
%
%
%
%
%
$(\Formula19^{-})$&$(-,t_{1}{\vee}\ldots{\vee}t_{n}{\vee}(t{\leftharpoondown}u){\vee}t{\leftharpoondown}u)$
\\
\hline
%
%
%
%
%
%
%
%
$(\Formula20^{-})$&$(-,{\blacksquare}(({\top}{\rightharpoonup}{\bot}){\leftharpoondown}({\bot}{\leftharpoondown}{\top}){\wedge}p){\leftharpoondown}{\square}p)$
\\
\hline
$(\Formula21^{-})$&$(-,{\blacksquare}((p{\rightharpoonup}{\bot}){\leftharpoondown}p{\wedge}q){\leftharpoondown}{\blacksquare}(p{\wedge}q))$
\\
\hline
$(\Formula22^{-})$&$(-,{\blacksquare}(({\top}{\rightharpoonup}p){\leftharpoondown}(p{\leftharpoondown}{\top}){\wedge}q){\leftharpoondown}{\blacksquare}(p{\leftharpoondown}q))$
\\
\hline
$(\Formula23^{-})$&$(-,{\square}p{\wedge}{\square}q{\leftharpoondown}{\square}(p{\wedge}q))$
\\
\hline
$(\Formula24^{-})$&$(-,{\square}(p{\wedge}q){\leftharpoondown}{\square}p{\wedge}{\square}q)$
\\
\hline
\end{tabular}
\vspace{+0.25cm}
\caption{}\label{table:formulas:for:Proposition:for:all:formulas:and:and:1}
\end{center}
\end{table}
\vspace{-0.50cm}
\begin{proposition}\label{Proposition:for:all:formulas:and:and:2:axioms}
For all IMLs $\L$, $\L$ contains the signed formulas in Table~\ref{table:formulas:for:Proposition:for:all:formulas:and:and:1}.
\end{proposition}
%
%
%
%
Obviously, for all IMLs $\L$ and for all sets $G$ of signed formulas, there exists a least IML (denoted $\L{\odot}G$) containing $\L{\cup}G$.
For all IMLs $\L$ and for all $(\alpha,A){\in}\SFo$, we write $\L{\odot}(\alpha,A)$ instead of $\L{\odot}\{(\alpha,A)\}$.
\begin{proposition}\label{Proposition:about:dual:IML}
For all dual IMLs $\L$ and for all dual sets $G$ of signed formulas, $\L{\odot}G$ is dual.
\end{proposition}
An IML $\L$ is {\em consistent}\/ if for all $A{\in}\Fo$, either $(+,A){\not\in}\L$, or $(-,A){\not\in}\L$.
\begin{proposition}\label{Proposition:for:all:constructive:model:logics:the:following:conditions:are:equivalent:is:consistent}
For all IMLs $\L$, the following conditions are equivalent: $\mathbf{1)}$~$\L$ is consistent; $\mathbf{2)}$~$\L{\not=}\SFo$; $\mathbf{3)}$~$(+,{\bot}){\not\in}\L$; $\mathbf{4)}$~$(-,{\top}){\not\in}\L$.
\end{proposition}
\begin{proof}
$\mathbf{2){\Rightarrow}3)}$:
Suppose $\L{\not=}\SFo$.
Hence, there exists $(\alpha,A){\in}\SFo$ such that $(\alpha,A){\not\in}
$\linebreak$
\L$.
Obviously, either $\alpha{=}+$, or $\alpha{=}-$.
\\
\\
In the former case, since $(+,{\bot}{\rightharpoonup}A){\in}\L$ using $(\Formula14^{+})$, therefore $(+,{\bot}){\not\in}\L$ using $(\Rule1^{+})$.
\\
\\
In the latter case, since $(-,{\bot}){\in}\L$ using $(\Formula5^{-})$, therefore $(-,{\bot}{\leftharpoondown}A){\not\in}\L$ using $(\Rule1^{-})$.
Thus, $(+,A{\rightharpoonup}{\bot}){\not\in}\L$ using $(\Rule2^{-})$.
Since $(+,{\bot}{\rightharpoonup}(A{\rightharpoonup}{\bot})){\in}\L$ using $(\Formula2^{+})$, therefore $(+,{\bot}){\not\in}\L$ using $(\Rule1^{+})$.
\\
\\
$\mathbf{3){\Rightarrow}4)}$:
Suppose $(+,{\bot}){\not\in}\L$.
Since $(+,{\top}){\in}\L$ using $(\Formula5^{+})$, therefore $(+,{\top}{\rightharpoonup}{\bot}){\not\in}\L$ using $(\Rule1^{+})$.
Consequently, $(-,{\bot}{\leftharpoondown}{\top}){\not\in}\L$ using $(\Rule2^{+})$.
Since $(-,{\top}{\leftharpoondown}({\bot}{\leftharpoondown}{\top})){\in}\L$ using $(\Formula2^{-})$, therefore $(-,{\top}){\not\in}\L$ using $(\Rule1^{-})$.
\\
\\
$\mathbf{4){\Rightarrow}1)}$:
Suppose there exists $A{\in}\Fo$ such that $(+,A){\in}\L$ and $(-,A){\in}\L$.
Since $(+,A
$\linebreak$
{\rightharpoonup}({\top}{\rightharpoonup}A)){\in}\L$ using $(\Formula2^{+})$, therefore $(+,{\top}{\rightharpoonup}A){\in}\L$ using $(\Rule1^{+})$.
Hence, $(-,A{\leftharpoondown}
$\linebreak$
{\top}){\in}\L$ using $(\Rule2^{-})$.
Since $(-,A){\in}\L$, therefore $(-,{\top}){\in}\L$ using $(\Rule1^{-})$.
\medskip
\end{proof}
{\bf From now on in this article, let $\L$ be a consistent IML.}
\section{Filters, ideals, tips and clips}\label{section:filters:ideals:tips:and:clips}
%
%
%
%
%
%
%
%
%
A {\em filter}\/ is a set $\Gamma$ of formulas such that
\begin{itemize}
\item $\L^{+}{\subseteq}\Gamma$,
\item for all $A,B{\in}\Fo$, if $A{\in}\Gamma$ and $A{\rightharpoonup}B{\in}\Gamma$ then $B{\in}\Gamma$.
\end{itemize}
Obviously, for all families $(\Gamma_{i})_{i{\in}I}$ of filters, $\bigcap\{\Gamma_{i}:\ i{\in}I\}$ is a filter and for all nonempty chains $(\Gamma_{i})_{i{\in}I}$ of filters, $\bigcup\{\Gamma_{i}:\ i{\in}I\}$ is a filter.
As a result, there exists a least filter (which is nothing but $\L^{+}$).
In other respect, obviously, $\Fo$ is the greatest filter.
\\
\\
A filter $\Gamma$ is {\em proper}\/ if ${\bot}{\not\in}\Gamma$.
\begin{proposition}\label{Proposition:for:all:constructive:model:logics:the:following:proper:filters}
For all filters $\Gamma$, $\Gamma$ is proper if and only if $\Gamma{\not=}\Fo$.
\end{proposition}
\begin{proposition}\label{Proposition:48bis:Gamma:plus:Sigma:is:a:filter}
For all filters $\Gamma$ and for all sets $\Sigma$ of formulas, $\Gamma{+}\Sigma$ is a filter.
\end{proposition}
\begin{proposition}\label{Proposition:about:filters:formulas:and:operation:plus:for:a:single:formula}
For all filters $\Gamma$ and for all $A{\in}\Fo$, $\Gamma{+}A{=}\{B{\in}\Fo:\ A{\rightharpoonup}B{\in}\Gamma\}$.
\end{proposition}
\begin{proposition}\label{Proposition:for:all:filters:and:for:all:formulas:is:a:filter:for:all:filters:if:and:then}
For all filters $\Gamma$ and for all $A{\in}\Fo$, $\Gamma{\subseteq}\Gamma{+}A$, $A{\in}\Gamma{+}A$, $\Gamma{+}A$ is a filter and for all filters $\Delta$, if $\Gamma{\subseteq}\Delta$ and $A{\in}\Delta$ then $\Gamma{+}A{\subseteq}\Delta$.
\end{proposition}
\begin{proposition}\label{Proposition:about:filters:and:square:operation}
For all filters $\Gamma$, ${\square}\Gamma$ is a filter.
\end{proposition}
A proper filter $\Gamma$ is {\em prime}\/ if for all $A,B{\in}\Fo$, if $A{\vee}B{\in}\Gamma$ then either $A{\in}\Gamma$, or $B{\in}\Gamma$.
\begin{proposition}[Lindenbaum Lemma for filters]\label{Proposition:let:be:a:filter:and:be:a:formula:if:then:there:exists:a:prime:filter:such:that:and:for:all:filters:if:and:then}
Let $\Gamma$ be a filter and $A$ be a formula.
If $A{\not\in}\Gamma$ then there exists a prime filter $\Delta$ such that $\Gamma{\subseteq}\Delta$, $A{\not\in}\Delta$ and for all filters $\Lambda$, if $\Delta{\subseteq}\Lambda$ and $A{\not\in}\Lambda$ then $\Delta{=}\Lambda$.
\end{proposition}
An {\em ideal}\/ is a set $\Gamma$ of formulas such that
\begin{itemize}
\item $\L^{-}{\subseteq}\Gamma$,
\item for all $A,B{\in}\Fo$, if $A{\in}\Gamma$ and $A{\leftharpoondown}B{\in}\Gamma$ then $B{\in}\Gamma$.
\end{itemize}
Obviously, for all families $(\Gamma_{i})_{i{\in}I}$ of ideals, $\bigcap\{\Gamma_{i}:\ i{\in}I\}$ is an ideal and for all nonempty chains $(\Gamma_{i})_{i{\in}I}$ of ideal, $\bigcup\{\Gamma_{i}:\ i{\in}I\}$ is an ideal.
As a result, there exists a least ideal (which is nothing but $\L^{-}$).
In other respect, obviously, $\Fo$ is the greatest ideal.
\\
\\
An ideal $\Gamma$ is {\em proper}\/ if ${\top}{\not\in}\Gamma$.
\begin{proposition}\label{Proposition:for:all:constructive:model:logics:the:following:proper:ideals}
For all ideals $\Gamma$, $\Gamma$ is proper if and only if $\Gamma{\not=}\Fo$.
\end{proposition}
\begin{proposition}\label{Proposition:53bis:Gamma:moins:Sigma:is:an:ideal}
For all ideals $\Gamma$ and for all sets $\Sigma$ of formulas, $\Gamma{-}\Sigma$ is an ideal.
\end{proposition}
\begin{proposition}\label{Proposition:about:ideals:formulas:and:operation:moins:for:a:single:formula}
For all ideals $\Gamma$ and for all $A{\in}\Fo$, $\Gamma{-}A{=}\{B{\in}\Fo:\ A{\leftharpoondown}B{\in}\Gamma\}$.
\end{proposition}
\begin{proposition}\label{Proposition:for:all:ideals:and:for:all:formulas:is:an:ideal:for:all:ideals:if:and:then}
For all ideals $\Gamma$ and for all $A{\in}\Fo$, $\Gamma{\subseteq}\Gamma{-}A$, $A{\in}\Gamma{-}A$, $\Gamma{-}A$ is an ideal and for all ideals $\Delta$, if $\Gamma{\subseteq}\Delta$ and $A{\in}\Delta$ then $\Gamma{-}A{\subseteq}\Delta$.
\end{proposition}
\begin{proposition}\label{Proposition:about:ideals:and:lozenge:operation}
For all ideals $\Gamma$, ${\lozenge}\Gamma$ is an ideal.
\end{proposition}
A proper ideal $\Gamma$ is {\em prime}\/ if for all $A,B{\in}\Fo$, if $A{\wedge}B{\in}\Gamma$ then either $A{\in}\Gamma$, or $B{\in}\Gamma$.
\begin{proposition}[Lindenbaum Lemma for ideals]\label{Proposition:let:be:an:ideal:and:be:a:formula:if:then:there:exists:a:prime:ideal:such:that:and:for:all:ideals:if:and:then}
Let $\Gamma$ be an ideal and $A$ be a formula.
If $A{\not\in}\Gamma$ then there exists a prime ideal $\Delta$ such that $\Gamma{\subseteq}\Delta$, $A{\not\in}\Delta$ and for all ideals $\Lambda$, if $\Delta{\subseteq}\Lambda$ and $A{\not\in}\Lambda$ then $\Delta{=}\Lambda$.
\end{proposition}
%
%
%
%
%
%
%
%
%
%
%
%
%
%
%
%
%
%
%
%
%
%
%
%
%
%
%
%
%
%
%
%
%
%
%
%
%
%
%
%
%
%
%
%
%
%
%
%
A {\em tip}\/ is a couple $(\Gamma^{+},\Gamma^{-})$ where $\Gamma^{+}$ is a filter and $\Gamma^{-}$ is an ideal.
\\
\\
A tip $(\Gamma^{+},\Gamma^{-})$ is {\em coherent}\/ if $\Gamma^{+}{\cap}\Gamma^{-}{=}{\emptyset}$.
\begin{proposition}\label{Proposition:let:be:a:tip:if:is:coherent:then:is:proper:and:is:proper}
Let $(\Gamma^{+},\Gamma^{-})$ be a tip.
If $(\Gamma^{+},\Gamma^{-})$ is coherent then $\Gamma^{+}$ is proper and $\Gamma^{-}$ is proper.
\end{proposition}
%
%
%
%
%
%
%
%
%
%
%
%
%
%
%
%
%
%
%
%
\begin{proposition}\label{proposition:consistent:BML:a:property:about:tips}
For all $A{\in}\Fo$,
\begin{enumerate}
\item if $A{\not\in}\L^{-}$ then $(\L^{+}{+}A,\L^{-})$ is coherent,
\item if $A{\not\in}\L^{+}$ then $(\L^{+},\L^{-}{-}A)$ is coherent.
\end{enumerate}
\end{proposition}
\begin{proof}
$\mathbf{1)}$ Suppose $(\L^{+}{+}A,\L^{-})$ is not coherent.
Hence, there exists $B{\in}\Fo$ such that $B{\in}\L^{+}{+}A$ and $B{\in}\L^{-}$.
Thus, $A{\rightharpoonup}B{\in}\L^{+}$.
Consequently, $B{\leftharpoondown}A{\in}\L^{-}$ using $(\Rule2^{-})$.
Since $B{\in}\L^{-}$, therefore $A{\in}\L^{-}$ using $(\Rule1^{-})$.
\\
\\
$\mathbf{2)}$ Dual to the proof of Item~$\mathbf{1)}$.
%
%
\medskip
\end{proof}
A coherent tip $(\Gamma^{+},\Gamma^{-})$ is {\em exhaustive}\/ if $\Gamma^{+}{\cup}\Gamma^{-}{=}\Fo$.
\begin{proposition}\label{Proposition:let:be:a:tip:if:is:exhaustive:then:is:prime:and:is:prime}
For all coherent tips $(\Gamma^{+},\Gamma^{-})$, if $(\Gamma^{+},\Gamma^{-})$ is exhaustive then $\Gamma^{+}$ is prime and $\Gamma^{-}$ is prime.
\end{proposition}
\begin{proposition}\label{Proposition:for:all:exhaustive:tips:if:then:if:then}
For all exhaustive tips $(\Gamma^{+},\Gamma^{-}),(\Delta^{+},\Delta^{-})$, $\Gamma^{+}{\subseteq}\Delta^{+}$ if and only if $\Delta^{-}{\subseteq}\Gamma^{-}$.
\end{proposition}
\begin{proposition}[Lindenbaum Lemma for tips]\label{Proposition:for:all:coherent:tips:there:exists:an:exhaustive:tip:such:that:and}
For all coherent tips $(\Gamma^{+},\Gamma^{-})$, there exists an exhaustive tip $(\Delta^{+},\Delta^{-})$ such that $\Gamma^{+}{\subseteq}\Delta^{+}$ and $\Gamma^{-}{\subseteq}\Delta^{-}$.
\end{proposition}
A {\em clip}\/ is a triple $((\Gamma^{+},\Gamma^{-}),A^{+},A^{-})$ where $(\Gamma^{+},\Gamma^{-})$ is an exhaustive tip and $A^{+},A^{-}$ are formulas.
\\
\\
A clip $((\Gamma^{+},\Gamma^{-}),A^{+},A^{-})$ is {\em balanced}\/ if for all $C{\in}\Fo$,
\begin{itemize}
\item if ${\blacklozenge}((A^{-}{\leftharpoondown}A^{+}){\rightharpoonup}(A^{+}{\rightharpoonup}A^{-}){\vee}C){\in}\Gamma^{+}$ then ${\lozenge}C{\in}\Gamma^{+}$,
\item if ${\blacksquare}((A^{+}{\rightharpoonup}A^{-}){\leftharpoondown}(A^{-}{\leftharpoondown}A^{+}){\wedge}C){\in}\Gamma^{-}$ then ${\square}C{\in}\Gamma^{-}$.
\end{itemize}
\begin{proposition}\label{Proposition:for:all:exhaustive:tips:is:balanced}
For all exhaustive tips $(\Gamma^{+},\Gamma^{-})$, $((\Gamma^{+},\Gamma^{-}),{\top},{\bot})$ is balanced.
\end{proposition}
\begin{proof}
Let $(\Gamma^{+},\Gamma^{-})$ be an exhaustive tip.
For the sake of the contradiction, suppose $((\Gamma^{+},\Gamma^{-}),{\top},{\bot})$ is not balanced.
Hence, there exists $C{\in}\Fo$ such that either ${\blacklozenge}(({\bot}{\leftharpoondown}{\top}){\rightharpoonup}({\top}{\rightharpoonup}{\bot}){\vee}C){\in}\Gamma^{+}$ and ${\lozenge}C{\not\in}\Gamma^{+}$, or ${\blacksquare}(({\top}{\rightharpoonup}{\bot}){\leftharpoondown}({\bot}{\leftharpoondown}{\top}){\wedge}C){\in}\Gamma^{-}$ and ${\square}C{\not\in}\Gamma^{-}$.
\\
\\
In the former case, since ${\blacklozenge}(({\bot}{\leftharpoondown}{\top}){\rightharpoonup}({\top}{\rightharpoonup}{\bot}){\vee}C){\rightharpoonup}{\lozenge}C{\in}\L^{+}$ using $(\Formula20^{+})$, therefore ${\lozenge}C{\in}\Gamma^{+}$: a contradiction.
%
%
\\
\\
In the latter case, a dual reasoning can be done.
\medskip
\end{proof}
\section{Existence Properties}\label{section:existence:properties}
The following Existence Properties will be crucially used in the proof of Proposition~\ref{Proposition:let:be:a:formula:for:all}.
%
%
%
%
%
%
%
%
\begin{proposition}[Existence Property for $\rightharpoonup$]\label{Proposition:let:be:a:balanced:clip:let:be:formulas:if:then:there:exists:a:balanced:clip:such:that:and:1}
Let $((\Gamma^{+},\Gamma^{-}),A^{+},A^{-})$ be a balanced clip.
Let $C,D$ be formulas.
If $C{\rightharpoonup}D{\in}\Gamma^{-}$ then there exists a balanced clip $((\Delta^{+},\Delta^{-}),
$\linebreak$
B^{+},B^{-})$ such that $\Gamma^{+}{\subseteq}\Delta^{+}$, $\Delta^{-}{\subseteq}\Gamma^{-}$, $C{\in}\Delta^{+}$ and $D{\in}\Delta^{-}$.
\end{proposition}
\begin{proof}
%
%
Suppose $C{\rightharpoonup}D{\in}\Gamma^{-}$.
Let ${\mathcal S}$ be the set of all filters $\Delta$ such that $\Gamma^{+}{\subseteq}\Delta$ and $C{\rightharpoonup}D{\not\in}\Delta$.
\begin{claim}
$\Gamma^{+}$ is in ${\mathcal S}$.
\end{claim}
\begin{proof}
For the sake of the contradiction, suppose $\Gamma^{+}$ is not in ${\mathcal S}$.
Hence, $C{\rightharpoonup}D{\in}\Gamma^{+}$.
Thus $C{\rightharpoonup}D{\not\in}\Gamma^{-}$: a contradiction.
\medskip
\end{proof}
Moreover, for all nonempty chains $(\Delta_{i})_{i{\in}I}$ in ${\mathcal S}$, $\bigcup\{\Delta_{i}\ :\ i{\in}I\}$ is in ${\mathcal S}$.
Hence, by Kuratowski-Zorn Lemma,\footnote{See~\cite[Chapter~$10$]{Davey:Priestley:2002} and~\cite[Chapter~$1$]{Wechler:1992}.} ${\mathcal S}$ possesses a maximal element $\Delta$.
It is a routine task to show that $\Delta$ is a prime filter such that $\Gamma^{+}{\subseteq}\Delta$ and $C{\rightharpoonup}D{\not\in}\Delta$.
Moreover, $C{\in}\Delta$ and $D{\not\in}\Delta$.
And in the end, choose $\Delta^{+}{=}\Delta$, $\Delta^{-}{=}\Fo{\setminus}\Delta$, $B^{+}{=}{\top}$ and $B^{-}{=}{\bot}$.
Thus, $\Gamma^{+}{\subseteq}\Delta^{+}$, $\Delta^{-}{\subseteq}\Gamma^{-}$, $C{\in}\Delta^{+}$ and $D{\in}\Delta^{-}$.
\medskip
\end{proof}
\begin{proposition}[Existence Property for $\leftharpoondown$]\label{Proposition:let:be:a:balanced:clip:let:be:formulas:if:then:there:exists:a:balanced:clip:such:that:and:2}
Let $((\Gamma^{+},\Gamma^{-}),A^{+},A^{-})$ be a balanced clip.
Let $C,D$ be formulas.
If $C{\leftharpoondown}D{\in}\Gamma^{+}$ then there exists a balanced clip $((\Delta^{+},\Delta^{-}),
$\linebreak$
B^{+},B^{-})$ such that $\Delta^{+}{\subseteq}\Gamma^{+}$, $\Gamma^{-}{\subseteq}\Delta^{-}$, $C{\in}\Delta^{-}$ and $D{\in}\Delta^{+}$.
\end{proposition}
\begin{proof}
Dual to the proof of Proposition~\ref{Proposition:let:be:a:balanced:clip:let:be:formulas:if:then:there:exists:a:balanced:clip:such:that:and:1}.
\medskip
\end{proof}
%
%
%
%
%
%
%
%
\begin{proposition}[Existence Property for $\lozenge$]\label{Proposition:let:be:a:balanced:clip:let:be:a:formula:if:then:for:all:balanced:clips:if:and:then:there:exists:a:balanced:clip:such:that:and:1:lozenge:Gamma:plus}
Let $((\Gamma^{+},\Gamma^{-}),A^{+},A^{-})$ be a balanced clip.
Let $D$ be a formula.
If ${\lozenge}D{\in}\Gamma^{+}$ then there exists a balanced clip $((\Delta^{+},\Delta^{-}),B^{+},
$\linebreak$
B^{-})$ such that $\Delta^{+}{\subseteq}\Gamma^{+}$ and $\Gamma^{-}{\subseteq}\Delta^{-}$ and there exists a balanced clip $((\Lambda^{+},\Lambda^{-}),C^{+},
$\linebreak$
C^{-})$ such that ${\square}\Delta^{+}{\subseteq}\Lambda^{+}$, ${\lozenge}\Delta^{-}{\subseteq}\Lambda^{-}$, $B^{-}{\leftharpoondown}B^{+}{\in}\Lambda^{+}$, $B^{+}{\rightharpoonup}B^{-}{\in}\Lambda^{-}$ and $D{\in}\Lambda^{+}$.
\end{proposition}
\begin{proof}\footnote{The reader is invited to see where is the only use of the inference rule $(\Rule8^{+})$ in this proof.}
Suppose ${\lozenge}D{\in}\Gamma^{+}$.
Let ${\mathcal S}$ be the set of all filters $\Lambda$ such that for all $E{\in}\Fo$, if $E{\in}\Lambda$ then ${\lozenge}E{\in}\Gamma^{+}$ and $D{\in}\Lambda$.
\begin{claim}
$\L^{+}{+}D$ is in ${\mathcal S}$.
\end{claim}
\begin{proof}
For the sake of the contradiction, suppose $\L^{+}{+}D$ is not in ${\mathcal S}$.
Hence, there exists $E{\in}\Fo$ such that $E{\in}\L^{+}{+}D$ and ${\lozenge}E{\not\in}\Gamma^{+}$.
Thus, $D{\rightharpoonup}E{\in}\L^{+}$.
Consequently, ${\lozenge}D{\rightharpoonup}{\lozenge}E{\in}\L^{+}$ using $(\Rule3^{+})$.
Since ${\lozenge}D{\in}\Gamma^{+}$, therefore ${\lozenge}E{\in}\Gamma^{+}$: a contradiction.
\medskip
\end{proof}
Moreover, for all nonempty chains $(\Lambda_{i})_{i{\in}I}$ in ${\mathcal S}$, $\bigcup\{\Lambda_{i}\ :\ i{\in}I\}$ is in ${\mathcal S}$.
Hence, by Kuratowski-Zorn Lemma, ${\mathcal S}$ possesses a maximal element $\Lambda$.
It is a routine task to show that $\Lambda$ is a prime filter such that for all $E{\in}\Fo$, if $E{\in}\Lambda$ then ${\lozenge}E{\in}\Gamma^{+}$ and $D{\in}\Lambda$.
Let ${\mathcal T}$ be the set of all filters $\Delta$ such that for all $G,H{\in}\Fo$, if $G{\vee}{\square}H{\in}\Delta$ then either $G{\in}\Gamma^{+}$, or $H{\in}\Lambda$ and for all $I{\in}\Fo$, if $I{\in}\Lambda$ then ${\lozenge}I{\in}\Delta$.
\begin{claim}
$\L^{+}{+}\{{\lozenge}E:\ E{\in}\Lambda\}$ is in ${\mathcal T}$.
\end{claim}
\begin{proof}
%
%
For the sake of the contradiction, suppose $\L^{+}{+}\{{\lozenge}E:\ E{\in}\Lambda\}$ is not in ${\mathcal T}$.
Hence, either there exists $G,H{\in}\Fo$ such that $G{\vee}{\square}H{\in}\L^{+}{+}\{{\lozenge}E:\ E{\in}\Lambda\}$, $G{\not\in}\Gamma^{+}$ and $H{\not\in}\Lambda$, or there exists $I{\in}\Fo$ such that $I{\in}\Lambda$ and ${\lozenge}I{\not\in}\L^{+}{+}\{{\lozenge}E:\ E{\in}\Lambda\}$.
\\
\\
In the former case, there exists $n{\in}\N$ and there exists $E_{1},\ldots,E_{n}{\in}\Lambda$ such that ${\lozenge}E_{1}{\wedge}\ldots
$\linebreak$
{\wedge}{\lozenge}E_{n}{\rightharpoonup}G{\vee}{\square}H{\in}\L^{+}$.
Since $H{\not\in}\Lambda$, therefore by the maximality of $\Lambda$ in ${\mathcal S}$, there exists $E^{\prime}{\in}\Fo$ such that $E^{\prime}{\in}\Lambda{+}H$ and ${\lozenge}E^{\prime}{\not\in}\Gamma^{+}$.
Thus, $H{\rightharpoonup}E^{\prime}{\in}\Lambda$.
Since $(E_{1}{\rightharpoonup}\ldots(E_{n}{\rightharpoonup}
$\linebreak$
((H{\rightharpoonup}E^{\prime}){\rightharpoonup}E_{1}{\wedge}\ldots{\wedge}E_{n}{\wedge}(H{\rightharpoonup}E^{\prime})))\ldots){\in}\L^{+}$ using $(\Formula18^{+})$ and $E_{1},\ldots,E_{n}{\in}\Lambda$,
\linebreak
therefore $E_{1}{\wedge}\ldots{\wedge}E_{n}{\wedge}(H{\rightharpoonup}E^{\prime}){\in}\Lambda$.
Consequently, ${\lozenge}(E_{1}{\wedge}\ldots{\wedge}E_{n}{\wedge}(H{\rightharpoonup}E^{\prime})){\in}
$\linebreak$
\Gamma^{+}$.
Since $({\lozenge}(E_{1}{\wedge}\ldots{\wedge}E_{n}{\wedge}(H{\rightharpoonup}E^{\prime})){\rightharpoonup}{\lozenge}E_{1}{\wedge}\ldots{\wedge}{\lozenge}E_{n}){\rightharpoonup}(({\lozenge}E_{1}{\wedge}\ldots{\wedge}{\lozenge}E_{n}{\rightharpoonup}G
$\linebreak$
{\vee}{\square}H){\rightharpoonup}({\lozenge}(E_{1}{\wedge}\ldots{\wedge}E_{n}{\wedge}(H{\rightharpoonup}E^{\prime})){\rightharpoonup}G{\vee}{\square}H)){\in}\L^{+}$ using $(\Formula16^{+})$, ${\lozenge}(E_{1}{\wedge}\ldots{\wedge}
$\linebreak$
E_{n}{\wedge}(H{\rightharpoonup}E^{\prime})){\rightharpoonup}{\lozenge}E_{1}{\wedge}\ldots{\wedge}{\lozenge}E_{n}{\in}\L^{+}$ using $(\Formula17^{+})$ and ${\lozenge}E_{1}{\wedge}\ldots{\wedge}{\lozenge}E_{n}{\rightharpoonup}G{\vee}{\square}H
$\linebreak$
{\in}\L^{+}$, therefore ${\lozenge}(E_{1}{\wedge}\ldots{\wedge}E_{n}{\wedge}(H{\rightharpoonup}E^{\prime})){\rightharpoonup}G{\vee}{\square}H{\in}\L^{+}$ using $(\Rule1^{+})$.
Hence,
\linebreak$
{\lozenge}(E_{1}{\wedge}\ldots{\wedge}E_{n}{\wedge}(H{\rightharpoonup}E^{\prime})){\rightharpoonup}G{\vee}{\lozenge}(E_{1}{\wedge}\ldots{\wedge}E_{n}{\wedge}(H{\rightharpoonup}E^{\prime}){\wedge}H){\in}\L^{+}$ using $(\Rule8^{+})$.
Since ${\lozenge}(E_{1}{\wedge}\ldots{\wedge}E_{n}{\wedge}(H{\rightharpoonup}E^{\prime})){\in}\Gamma^{+}$, therefore $G{\vee}{\lozenge}(E_{1}{\wedge}\ldots{\wedge}E_{n}{\wedge}(H{\rightharpoonup}E^{\prime}){\wedge}H)
$\linebreak$
{\in}\Gamma^{+}$.
Thus, either $G{\in}\Gamma^{+}$, or ${\lozenge}(E_{1}{\wedge}\ldots{\wedge}E_{n}{\wedge}(H{\rightharpoonup}E^{\prime}){\wedge}H){\in}\Gamma^{+}$.
Since $G{\not\in}\Gamma^{+}$, therefore ${\lozenge}(E_{1}{\wedge}\ldots{\wedge}E_{n}{\wedge}(H{\rightharpoonup}E^{\prime}){\wedge}H){\in}\Gamma^{+}$.
Since $E_{1}{\wedge}\ldots{\wedge}E_{n}{\wedge}(H{\rightharpoonup}E^{\prime}){\wedge}H{\rightharpoonup}E^{\prime}{\in}
$\linebreak$
\L^{+}$ using $(\Formula19^{+})$, therefore ${\lozenge}(E_{1}{\wedge}\ldots{\wedge}E_{n}{\wedge}(H{\rightharpoonup}E^{\prime}){\wedge}H){\rightharpoonup}{\lozenge}E^{\prime}{\in}\L^{+}$ using $(\Rule3^{+})$.
Since ${\lozenge}(E_{1}{\wedge}\ldots{\wedge}E_{n}{\wedge}(H{\rightharpoonup}E^{\prime}){\wedge}H){\in}\Gamma^{+}$, therefore ${\lozenge}E^{\prime}{\in}\Gamma^{+}$ a contradiction.
\\
\\
In the latter case, since ${\lozenge}I{\rightharpoonup}{\lozenge}I{\in}\L^{+}$ using $(\Formula15^{+})$, therefore ${\lozenge}I{\in}\L^{+}{+}\{{\lozenge}E:\ E{\in}\Lambda\}$: a contradiction.
\medskip
\end{proof}
Moreover, for all nonempty chains $(\Delta_{i})_{i{\in}I}$ in ${\mathcal T}$, $\bigcup\{\Delta_{i}\ :\ i{\in}I\}$ is in ${\mathcal T}$.
Hence, by Kuratowski-Zorn Lemma, ${\mathcal T}$ possesses a maximal element $\Delta$.
It is a routine task to show that $\Delta$ is a prime filter such that for all $G,H{\in}\Fo$, if $G{\vee}{\square}H{\in}\Delta$ then either $G{\in}\Gamma^{+}$, or $H{\in}\Lambda$ and for all $I{\in}\Fo$, if $I{\in}\Lambda$ then ${\lozenge}I{\in}\Delta$.
And in the end, choose $\Delta^{+}{=}\Delta$, $\Delta^{-}{=}\Fo{\setminus}\Delta$, $B^{+}{=}{\top}$, $B^{-}{=}{\bot}$, $\Lambda^{+}{=}\Lambda$, $\Lambda^{-}{=}\Fo{\setminus}\Lambda$, $C^{+}{=}{\top}$ and $C^{-}{=}{\bot}$.
Thus, $\Delta^{+}{\subseteq}
$\linebreak$
\Gamma^{+}$, $\Gamma^{-}{\subseteq}\Delta^{-}$, ${\square}\Delta^{+}{\subseteq}\Lambda^{+}$, ${\lozenge}\Delta^{-}{\subseteq}\Lambda^{-}$, $B^{-}{\leftharpoondown}B^{+}{\in}\Lambda^{+}$, $B^{+}{\rightharpoonup}B^{-}{\in}\Lambda^{-}$ and $D{\in}\Lambda^{+}$.
\medskip
\end{proof}
%
%
%
%
%
%
%
%
%
%
%
%
%
%
\begin{proposition}[Existence Property for $\square$]\label{Proposition:let:be:a:balanced:clip:let:be:a:formula:if:then:for:all:balanced:clips:if:and:then:there:exists:a:balanced:clip:such:that:and:2:square:Gamma:moins}
Let $((\Gamma^{+},\Gamma^{-}),A^{+},A^{-})$ be a balanced clip.
Let $D$ be a formula.
If ${\square}D{\in}\Gamma^{-}$ then there exists a balanced clip $((\Delta^{+},\Delta^{-}),B^{+},
$\linebreak$
B^{-})$ such that $\Gamma^{+}{\subseteq}\Delta^{+}$ and $\Delta^{-}{\subseteq}\Gamma^{-}$ and there exists a balanced clip $((\Lambda^{+},\Lambda^{-}),C^{+},
$\linebreak$
C^{-})$ such that ${\square}\Delta^{+}{\subseteq}\Lambda^{+}$, ${\lozenge}\Delta^{-}{\subseteq}\Lambda^{-}$, $B^{-}{\leftharpoondown}B^{+}{\in}\Lambda^{+}$, $B^{+}{\rightharpoonup}B^{-}{\in}\Lambda^{-}$ and $D{\in}\Lambda^{-}$.
\end{proposition}
\begin{proof}
Dual to the proof of Proposition~\ref{Proposition:let:be:a:balanced:clip:let:be:a:formula:if:then:for:all:balanced:clips:if:and:then:there:exists:a:balanced:clip:such:that:and:1:lozenge:Gamma:plus}.
\medskip
\end{proof}
\begin{proposition}[Positive Existence Property for $\blacklozenge$]\label{Proposition:let:be:a:balanced:clip:let:be:a:formula:if:then:for:all:balanced:clips:if:and:then:there:exists:a:balanced:clip:such:that:and:1}
Let $((\Gamma^{+},\Gamma^{-}),A^{+},A^{-})$ be a
\linebreak
balanced clip.
Let $D$ be a formula.
If ${\blacklozenge}D{\in}\Gamma^{+}$ then for all balanced clips $((\Delta^{+},\Delta^{-}),
$\linebreak$
B^{+},B^{-})$, if $\Gamma^{+}{\subseteq}\Delta^{+}$ and $\Delta^{-}{\subseteq}\Gamma^{-}$ then there exists a balanced clip $((\Lambda^{+},\Lambda^{-}),C^{+},
$\linebreak$
C^{-})$ such that ${\square}\Delta^{+}{\subseteq}\Lambda^{+}$, ${\lozenge}\Delta^{-}{\subseteq}\Lambda^{-}$, $B^{-}{\leftharpoondown}B^{+}{\in}\Lambda^{+}$, $B^{+}{\rightharpoonup}B^{-}{\in}\Lambda^{-}$ and $D{\in}\Lambda^{+}$.
\end{proposition}
\begin{proof}
Suppose ${\blacklozenge}D{\in}\Gamma^{+}$.
Let $((\Delta^{+},\Delta^{-}),B^{+},B^{-})$ be a balanced clip such that $\Gamma^{+}{\subseteq}\Delta^{+}$ and $\Delta^{-}{\subseteq}\Gamma^{-}$.
Since ${\blacklozenge}D{\in}\Gamma^{+}$, therefore ${\blacklozenge}D{\in}\Delta^{+}$.
Let ${\mathcal S}$ be the set of all filters $\Lambda$ such that ${\square}\Delta^{+}{\subseteq}\Lambda$, for all $E{\in}\Fo$, if $(B^{+}{\rightharpoonup}B^{-}){\vee}E{\in}\Lambda$ then ${\lozenge}E{\in}\Delta^{+}$, $B^{-}{\leftharpoondown}B^{+}{\in}\Lambda$ and $D{\in}\Lambda$.
\begin{claim}
$({\square}\Delta^{+}{+}D){+}(B^{-}{\leftharpoondown}B^{+})$ is in ${\mathcal S}$.
\end{claim}
\begin{proof}
For the sake of the contradiction, suppose $({\square}\Delta^{+}{+}D){+}(B^{-}{\leftharpoondown}B^{+})$ is not in ${\mathcal S}$.
Hence, there exists $E{\in}\Fo$ such that $(B^{+}{\rightharpoonup}B^{-}){\vee}E{\in}({\square}\Delta^{+}{+}D){+}
(B^{-}{\leftharpoondown}B^{+})$ and ${\lozenge}E{\not\in}\Delta^{+}$.
Thus, ${\square}(D{\rightharpoonup}((B^{-}{\leftharpoondown}B^{+}){\rightharpoonup}(B^{+}{\rightharpoonup}B^{-}){\vee}E)){\in}\Delta^{+}$.
Since ${\square}(D{\rightharpoonup}((B^{-}{\leftharpoondown}
$\linebreak$
B^{+}){\rightharpoonup}(B^{+}{\rightharpoonup}B^{-}){\vee}E)){\rightharpoonup}({\blacklozenge}D{\rightharpoonup}{\blacklozenge}((B^{-}{\leftharpoondown}B^{+}){\rightharpoonup}(B^{+}{\rightharpoonup}B^{-}){\vee}E)){\in}\L^{+}$ using
\linebreak$
(\Formula10^{+})$ and ${\blacklozenge}D{\in}\Delta^{+}$, therefore ${\blacklozenge}((B^{-}{\leftharpoondown}B^{+}){\rightharpoonup}(B^{+}{\rightharpoonup}B^{-}){\vee}E){\in}\Delta^{+}$.
Consequent\-ly, $((\Delta^{+},\Delta^{-}),B^{+},B^{-})$ being a balanced clip, ${\lozenge}E{\in}\Delta^{+}$: a contradiction.
\medskip
\end{proof}
Moreover, for all nonempty chains $(\Lambda_{i})_{i{\in}I}$ in ${\mathcal S}$, $\bigcup\{\Lambda_{i}\ :\ i{\in}I\}$ is in ${\mathcal S}$.
Hence, by Kuratowski-Zorn Lemma, ${\mathcal S}$ possesses a maximal element $\Lambda$.
It is a routine task to show that $\Lambda$ is a prime filter such that ${\square}\Delta^{+}{\subseteq}\Lambda$, ${\lozenge}\Delta^{-}{\cap}\Lambda{=}\emptyset$, $B^{-}{\leftharpoondown}B^{+}{\in}\Lambda$, $B^{+}{\rightharpoonup}B^{-}
$\linebreak$
{\not\in}\Lambda$ and $D{\in}\Lambda$.
And in the end, choose $\Lambda^{+}{=}\Lambda$, $\Lambda^{-}{=}\Fo{\setminus}\Lambda$, $C^{+}{=}{\top}$ and $C^{-}{=}{\bot}$.
Thus, ${\square}\Delta^{+}{\subseteq}\Lambda^{+}$, ${\lozenge}\Delta^{-}{\subseteq}\Lambda^{-}$, $B^{-}{\leftharpoondown}B^{+}{\in}\Lambda^{+}$, $B^{+}{\rightharpoonup}B^{-}{\in}\Lambda^{-}$ and $D{\in}\Lambda^{+}$.
\medskip
\end{proof}
\begin{proposition}[Negative Existence Property for $\blacklozenge$]\label{Proposition:let:be:a:balanced:clip:let:be:a:formula:if:then:there:exists:a:balanced:clip:such:that:and:and:for:all:balanced:clips:if:and:then:1}
Let $((\Gamma^{+},\Gamma^{-}),A^{+},A^{-})$ be a
\linebreak
balanced clip.
Let $D$ be a formula.
If ${\blacklozenge}D{\in}\Gamma^{-}$ then there exists a balanced clip $((\Delta^{+},\Delta^{-}),B^{+},B^{-})$ such that $\Gamma^{+}{\subseteq}\Delta^{+}$ and $\Delta^{-}{\subseteq}\Gamma^{-}$ and for all balanced clips $((\Lambda^{+},\Lambda^{-}),C^{+},C^{-})$, if ${\square}\Delta^{+}{\subseteq}\Lambda^{+}$, ${\lozenge}\Delta^{-}{\subseteq}\Lambda^{-}$, $B^{-}{\leftharpoondown}B^{+}{\in}\Lambda^{+}$ and $B^{+}{\rightharpoonup}B^{-}{\in}\Lambda^{-}$ then $D{\in}\Lambda^{-}$.
\end{proposition}
\begin{proof}
Suppose ${\blacklozenge}D{\in}\Gamma^{-}$.
Let ${\mathcal S}$ be the set of all filters $\Delta$ such that $\Gamma^{+}{\subseteq}\Delta$ and ${\blacklozenge}D{\not\in}\Delta$.
\begin{claim}
$\Gamma^{+}$ is in ${\mathcal S}$.
\end{claim}
\begin{proof}
For the sake of the contradiction, suppose $\Gamma^{+}$ is not in ${\mathcal S}$.
Hence, ${\blacklozenge}D{\in}\Gamma^{+}$.
Thus, ${\blacklozenge}D{\not\in}\Gamma^{-}$: a contradiction.
\medskip
\end{proof}
Moreover, for all nonempty chains $(\Delta_{i})_{i{\in}I}$ in ${\mathcal S}$, $\bigcup\{\Delta_{i}\ :\ i{\in}I\}$ is in ${\mathcal S}$.
Hence, by Kuratowski-Zorn Lemma, ${\mathcal S}$ possesses a maximal element $\Delta$.
It is a routine task to show that $\Delta$ is a prime filter such that $\Gamma^{+}{\subseteq}\Delta$ and ${\blacklozenge}D{\not\in}\Delta$.
\begin{claim}
$((\Delta,\Fo{\setminus}\Delta),{\top},D)$ is balanced.
\end{claim}
\begin{proof}
For the sake of the contradiction, suppose $((\Delta,\Fo{\setminus}\Delta),{\top},D)$ is not balanced.
Hence, there exists $C{\in}\Fo$ such that either ${\blacklozenge}((D{\leftharpoondown}{\top}){\rightharpoonup}({\top}{\rightharpoonup}D){\vee}C){\in}\Delta$ and ${\lozenge}C{\not\in}\Delta$, or ${\blacksquare}(({\top}{\rightharpoonup}D){\leftharpoondown}(D{\leftharpoondown}{\top}){\wedge}C){\not\in}\Delta$ and ${\square}C{\in}\Delta$.
\\
\\
In the former case, since ${\blacklozenge}((D{\leftharpoondown}{\top}){\rightharpoonup}({\top}{\rightharpoonup}D){\vee}C){\rightharpoonup}{\blacklozenge}(C{\vee}D){\in}\L^{+}$ using $(\Formula21^{+})$, therefore ${\blacklozenge}(C{\vee}D){\in}\Delta$.
Since ${\blacklozenge}(C{\vee}D){\rightharpoonup}(({\lozenge}C{\rightharpoonup}{\blacklozenge}D){\rightharpoonup}{\blacklozenge}D){\in}\L^{+}$ using $(\Formula11^{+})$,
\linebreak
therefore $({\lozenge}C{\rightharpoonup}{\blacklozenge}D){\rightharpoonup}{\blacklozenge}D{\in}\Delta$.
Since ${\lozenge}C{\not\in}\Delta$, therefore by the maximality of $\Delta$ in ${\mathcal S}$, ${\blacklozenge}D{\in}\Delta{+}{\lozenge}C$.
Thus, ${\lozenge}C{\rightharpoonup}{\blacklozenge}D{\in}\Delta$.
Since $({\lozenge}C{\rightharpoonup}{\blacklozenge}D){\rightharpoonup}{\blacklozenge}D{\in}\Delta$, therefore ${\blacklozenge}D{\in}\Delta$: a contradiction.
\\
\\
In the latter case, since ${\blacksquare}(({\top}{\rightharpoonup}D){\leftharpoondown}(D{\leftharpoondown}{\top}){\wedge}C){\leftharpoondown}{\blacksquare}(D{\leftharpoondown}C){\in}\L^{-}$ using $(\Formula22^{-})$, therefore ${\blacksquare}(D{\leftharpoondown}C){\not\in}\Delta$.
Consequently, by the maximality of $\Delta$ in ${\mathcal S}$, ${\blacklozenge}D{\in}\Delta{+}{\blacksquare}(D{\leftharpoondown}
$\linebreak$
C)$.
Hence, ${\blacksquare}(D{\leftharpoondown}C){\rightharpoonup}{\blacklozenge}D{\in}\Delta$.
Since ${\square}C{\rightharpoonup}(({\blacksquare}(D{\leftharpoondown}C){\rightharpoonup}{\blacklozenge}D){\rightharpoonup}{\blacklozenge}D){\in}\L^{+}$ using $\Formula12^{+}$ and ${\square}C{\in}\Delta$, therefore ${\blacklozenge}D{\in}\Delta$: a contradiction.
\medskip
\end{proof}
And in the end, choose $\Delta^{+}{=}\Delta$, $\Delta^{-}{=}\Fo{\setminus}\Delta$, $B^{+}{=}{\top}$ and $B^{-}{=}D$.
Thus, $\Gamma^{+}{\subseteq}\Delta^{+}$ and $\Delta^{-}{\subseteq}\Gamma^{-}$ and for all balanced clips $((\Lambda^{+},\Lambda^{-}),C^{+},C^{-})$, if ${\square}\Delta^{+}{\subseteq}\Lambda^{+}$, ${\lozenge}\Delta^{-}{\subseteq}\Lambda^{-}$, $B^{-}{\leftharpoondown}B^{+}{\in}\Lambda^{+}$ and $B^{+}{\rightharpoonup}B^{-}{\in}\Lambda^{-}$ then $D{\in}\Lambda^{-}$.
\medskip
\end{proof}
\begin{proposition}[Positive Existence Property for $\blacksquare$]\label{Proposition:let:be:a:balanced:clip:let:be:a:formula:if:then:there:exists:a:balanced:clip:such:that:and:and:for:all:balanced:clips:if:and:then:2}
Let $((\Gamma^{+},\Gamma^{-}),A^{+},A^{-})$ be a
\linebreak
balanced clip.
Let $D$ be a formula.
If ${\blacksquare}D{\in}\Gamma^{+}$ then there exists a balanced clip $((\Delta^{+},\Delta^{-}),B^{+},B^{-})$ such that $\Delta^{+}{\subseteq}\Gamma^{+}$ and $\Gamma^{-}{\subseteq}\Delta^{-}$ and for all balanced clips $((\Lambda^{+},\Lambda^{-}),C^{+},C^{-})$, if ${\square}\Delta^{+}{\subseteq}\Lambda^{+}$, ${\lozenge}\Delta^{-}{\subseteq}\Lambda^{-}$, $B^{-}{\leftharpoondown}B^{+}{\in}\Lambda^{+}$ and $B^{+}{\rightharpoonup}B^{-}{\in}\Lambda^{-}$ then $D{\in}\Lambda^{+}$.
\end{proposition}
\begin{proof}
Dual to the proof of Proposition~\ref{Proposition:let:be:a:balanced:clip:let:be:a:formula:if:then:there:exists:a:balanced:clip:such:that:and:and:for:all:balanced:clips:if:and:then:1}.
\medskip
\end{proof}
\begin{proposition}[Negative Existence Property for $\blacksquare$]\label{Proposition:let:be:a:balanced:clip:let:be:a:formula:if:then:for:all:balanced:clips:if:and:then:there:exists:a:balanced:clip:such:that:and:2}
Let $((\Gamma^{+},\Gamma^{-}),A^{+},A^{-})$ be a balanced clip.
Let $D$ be a formula.
If ${\blacksquare}D{\in}\Gamma^{-}$ then for all balanced clips $((\Delta^{+},\Delta^{-}),
$\linebreak$
B^{+},B^{-})$, if $\Delta^{+}{\subseteq}\Gamma^{+}$ and $\Gamma^{-}{\subseteq}\Delta^{-}$ then there exists a balanced clip $((\Lambda^{+},\Lambda^{-}),C^{+},
$\linebreak$
C^{-})$ such that ${\square}\Delta^{+}{\subseteq}\Lambda^{+}$, ${\lozenge}\Delta^{-}{\subseteq}\Lambda^{-}$, $B^{-}{\leftharpoondown}B^{+}{\in}\Lambda^{+}$, $B^{+}{\rightharpoonup}B^{-}{\in}\Lambda^{-}$ and $D{\in}\Lambda^{-}$.
\end{proposition}
\begin{proof}
Dual to the proof of Proposition~\ref{Proposition:let:be:a:balanced:clip:let:be:a:formula:if:then:for:all:balanced:clips:if:and:then:there:exists:a:balanced:clip:such:that:and:1}.
\medskip
\end{proof}
\section{Canonical Model Construction}\label{section:canonical:model:construction}
Let $W_{\mathbf{c}}$ be the nonempty set of all balanced clips.
\\
\\
Let $\leq_{\mathbf{c}}$ be the preorder on $W_{\mathbf{c}}$ such that for all $((\Gamma^{+},\Gamma^{-}),A^{+},A^{-}),((\Delta^{+},\Delta^{-}),B^{+},
$\linebreak$
B^{-}){\in}W_{\mathbf{c}}$,
\begin{itemize}
\item $((\Gamma^{+},\Gamma^{-}),A^{+},A^{-}){\leq_{\mathbf{c}}}((\Delta^{+},\Delta^{-}),B^{+},B^{-})$ if and only if $\Gamma^{+}{\subseteq}\Delta^{+}$ and $\Delta^{-}{\subseteq}
$\linebreak$
\Gamma^{-}$.
\end{itemize}
Let $R_{\mathbf{c}}$ be the binary relation on $W_{\mathbf{c}}$ such that for all $((\Gamma^{+},\Gamma^{-}),A^{+},A^{-}),((\Delta^{+},\Delta^{-}),
$\linebreak$
B^{+},B^{-}){\in}W_{\mathbf{c}}$,
\begin{itemize}
\item $((\Gamma^{+},\Gamma^{-}),A^{+},A^{-}){R_{\mathbf{c}}}((\Delta^{+},\Delta^{-}),B^{+},B^{-})$ if and only if ${\square}\Gamma^{+}{\subseteq}\Delta^{+}$, ${\lozenge}\Gamma^{-}{\subseteq}
$\linebreak$
\Delta^{-}$, $A^{-}{\leftharpoondown}A^{+}{\in}\Delta^{+}$ and $A^{+}{\rightharpoonup}A^{-}{\in}\Delta^{-}$.
\end{itemize}
The frame $(W_{\mathbf{c}},{\leq_{\mathbf{c}}},{R_{\mathbf{c}}})$ is called {\em canonical frame of $\L$.}
%
%
\\
\\
Notice that for all atoms $p$, $\{((\Gamma^{+},\Gamma^{-}),A^{+},A^{-}){\in}W_{\mathbf{c}}\ :\ p{\in}\Gamma^{+}\}$ is $\leq_{\mathbf{c}}$-closed.
\\
\\
Let $V_{\mathbf{c}}$ be the valuation on $(W_{\mathbf{c}},{\leq_{\mathbf{c}}})$ such that for all atoms $p$,
\begin{itemize}
\item $V_{\mathbf{c}}(p){=}\{((\Gamma^{+},\Gamma^{-}),A^{+},A^{-}){\in}W_{\mathbf{c}}\ :\ p{\in}\Gamma^{+}\}$.
\end{itemize}
The valuation $V_{\mathbf{c}}$ on $(W_{\mathbf{c}},{\leq_{\mathbf{c}}})$ is called {\em canonical valuation of $\L$.}
\\
\\
The model $(W_{\mathbf{c}},{\leq_{\mathbf{c}}},{R_{\mathbf{c}}},V_{\mathbf{c}})$ is called {\em canonical model of $\L$.}
\begin{proposition}\label{Proposition:let:be:a:formula:for:all}
Let $E$ be a formula.
In $(W_{\mathbf{c}},{\leq_{\mathbf{c}}},{R_{\mathbf{c}}},V_{\mathbf{c}})$, for all $((\Gamma^{+},\Gamma^{-}),A^{+},A^{-})
$\linebreak$
{\in}W_{\mathbf{c}}$,
\begin{itemize}
\item if $E{\in}\Gamma^{+}$ then $((\Gamma^{+},\Gamma^{-}),A^{+},A^{-}){\models}E$,
\item if $E{\in}\Gamma^{-}$ then $((\Gamma^{+},\Gamma^{-}),A^{+},A^{-}){\not\models}E$.
\end{itemize}
\end{proposition}
\begin{proof}
By induction on $E$, using Proposition~\ref{Proposition:let:be:a:balanced:clip:let:be:formulas:if:then:there:exists:a:balanced:clip:such:that:and:1}--\ref{Proposition:let:be:a:balanced:clip:let:be:a:formula:if:then:for:all:balanced:clips:if:and:then:there:exists:a:balanced:clip:such:that:and:2} in cases when $E{=}E^{\prime}{\rightharpoonup}E^{\prime\prime}$, $E{=}E^{\prime}
$\linebreak$
{\leftharpoondown}E^{\prime\prime}$, $E{=}{\lozenge}E^{\prime}$, $E{=}{\square}E^{\prime}$, $E{=}{\blacklozenge}E^{\prime}$ and $E{=}{\blacksquare}E^{\prime}$.
\medskip
\end{proof}
\begin{proposition}\label{Proposition:the:following:conditions:are:equivalent:is:reflexive}
If ${\blacksquare}p{\rightharpoonup}p{\in}\L^{+}$ and ${\blacklozenge}p{\leftharpoondown}p{\in}\L^{-}$ then $(W_{\mathbf{c}},{\leq_{\mathbf{c}}},{R_{\mathbf{c}}})$ is reflexive.
\end{proposition}
\begin{proof}
Suppose ${\blacksquare}p{\rightharpoonup}p{\in}\L^{+}$ and ${\blacklozenge}p{\leftharpoondown}p{\in}\L^{-}$.
For the sake of the contradiction, suppose $(W_{\mathbf{c}},{\leq_{\mathbf{c}}},{R_{\mathbf{c}}})$ is not reflexive.
Hence, there exists a balanced clip $((\Delta^{+},\Delta^{-}),
$\linebreak$
B^{+},B^{-})$ such that not $((\Delta^{+},\Delta^{-}),B^{+},B^{-}){R_{\mathbf{c}}}((\Delta^{+},\Delta^{-}),B^{+},B^{-})$.
Thus, either ${\square}\Delta^{+}{\not\subseteq}\Delta^{+}$, or ${\lozenge}\Delta^{-}{\not\subseteq}\Delta^{-}$, or $B^{-}{\leftharpoondown}B^{+}{\not\in}\Delta^{+}$, or $B^{+}{\rightharpoonup}B^{-}{\not\in}\Delta^{-}$.
\\
\\
In the first case, there exists a formula $A$ such that ${\square}A{\in}\Delta^{+}$ and $A{\not\in}\Delta^{+}$.
Since ${\blacksquare}p{\rightharpoonup}p{\in}\L^{+}$, therefore ${\blacksquare}A{\rightharpoonup}A{\in}\L^{+}$.
Moreover, since ${\square}A{\rightharpoonup}{\blacksquare}A{\in}\L^{+}$ using $(\Formula9^{+})$, therefore ${\blacksquare}A{\in}\Delta^{+}$.
Since ${\blacksquare}A{\rightharpoonup}A{\in}\L^{+}$, therefore $A{\in}\Delta^{+}$: a contradiction.
%
%
\\
\\
In the second case, a dual reasoning can be done.
\\
\\
In the third case, $B^{-}{\leftharpoondown}B^{+}{\in}\Delta^{-}$.
Since $(B^{-}{\leftharpoondown}B^{+}){\leftharpoondown}(B^{-}{\leftharpoondown}B^{+}){\wedge}{\top}{\in}\L^{-}$ using
\linebreak$
(\Formula13^{-})$, therefore $(B^{-}{\leftharpoondown}B^{+}){\wedge}{\top}{\in}\Delta^{-}$.
Since $(B^{-}{\leftharpoondown}B^{+}){\wedge}{\top}{\leftharpoondown}((B^{+}{\rightharpoonup}B^{-}){\leftharpoondown}(B^{-}
$\linebreak$
{\leftharpoondown}B^{+}){\wedge}{\top}){\in}\L^{-}$ using $(\Formula2^{-})$, therefore $(B^{+}{\rightharpoonup}B^{-}){\leftharpoondown}(B^{-}{\leftharpoondown}B^{+}){\wedge}{\top}{\in}\Delta^{-}$.
Since ${\blacksquare}p{\rightharpoonup}p{\in}\L^{+}$, therefore ${\blacksquare}((B^{+}{\rightharpoonup}B^{-}){\leftharpoondown}(B^{-}{\leftharpoondown}B^{+}){\wedge}{\top}){\rightharpoonup}((B^{+}{\rightharpoonup}B^{-}){\leftharpoondown}(B^{-}{\leftharpoondown}B^{+})
$\linebreak$
{\wedge}{\top}){\in}\L^{+}$.
Consequently, $((B^{+}{\rightharpoonup}B^{-}){\leftharpoondown}(B^{-}{\leftharpoondown}B^{+}){\wedge}{\top}){\leftharpoondown}{\blacksquare}((B^{+}{\rightharpoonup}B^{-}){\leftharpoondown}(B^{-}{\leftharpoondown}
$\linebreak$
B^{+}){\wedge}{\top}){\in}\L^{-}$ using $(\Rule2^{-})$.
Since $(B^{+}{\rightharpoonup}B^{-}){\leftharpoondown}(B^{-}{\leftharpoondown}B^{+}){\wedge}{\top}{\in}\Delta^{-}$, therefore
\linebreak$
{\blacksquare}((B^{+}{\rightharpoonup}B^{-}){\leftharpoondown}(B^{-}{\leftharpoondown}B^{+}){\wedge}{\top}){\in}\Delta^{-}$.
Hence, $((\Delta^{+},\Delta^{-}),B^{+},B^{-})$ being a balan\-ced clip, ${\square}{\top}{\in}\Delta^{-}$.
Since ${\square}{\top}{\leftharpoondown}{\top}{\in}\L^{-}$ using $(\Formula7^{-})$, therefore ${\top}{\in}\Delta^{-}$.
Since ${\top}{\in}\L^{+}$ using $(\Formula5^{+})$, therefore ${\top}{\in}\Delta^{+}$.
Thus, ${\top}{\not\in}\Delta^{-}$: a contradiction.
%
%
\\
\\
In the fourth case, a dual reasoning can be done.
\medskip
\end{proof}
\begin{proposition}\label{Proposition:the:following:conditions:are:equivalent:is:serial}
If either ${\blacklozenge}{\top}{\in}\L^{+}$, or ${\blacksquare}{\bot}{\in}\L^{-}$ then $(W_{\mathbf{c}},{\leq_{\mathbf{c}}},{R_{\mathbf{c}}})$ is serial.
\end{proposition}
\begin{proof}
%
%
%
%
%
%
%
By Proposition~\ref{Proposition:let:be:a:balanced:clip:let:be:a:formula:if:then:for:all:balanced:clips:if:and:then:there:exists:a:balanced:clip:such:that:and:1} and~\ref{Proposition:let:be:a:balanced:clip:let:be:a:formula:if:then:for:all:balanced:clips:if:and:then:there:exists:a:balanced:clip:such:that:and:2}.
%
%
%
%
%
%
%
%
%
%
%
%
\medskip
\end{proof}
\begin{proposition}\label{Proposition:the:following:conditions:are:equivalent:is:forward:confluent}
\begin{enumerate}
\item If ${\square}(p{\rightharpoonup}q){\rightharpoonup}({\lozenge}p{\rightharpoonup}{\blacklozenge}q){\in}\L^{+}$ then $(W_{\mathbf{c}},{\leq_{\mathbf{c}}},{R_{\mathbf{c}}})$ is forward confluent,
\item if ${\lozenge}(p{\leftharpoondown}q){\leftharpoondown}({\square}p{\leftharpoondown}{\blacksquare}q){\in}\L^{-}$ then $(W_{\mathbf{c}},{\leq_{\mathbf{c}}},{R_{\mathbf{c}}})$ is downward confluent.
\end{enumerate}
\end{proposition}
\begin{proof}
$\mathbf{1)}$ Suppose ${\square}(p{\rightharpoonup}q){\rightharpoonup}({\lozenge}p{\rightharpoonup}{\blacklozenge}q){\in}\L^{+}$.
Let $((\Gamma^{+},\Gamma^{-}),A^{+},A^{-})$, $((\Delta^{+},\Delta^{-}),
$\linebreak$
B^{+},B^{-})$ and $((\Lambda^{+},\Lambda^{-}),C^{+},C^{-})$ be balanced clips such that $((\Gamma^{+},\Gamma^{-}),A^{+},A^{-}){\geq_{c}}
$\linebreak$
((\Delta^{+},\Delta^{-}),B^{+},B^{-})$ and $((\Delta^{+},\Delta^{-}),B^{+},B^{-}){R_{c}}((\Lambda^{+},\Lambda^{-}),C^{+},C^{-})$.
Hence, $\Gamma^{+}
$\linebreak$
{\supseteq}\Delta^{+}$ and ${\lozenge}\Delta^{-}{\subseteq}\Lambda^{-}$.
Let ${\mathcal S}$ be the set of all filters $\Theta$ such that $({\square}\Gamma^{+}{+}\Lambda^{+}){+}(A^{-}{\leftharpoondown}
$\linebreak$
A^{+}){\subseteq}\Theta$ and for all $E{\in}\Fo$, if $(A^{+}{\rightharpoonup}A^{-}){\vee}E{\in}\Theta$ then ${\lozenge}E{\in}\Gamma^{+}$.
\begin{claim}
$({\square}\Gamma^{+}{+}\Lambda^{+}){+}(A^{-}{\leftharpoondown}A^{+})$ is in ${\mathcal S}$.
\end{claim}
\begin{proof}
For the sake of the contradiction, suppose $({\square}\Gamma^{+}{+}\Lambda^{+}){+}(A^{-}{\leftharpoondown}A^{+})$ is not in ${\mathcal S}$.
Hence, there exists $E{\in}\Fo$ such that $(A^{+}{\rightharpoonup}A^{-}){\vee}E{\in}({\square}\Gamma^{+}{+}\Lambda^{+}){+}(A^{-}{\leftharpoondown}A^{+})$ and ${\lozenge}E{\not\in}\Gamma^{+}$.
Thus, there exists $n{\in}\N$ and there exists $F_{1},\ldots,F_{n}{\in}\Lambda^{+}$ such that ${\square}(F_{1}{\wedge}\ldots
$\linebreak$
{\wedge}F_{n}{\rightharpoonup}((A^{-}{\leftharpoondown}A^{+}){\rightharpoonup}(A^{+}{\rightharpoonup}A^{-}){\vee}E)){\in}\Gamma^{+}$.
Since $(F_{1}{\rightharpoonup}\ldots(F_{n}{\rightharpoonup}F_{1}{\wedge}\ldots{\wedge}F_{n})\ldots)
$\linebreak$
{\in}\L^{+}$ using $(\Formula15^{+})$ and $F_{1},\ldots,F_{n}{\in}\Lambda^{+}$, therefore $F_{1}{\wedge}\ldots{\wedge}F_{n}{\in}\Lambda^{+}$.
Since $\Gamma^{+}{\supseteq}\Delta^{+}$ and ${\lozenge}\Delta^{-}{\subseteq}\Lambda^{-}$, therefore ${\lozenge}(F_{1}{\wedge}\ldots{\wedge}F_{n}){\in}\Gamma^{+}$.
Since ${\square}(p{\rightharpoonup}q){\rightharpoonup}({\lozenge}p{\rightharpoonup}{\blacklozenge}q){\in}\L^{+}$, therefore ${\square}(F_{1}{\wedge}\ldots{\wedge}F_{n}{\rightharpoonup}((A^{-}{\leftharpoondown}A^{+}){\rightharpoonup}(A^{+}{\rightharpoonup}A^{-}){\vee}E)){\rightharpoonup}({\lozenge}(F_{1}{\wedge}\ldots{\wedge}F_{n}){\rightharpoonup}
$\linebreak$
{\blacklozenge}((A^{-}{\leftharpoondown}A^{+}){\rightharpoonup}(A^{+}{\rightharpoonup}A^{-}){\vee}E)){\in}\L^{+}$.
Since ${\square}(F_{1}{\wedge}\ldots{\wedge}F_{n}{\rightharpoonup}((A^{-}{\leftharpoondown}A^{+}){\rightharpoonup}(A^{+}
$\linebreak$
{\rightharpoonup}A^{-}){\vee}E)){\in}\Gamma^{+}$ and ${\lozenge}(F_{1}{\wedge}\ldots{\wedge}F_{n}){\in}\Gamma^{+}$, therefore ${\blacklozenge}((A^{-}{\leftharpoondown}A^{+}){\rightharpoonup}(A^{+}{\rightharpoonup}A^{-}){\vee}
$\linebreak$
E){\in}\Gamma^{+}$.
Consequently, $((\Gamma^{+},\Gamma^{-}),A^{+},A^{-})$ being a balanced clip, ${\lozenge}E{\in}\Gamma^{+}$: a contradiction.
\medskip
\end{proof}
Moreover, for all nonempty chains $(\Theta_{i})_{i{\in}I}$ in ${\mathcal S}$, $\bigcup\{\Theta_{i}\ :\ i{\in}I\}$ is in ${\mathcal S}$.
Thus, by Kuratowski-Zorn Lemma, ${\mathcal S}$ possesses a maximal element $\Theta$.
It is a routine task to show that $\Theta$ is a prime filter such that $({\square}\Gamma^{+}{+}\Lambda^{+}){+}(A^{-}{\leftharpoondown}A^{+}){\subseteq}\Theta$ and for all $E{\in}\Fo$, if $(A^{+}{\rightharpoonup}A^{-}){\vee}E{\in}\Theta$ then ${\lozenge}E{\in}\Gamma^{+}$.
And in the end, choose $\Theta^{+}{=}\Theta$ and $\Theta^{-}{=}\Fo{\setminus}\Theta$.
Consequently, ${\square}\Gamma^{+}{\subseteq}\Theta^{+}$, ${\lozenge}\Gamma^{-}{\subseteq}\Theta^{-}$, $A^{-}{\leftharpoondown}A^{+}{\in}\Theta^{+}$, $A^{+}{\rightharpoonup}A^{-}{\in}\Theta^{-}$, $\Lambda^{+}{\subseteq}\Theta^{+}$ and $\Theta^{-}{\subseteq}\Lambda^{-}$.
\\
\\
$\mathbf{2)}$ Dual to the proof of Item~$\mathbf{1)}$.
\medskip
\end{proof}
%
%
%
%
%
%
%
%
%
%
%
%
%
%
Now, we are ready to establish the following completeness results.
\begin{proposition}\label{Proposition:let:be:a:class:of:frames:if:then:1}
$\Log({\mathcal C}_{\allfra}){=}\L_{\min}$.
\end{proposition}
\begin{proof}
By Proposition~\ref{proposition:soundness}, \ref{proposition:consistent:BML:a:property:about:tips}, \ref{Proposition:for:all:coherent:tips:there:exists:an:exhaustive:tip:such:that:and}, \ref{Proposition:for:all:exhaustive:tips:is:balanced} and~\ref{Proposition:let:be:a:formula:for:all}.
\medskip
\end{proof}
\begin{proposition}\label{Proposition:let:ba:a:class:of:frames:if:then:3}
$\Log({\mathcal C}_{\reflexive}){=}\L_{\min}{\odot}\{(+,{\blacksquare}p{\rightharpoonup}p),(-,{\blacklozenge}p{\leftharpoondown}p)\}$.
\end{proposition}
\begin{proof}
By Proposition~\ref{Proposition:about:correpsondence:table:a}, \ref{proposition:soundness}, \ref{proposition:consistent:BML:a:property:about:tips}, \ref{Proposition:for:all:coherent:tips:there:exists:an:exhaustive:tip:such:that:and}, \ref{Proposition:for:all:exhaustive:tips:is:balanced}, \ref{Proposition:let:be:a:formula:for:all} and~\ref{Proposition:the:following:conditions:are:equivalent:is:reflexive}.
%
%
\medskip
\end{proof}
\begin{proposition}\label{Proposition:let:ba:a:class:of:frames:if:then:2}
$\Log({\mathcal C}_{\serial}){=}\L_{\min}{\odot}(+,{\blacklozenge}{\top}){=}\L_{\min}{\odot}(-,{\blacksquare}{\bot}){=}\L_{\min}{\odot}\{(+,{\blacklozenge}{\top}),(-,
$\linebreak$
{\blacksquare}{\bot})\}$.
\end{proposition}
\begin{proof}
By Proposition~\ref{proposition:positive:definition:seriality}, \ref{proposition:soundness}, \ref{proposition:consistent:BML:a:property:about:tips}, \ref{Proposition:for:all:coherent:tips:there:exists:an:exhaustive:tip:such:that:and}, \ref{Proposition:for:all:exhaustive:tips:is:balanced}, \ref{Proposition:let:be:a:formula:for:all} and~\ref{Proposition:the:following:conditions:are:equivalent:is:serial}.
\medskip
\end{proof}
\begin{proposition}\label{completeness:forward:conf}
\begin{enumerate}
\item $\Log({\mathcal C}_{\fcfra}){=}\L_{\min}{\odot}(+,{\square}(p{\rightharpoonup}q){\rightharpoonup}({\lozenge}p{\rightharpoonup}{\blacklozenge}q))$,
\item $\Log({\mathcal C}_{\dcfra}){=}\L_{\min}{\odot}(-,{\lozenge}(p{\leftharpoondown}q){\leftharpoondown}({\square}p{\leftharpoondown}{\blacksquare}q))$.
\end{enumerate}
\end{proposition}
\begin{proof}
By Proposition~\ref{Proposition:forward:confluence:is:modally:definable}, \ref{proposition:soundness}, \ref{proposition:consistent:BML:a:property:about:tips}, \ref{Proposition:for:all:coherent:tips:there:exists:an:exhaustive:tip:such:that:and}, \ref{Proposition:for:all:exhaustive:tips:is:balanced}, \ref{Proposition:let:be:a:formula:for:all} and~\ref{Proposition:the:following:conditions:are:equivalent:is:forward:confluent}.
\medskip
\end{proof}
%
%
%
%
%
%
%
%
%
%
%
%
%
%
%
%
\section{Decidability}\label{section:decidability}
Let $\Mona$ be a countable set (with typical members called {\em monadic predicates} and denoted $P$, $Q$, etc).
\\
\\
Let $(P_{i})_{i{\in}\N}$ be an enumeration without repetition of $\Mona$.
\\
\\
Let $x,y$ be distinct individual variables.
\\
\\
Let the sets $\GFo_{x}$ and $\GFo_{y}$ (with typical members respectively called {\em guarded $x$-formu\-las}\/ and {\em guarded $y$-formulas}\/ and respectively denoted $A_{x}$, $B_{x}$, etc and $A_{y}$, $B_{y}$, etc) be defined by
\begin{itemize}
\item $A_{x}::=P_{i}(x){\mid}{\top}{\mid}{\bot}{\mid}{\neg}A_{x}{\mid}(A_{x}{\vee}A_{x}){\mid}(A_{x}{\wedge}A_{x}){\mid}{\forall y}(x{\leq}y{\rightarrow}A_{y}){\mid}{\exists y}(y{\leq}x{\wedge}A_{y}){\mid}
$\linebreak$
{\exists y}(R(x,y){\wedge}A_{y}){\mid}{\forall y}(R(x,y){\rightarrow}A_{y})$,
\item $A_{y}::=P_{i}(y){\mid}{\top}{\mid}{\bot}{\mid}{\neg}A_{y}{\mid}(A_{y}{\vee}A_{y}){\mid}(A_{y}{\wedge}A_{y}){\mid}{\forall x}(y{\leq}x{\rightarrow}A_{x}){\mid}{\exists x}(x{\leq}y{\wedge}A_{x}){\mid}
$\linebreak$
{\exists x}(R(y,x){\wedge}A_{x}){\mid}{\forall x}(R(y,x){\rightarrow}A_{x})$,
\end{itemize}
$i$ ranging over $\N$.
\\
\\
For all guarded first-order formulas $A$, the {\em length of $A$}\/ (denoted ${\parallel}A{\parallel}$) is the number of symbols in $A$.
\\
\\
We follow the standard rules for omission of the parentheses.
\\
\\
The guarded formulas of the form $P_{i}(x)$ and $P_{i}(y)$ are called {\em atomic guarded formulas.}
\begin{proposition}\label{Proposition:two:variable:fragment:only}
Guarded formulas belong to the monadic two-variable guarded fragment considered by Ganzinger {\em et al.}~\cite[Section~$3$]{Ganzinger:et:al:1999}.
\end{proposition}
The {\em satisfiability of a guarded formula $A$ in a frame $(W,{\leq},{R})$ with respect to a couple $(s,t)$ of elements in $W$ and a valuation $V$ on $(W,{\leq})$}\/ (in symbols $(W,{\leq},{R},V){\models_{x}}
$\linebreak$
A_{x}\ \lbrack s,t\rbrack$ for a guarded $x$-formula $A_{x}$ and $(W,{\leq},{R},V){\models_{y}}A_{y}\ \lbrack s,t\rbrack$ for a guarded $y$-formula $A_{y}$) is defined as follows:
%
%
\begin{itemize}
\item $(W,{\leq},{R},V){\models_{x}}P_{i}(x)\ \lbrack s,t\rbrack$ if and only if $s{\in}V(p_{i})$,
\item $(W,{\leq},{R},V){\models_{y}}P_{i}(y)\ \lbrack s,t\rbrack$ if and only if $t{\in}V(p_{i})$,
\item $(W,{\leq},{R},V){\models_{x}}\top\ \lbrack s,t\rbrack$,
\item $(W,{\leq},{R},V){\models_{y}}\top\ \lbrack s,t\rbrack$,
\item $(W,{\leq},{R},V){\not\models_{x}}\bot\ \lbrack s,t\rbrack$,
\item $(W,{\leq},{R},V){\not\models_{y}}\bot\ \lbrack s,t\rbrack$,
\item $(W,{\leq},{R},V){\models_{x}}{\neg}A_{x}\ \lbrack s,t\rbrack$ if and only if $(W,{\leq},{R},V){\not\models_{x}}A_{x}\ \lbrack s,t\rbrack$,
\item $(W,{\leq},{R},V){\models_{y}}{\neg}A_{y}\ \lbrack s,t\rbrack$ if and only if $(W,{\leq},{R},V){\not\models_{y}}A_{y}\ \lbrack s,t\rbrack$,
\item $(W,{\leq},{R},V){\models_{x}}A_{x}{\vee}B_{x}\ \lbrack s,t\rbrack$ if and only if either $(W,{\leq},{R},V){\models_{x}}A_{x}\ \lbrack s,t\rbrack$, or $(W,{\leq},{R},V){\models_{x}}B_{x}\ \lbrack s,t\rbrack$,
\item $(W,{\leq},{R},V){\models_{y}}A_{y}{\vee}B_{y}\ \lbrack s,t\rbrack$ if and only if either $(W,{\leq},{R},V){\models_{y}}A_{y}\ \lbrack s,t\rbrack$, or $(W,{\leq},{R},V){\models_{y}}B_{y}\ \lbrack s,t\rbrack$,
\item $(W,{\leq},{R},V){\models_{x}}A_{x}{\wedge}B_{x}\ \lbrack s,t\rbrack$ if and only if $(W,{\leq},{R},V){\models_{x}}A_{x}\ \lbrack s,t\rbrack$ and $(W,
$\linebreak$
{\leq},{R},V){\models_{x}}B_{x}\ \lbrack s,t\rbrack$,
\item $(W,{\leq},{R},V){\models_{y}}A_{y}{\wedge}B_{y}\ \lbrack s,t\rbrack$ if and only if $(W,{\leq},{R},V){\models_{y}}A_{y}\ \lbrack s,t\rbrack$ and $(W,{\leq},
$\linebreak$
{R},V){\models_{y}}B_{y}\ \lbrack s,t\rbrack$,
\item $(W,{\leq},{R},V){\models_{x}}{\forall y}(x{\leq}y{\rightarrow}A_{y})\ \lbrack s,t\rbrack$ if and only if for all $u{\in}W$, if $s{\leq}u$ then $(W,{\leq},{R},V){\models_{y}}A_{y}\ \lbrack s,u\rbrack$,
\item $(W,{\leq},{R},V){\models_{y}}{\forall x}(y{\leq}x{\rightarrow}A_{x})\ \lbrack s,t\rbrack$ if and only if for all $u{\in}W$, if $t{\leq}u$ then $(W,{\leq},{R},V){\models_{x}}A_{x}\ \lbrack u,t\rbrack$,
\item $(W,{\leq},{R},V){\models_{x}}{\exists y}(y{\leq}x{\wedge}A_{y})\ \lbrack s,t\rbrack$ if and only if there exists $u{\in}W$ such that $u{\leq}s$ and $(W,{\leq},{R},V){\models_{y}}A_{y}\ \lbrack s,u\rbrack$,
\item $(W,{\leq},{R},V){\models_{y}}{\exists x}(x{\leq}y{\wedge}A_{x})\ \lbrack s,t\rbrack$ if and only if there exists $u{\in}W$ such that $u{\leq}t$ and $(W,{\leq},{R},V){\models_{x}}A_{x}\ \lbrack u,t\rbrack$,
\item $(W,{\leq},{R},V){\models_{x}}{\exists y}(R(x,y){\wedge}A_{y})\ \lbrack s,t\rbrack$ if and only if there exists $u{\in}W$ such that $s{R}u$ and $(W,{\leq},{R},V){\models_{y}}A_{y}\ \lbrack s,u\rbrack$,
\item $(W,{\leq},{R},V){\models_{y}}{\exists x}(R(y,x){\wedge}A_{x})\ \lbrack s,t\rbrack$ if and only if there exists $u{\in}W$ such that $t{R}u$ and $(W,{\leq},{R},V){\models_{x}}A_{x}\ \lbrack u,t\rbrack$,
\item $(W,{\leq},{R},V){\models_{x}}{\forall y}(R(x,y){\rightarrow}A_{y})\ \lbrack s,t\rbrack$ if and only if for all $u{\in}W$, if $s{R}u$ then $(W,{\leq},{R},V){\models_{y}}A_{y}\ \lbrack s,u\rbrack$,
\item $(W,{\leq},{R},V){\models_{y}}{\forall x}(R(y,x){\rightarrow}A_{x})\ \lbrack s,t\rbrack$ if and only if for all $u{\in}W$, if $t{R}u$ then $(W,{\leq},{R},V){\models_{x}}A_{x}\ \lbrack u,t\rbrack$.
\end{itemize}
A guarded formula $A$ is {\em valid in a frame $(W,{\leq},{R})$}\/ (in symbols $(W,{\leq},{R}){\models_{x}}A_{x}$ for a guarded $x$-formula $A_{x}$ and $(W,{\leq},{R}){\models_{y}}A_{y}$ for a guarded $y$-formula $A_{y}$) if $A$ is satisfied in $(W,{\leq},{R})$ with respect to all couples $(s,t)$ of elements in $W$ and all valuations $V$ on $(W,{\leq})$.
\\
\\
A guarded formula $A$ is {\em valid on a class ${\mathcal C}$ of frames}\/ (in symbols ${\mathcal C}{\models_{x}}A_{x}$ for a guarded $x$-formula $A_{x}$ and ${\mathcal C}{\models_{y}}A_{y}$ for a guarded $y$-formula $A_{y}$) if $A$ is valid in all frames in ${\mathcal C}$.
\\
\\
Let $f_{x}\ :\ \Fo{\longrightarrow}\GFo_{x}$ and $f_{y}\ :\ \Fo{\longrightarrow}\GFo_{y}$ be the functions defined as follows:
\begin{itemize}
\item $f_{x}(p_{i}){=}P_{i}(x)$,
\item $f_{y}(p_{i}){=}P_{i}(y)$,
\item $f_{x}(A{\rightharpoonup}B){=}{\forall y}(x{\leq}y{\rightarrow}(\neg f_{y}(A){\vee}f_{y}(B)))$,
\item $f_{y}(A{\rightharpoonup}B){=}{\forall x}(y{\leq}x{\rightarrow}(\neg f_{x}(A){\vee}f_{x}(B)))$,
\item $f_{x}(A{\leftharpoondown}B){=}{\exists y}(y{\leq}x{\wedge}(\neg f_{y}(A){\wedge}f_{y}(B)))$,
\item $f_{y}(A{\leftharpoondown}B){=}{\exists x}(x{\leq}y{\wedge}(\neg f_{x}(A){\wedge}f_{x}(B)))$,
\item $f_{x}({\top}){=}{\top}$,
\item $f_{y}({\top}){=}{\top}$,
\item $f_{x}({\bot}){=}{\bot}$,
\item $f_{y}({\bot}){=}{\bot}$,
\item $f_{x}(A{\vee}B){=}f_{x}(A){\vee}f_{x}(B)$,
\item $f_{y}(A{\vee}B){=}f_{y}(A){\vee}f_{y}(B)$,
\item $f_{x}(A{\wedge}B){=}f_{x}(A){\wedge}f_{x}(B)$,
\item $f_{y}(A{\wedge}B){=}f_{y}(A){\wedge}f_{y}(B)$,
\item $f_{x}({\lozenge}A){=}{\exists y}(y{\leq}x{\wedge}{\exists x}(R(y,x){\wedge}f_{x}(A)))$,
\item $f_{y}({\lozenge}A){=}{\exists x}(x{\leq}y{\wedge}{\exists y}(R(x,y){\wedge}f_{y}(A)))$,
\item $f_{x}({\square}A){=}{\forall y}(x{\leq}y{\rightarrow}{\forall x}(R(y,x){\rightarrow}f_{x}(A)))$,
\item $f_{y}({\square}A){=}{\forall x}(y{\leq}x{\rightarrow}{\forall y}(R(x,y){\rightarrow}f_{y}(A)))$,
\item $f_{x}({\blacklozenge}A){=}{\forall y}(x{\leq}y{\rightarrow}{\exists x}(R(y,x){\wedge}f_{x}(A)))$,
\item $f_{y}({\blacklozenge}A){=}{\forall x}(y{\leq}x{\rightarrow}{\exists y}(R(x,y){\wedge}f_{y}(A)))$,
\item $f_{x}({\blacksquare}A){=}{\exists y}(y{\leq}x{\wedge}{\forall x}(R(y,x){\rightarrow}f_{x}(A)))$,
\item $f_{y}({\blacksquare}A){=}{\exists x}(x{\leq}y{\wedge}{\forall y}(R(x,y){\rightarrow}f_{x}(A)))$.
\end{itemize}
\begin{proposition}\label{Proposition:about:translations:tau:x:tau:y:linear:size}
For all $A{\in}\Fo$, ${\parallel}f_{x}(A){\parallel}{\leq}19{\times}{\parallel}A{\parallel}$ and ${\parallel}f_{y}(A){\parallel}{\leq}19{\times}{\parallel}A{\parallel}$.
\end{proposition}
\begin{proof}
By induction on ${\parallel}A{\parallel}$.
\end{proof}
\begin{proposition}\label{Proposition:about:translations:tau:x:tau:y:soundness:correctness}
Let $(W,{\leq},{R})$ be a frame and $V$ be a valuation on $(W,{\leq})$.
For all $A{\in}\Fo$, the following conditions are equivalent for all $s$ in $W$:
\begin{itemize}
\item $(W,{\leq},{R},V),s{\models}A$,
\item there exists $t{\in}W$ such that $(W,{\leq},{R},V){\models_{x}}f_{x}(A)\ \lbrack s,t\rbrack$,
\item for all $t{\in}W$, $(W,{\leq},{R},V){\models_{x}}f_{x}(A)\ \lbrack s,t\rbrack$,
\end{itemize}
and the following conditions are equivalent for all $t$ in $W$:
\begin{itemize}
\item $(W,{\leq},{R},V),t{\models}A$,
\item there exists $s{\in}W$ such that $(W,{\leq},{R},V){\models_{y}}f_{y}(A)\ \lbrack s,t\rbrack$,
\item for all $s{\in}W$, $(W,{\leq},{R},V){\models_{y}}f_{y}(A)\ \lbrack s,t\rbrack$.
\end{itemize}
\end{proposition}
\begin{proof}
By induction on $A$.
\medskip
\end{proof}
The {\em membership problem in $\L$}\/ is the following decision problem: determine whether a given signed formula is in $\L$.
\begin{proposition}\label{proposition:two:variable:guarded:fragment:is:decidable}
The membership problem in $\L_{\min}$ is decidable.
\end{proposition}
\begin{proof}
By~\cite[Section~$3$]{Ganzinger:et:al:1999} and Proposition~\ref{Proposition:two:variable:fragment:only} and~\ref{Proposition:about:translations:tau:x:tau:y:soundness:correctness}.
\medskip
\end{proof}
\begin{proposition}\label{proposition:two:variable:guarded:fragment:is:decidable:ser}
The membership problem in $\L_{\min}{\odot}(+,{\blacklozenge}{\top})$, $\L_{\min}{\odot}(-,{\blacksquare}{\bot})$ and
\linebreak$
\L_{\min}{\odot}\{(+,{\blacklozenge}{\top}),(-,{\blacksquare}{\bot})\}$ is decidable.
\end{proposition}
\begin{proposition}\label{proposition:two:variable:guarded:fragment:is:decidable:ref}
The membership problem in $\L_{\min}{\odot}\{(+,{\blacksquare}p{\rightharpoonup}p),(-,{\blacklozenge}p{\leftharpoondown}p)\}$ is decidable.
\end{proposition}
\section{Conclusion}\label{section:conclusion}
Much remains to be done.
\\
\\
Inference rules $(\Rule8^{+})$ and $(\Rule8^{-})$ have ben used in the proofs of Proposition~\ref{Proposition:let:be:a:balanced:clip:let:be:a:formula:if:then:for:all:balanced:clips:if:and:then:there:exists:a:balanced:clip:such:that:and:1:lozenge:Gamma:plus} and~\ref{Proposition:let:be:a:balanced:clip:let:be:a:formula:if:then:for:all:balanced:clips:if:and:then:there:exists:a:balanced:clip:such:that:and:2:square:Gamma:moins}.
However, we do not know of any formula whose formal proof in the IMLs considered in this article requires the use of these inference rules.
Therefore, a natural question is to determine if one can replace inference rules $(\Rule8^{+})$ and $(\Rule8^{-})$ by finitely many signed formulas considered as additional axioms without affecting the meaning of the definition of IMLs.
\\
\\
%
%
We have seen in Proposition~\ref{Proposition:let:ba:a:class:of:frames:if:then:3} and~\ref{Proposition:let:ba:a:class:of:frames:if:then:2} that $\Log({\mathcal C}_{\reflexive})$ and $\Log({\mathcal C}_{\serial})$ are finitely axiomatizable.
Therefore, a natural question is to determine whether $\Log({\mathcal C}_{\sym})$ and $\Log({\mathcal C}_{\transitive})$ are finitely axiomatizable too.
In this respect, it might be helpful to understand how far is $\L_{\min}{\odot}\{(+,p{\rightharpoonup}{\square}{\lozenge}p),(-,p{\leftharpoondown}{\lozenge}{\square}p))\}$ from $\Log({\mathcal C}_{\sym})$ and how far is $\L_{\min}{\odot}\{(+,{\square}p{\rightharpoonup}{\square}{\square}p),(-,{\lozenge}p{\leftharpoondown}{\lozenge}{\lozenge}p))\}$ from $\Log({\mathcal C}_{\transitive})$.
\\
\\
We have seen in Proposition~\ref{completeness:forward:conf} that $\Log({\mathcal C}_{\fcfra})$ and $\Log({\mathcal C}_{\dcfra})$ are finitely axiomatizable.
Therefore, a natural question is to determine whether $\Log({\mathcal C}_{\bcfra})$ and $\Log({\mathcal C}_{\ucfra})$~---~as well as $\Log({\mathcal C}_{\fbcfra})$, $\Log({\mathcal C}_{\fdcfra})$, etc~---~are finitely axiomatizable too.
Notice that 
by Proposition~\ref{Proposition:forward:confluence:is:modally:definable}, \ref{proposition:soundness}, \ref{proposition:consistent:BML:a:property:about:tips}, \ref{Proposition:for:all:coherent:tips:there:exists:an:exhaustive:tip:such:that:and}, \ref{Proposition:for:all:exhaustive:tips:is:balanced}, \ref{Proposition:let:be:a:formula:for:all} and~\ref{Proposition:the:following:conditions:are:equivalent:is:forward:confluent}, we already know that $\Log({\mathcal C}_{\fdcfra}){=}\L_{\min}{\odot}\{(+,{\square}(p{\rightharpoonup}
$\linebreak$
q){\rightharpoonup}({\lozenge}p{\rightharpoonup}{\blacklozenge}q)),(-,{\lozenge}(p{\leftharpoondown}q){\leftharpoondown}({\square}p{\leftharpoondown}{\blacksquare}q))\}$.
%
%
\\
\\
%
%
%
%
%
%
We have seen in Proposition~\ref{proposition:two:variable:guarded:fragment:is:decidable} that the membership problem in $\L_{\min}$ is decidable.
Therefore, a natural question is to determine the complexity class to which this problem belongs.
In this respect, it might be helpful to use an embedding of $\L_{\min}$ into $\S4_{\mathbf{t}}{\otimes}\K$~---~the fusion of tense $\S4$ and $\K$~---~in the spirit of the embedding of IMLs into modal logics containing $\S4{\otimes}\K$~---~the fusion of $\S4$ and $\K$~---~proposed by Wolter and Zakharyaschev~\cite{Wolter:Zakharyaschev:1997}.
\\
\\
%
%
With a few exceptions such as~\cite{Hasimoto:2001,Takano:2003}, the technique of the filtration has not been so much adapted to IMLs, probably because it does not easily work with conditions such as forward confluence, backward confluence, downward confluence and upward confluence.
Therefore, a natural question is to understand how filtration-like arguments can be used in order to determine whether the membership problems in the IMLs considered in this article are decidable.
\section*{Acknowledgements}
We wish to thank our colleagues of the {\em Institut de recherche en informatique de Toulou\-se}\/ for many stimulating discussions about intuitionistic modal logics.
\bibliographystyle{named}

\begin{thebibliography}{}
%
%
%
%
\bibitem{Alechina:et:al:2001}
Alechina, N., Mendler, M., de Paiva, V., Ritter, E.:
{\it Categorical and Kripke semantics for constructive $\S4$ modal logic.}
In {\it CSL 2001.}
Springer (2001) 292--307.
%
%
%
%
%
%
\bibitem{Arisaka:et:al:2015}
Arisaka, R., Das, A., Stra\ss burger, L.:
{\it On nested sequents for constructive modal logics.}
Logical Methods in Computer Science {\bf 11} (2015) 1--33.
%
%
%
%
%
%
%
%
%
%
%
%
%
%
%
%
%
%
%
%
%
%
%
%
\bibitem{Balbiani:et:al:2024a}
Balbiani, P., Gao, H., Gencer, \c{C}., Olivetti, N.:
{\it A natural intuitionistic modal logic: axiomatization and bi-nested calculus.}
In {\it 32nd EACSL Annual Conference on Computer Science Logic.}
LIPICS (2024) 13:1--13:21.
%
%
\bibitem{Balbiani:et:al:2024b}
Balbiani, P., Gao, H., Gencer, \c{C}., Olivetti, N.:
{\it Local intuitionistic modal logics and their calculi.}
In {\it Automated Reasoning.}
Springer (2024) 78--96.
%
%
\bibitem{Balbiani:Gencer:preliminary:draft:SL}
Balbiani, P., Gencer, \c{C}.:
{\it Intuitionistic modal logics: a minimal setting.}
Studia Logica (2026) doi.org/10.1007/s11225-025-10224-7.
%
%
%
%
%
%
\bibitem{Balbiani:Tinchev:2017}
Balbiani, P., Tinchev, T.:
{\it Undecidable problems for modal definability.}
Journal of Logic and Computation {\bf 27} (2017) 901--920.
%
%
%
%
\bibitem{Ballarin:2023}
Ballarin, R.:
{\it Modern origins of modal logic.}
In {\it The Stanford Encyclopedia of Philosophy (Fall 2023 Edition).}
Metaphysics Research Lab (2023) plato.stanford.edu/archives/fall2023/entries/logic-modal-origins/.
%
%
%
%
\bibitem{Bierman:dePaiva:2000}
Bierman, G., de Paiva, V.:
{\it On an intuitionistic modal logic.}
Studia Logica {\bf 65} (2000) 383--416.
%
%
%
%
\bibitem{Bozic:Dosen:1984}
Bo\v{z}i\'{c}, M., Do\v{s}en, K.:
{\it Models for normal intuitionistic modal logics.}
Studia Logica {\bf 43} (1984) 217--245.
%
%
%
%
\bibitem{Chagrov:Zakharyaschev:1997}
Chagrov, A., Zakharyaschev, M.:
{\it Modal Logic.}
Oxford University Press (1997).
%
%
%
%
%
%
%
%
\bibitem{Dalmonte:et:al:2021}
Dalmonte, T., Grellois, C., Olivetti, N.:
{\it Terminating calculi and countermodels for constructive modal logics.}
In {\it Automated Reasoning with Analytic Tableaux and Related Methods.}
Springer (2021) 391--408.
%
%
%
%
\bibitem{Davey:Priestley:2002}
Davey, B., Priestley, H.:
{\it Introduction to Lattices and Order.}
Cambridge University Press (2002).
%
%
%
%
%
%
%
%
%
%
%
%
%
%
%
%
%
%
%
%
%
%
%
%
%
%
%
%
%
%
%
%
%
%
%
%
\bibitem{FischerServi:1984}
Fischer Servi, G.:
{\it Axiomatizations for some intuitionistic modal logics.}
Rendiconti del Seminario Matematico Universit\`a e Politecnico di Torino {\bf 42} (1984) 179--194.
%
%
%
%
%
%
\bibitem{Gabbay:et:al:2009}
Gabbay, D., Shehtman, V., Skvortsov, D.:
{\it Quantification in Nonclassical Logic. Volume~$1$.}
Elsevier (2009).
%
%
%
%
\bibitem{Ganzinger:et:al:1999}
Ganzinger, H., Meyer, C., Veanes, M.:
{\it The two-variable guarded fragment with transitive relations.}
In {\it Fourteenth Annual IEEE Symposium on Logic in Computer Science.}
IEEE (1999) 24--34.
%
%
%
%
%
%
%
%
%
%
\bibitem{Girlando:et:al:2023}
Girlando, M., Kuznets, R., Marin, S., Morales, M., Stra\ss burger, L.:
{\it Intuitionistic $\S4$ is decidable.}
In {\it 38th Annual ACM/IEEE Symposium on Logic in Computer Science.}
IEEE (2023) 10.1109/LICS56636.2023.10175684.
%
%
\bibitem{Girlando:et:al:2024}
Girlando, M., Kuznets, R., Marin, S., Morales, M., Stra\ss burger, L.:
{\it A simple loopcheck for intuitionistic $\K$.}
In {\it Logic, Language, Information, and Computation.}
Springer (2024) 47--63.
%
%
\bibitem{Gore:Shillito:2020}
Gor\'e, R., Shillito, I.:
{\it Bi-intuitionistic logics: a new instance of an old problem.}
In {\it Advances in Modal Logic.}
College Publications (2020) 269--288.
%
%
\bibitem{Grefe:1996}
Grefe, C.:
{\it Fischer Servi's intuitionistic modal logic has the finite model property.}
In {\it Advances in Modal Logic.}
CSLI Publications (1996) 85--98.
%
%
%
%
\bibitem{deGroot:et:al:2026}
de Groot, J., Shillito, I., Clouston, R.:
{\it Duality for constructive modal logics: from Sahlqvist to Goldblatt-Thomason.}
arXiv (2026) 2601.03762.
%
%
%
%
%
%
\bibitem{Hasimoto:2001}
Hasimoto, Y.:
{\it Finite model property for some intuitionistic modal logics.}
Bulletin of the Section of Logic {\bf 30} (2001) 87--97.
%
%
\bibitem{Hodges:1993}
Hodges, W.:
{\it Model Theory.}
Cambridge University Press (1993).
%
%
\bibitem{Hughes:Cresswell:1996}
Hughes, G., Cresswell, M.:
{\it A New Introduction to Modal Logic.}
Routledge (1996).
%
%
%
%
%
%
\bibitem{Kalmar:1937}
Kalm\'ar, L.:
{\it Zur\"uckf\"uhrung des Entscheidungsproblems auf den Fall von Formeln mit einer einzigen, bin\"aren, Funktionsvariablen.}
Compositio Mathematica {\bf 4} (1937) 137--144.
%
%
%
%
%
%
%
%
%
%
\bibitem{Kripke:1965}
Kripke, S.:
{\it Semantical analysis of intuitionistic logic~$I$.}
In: {\it Formal Systems and Recursive Functions.}
Elsevier (1965) 92--130.
%
%
%
%
%
%
\bibitem{Lin:Ma:2019}
Lin, Z., Ma, M.:
{\it Gentzen sequent calculi for some intuitionistic modal logics.}
Logic Journal of the IGPL {\bf 27} (2019) 596--623.
%
%
%
%
%
%
%
%
%
%
\bibitem{Marin:et:al:2021}
Marin, S., Morales, M., Stra\ss burger, L.:
{\it A fully labelled proof system for intuitionistic modal logics.}
Journal of Logic and Computation {\bf 31} (2021) 998--1022.
%
%
%
%
%
%
%
%
%
%
%
%
%
%
\bibitem{Olivetti:2022}
Olivetti, N.:
{\it A journey in intuitionistic modal logic: normal and non-normal modalities.}
In {\it LATD 2022 and MOSAIC Kick Off Conference.}
University of Salerno (2022) 12--13.
%
%
%
%
%
%
%
%
%
%
%
%
\bibitem{IMLA:2017}
de Paiva, V., Artemov, S. (editors):
{\it Intuitionistic Modal Logic 2017.}
Journal of Applied Logics {\bf 8} (2021) special issue.
%
%
%
%
%
%
%
%
%
%
%
%
\bibitem{Plotkin:Stirling:1986}
Plotkin, G., Stirling, C.:
{\it A framework for intuitionistic modal logics.}
In {\it Theoretical Aspects of Reasoning About Knowledge.}
Morgan Kaufmann Publishers (1986) 399--406.
%
%
%
%
\bibitem{Prenosil:2014}
P\v{r}enosil, A.:
{\it A duality for distributive unimodal logic.}
In {\it Advances in Modal Logic.
Volume 10.}
College Publications (2014) 423--438.
%
%
%
%
%
%
%
%
\bibitem{Rauszer:1980}
Rauszer, C.:
{\it An Algebraic and Kripke-Style Approach to a Certain Extension of Intuitionistic Logic.}
PWN~---~Polish Scientific Publishers (1980).
%
%
%
%
%
%
%
%
%
%
%
%
\bibitem{Simpson:1994}
Simpson, A.:
{\it The Proof Theory and Semantics of Intuitionistic Modal Logic.}
Doctoral thesis at the University of Edinburgh (1994).
%
%
%
%
\bibitem{Sotirov:1984}
Sotirov, V.:
{\it Modal theories with intuitionistic logic.}
In {\it Mathematical Logic.}
Publishing House of the Bulgarian Academy of Sciences (1984) 139--171.
%
%
%
%
%
%
%
%
\bibitem{Takano:2003}
Takano, M.:
{\it Finite model property for an intuitionistic modal logic.}
Nihonkai Mathematical Journal {\bf 14} (2003) 125--132.
%
%
%
%
%
%
%
%
%
%
%
%
%
%
%
%
%
%
%
%
%
%
%
%
%
%
%
%
\bibitem{Wechler:1992}
Wechler, W.:
{\it Universal Algebra for Computer Scientists.}
Springer (1992).
%
%
\bibitem{Wijesekera:1990}
Wijesekera, D.:
{\it Constructive modal logics~I.}
Annals of Pure and Applied Logic {\bf 50} (1990) 271--301.
%
%
\bibitem{Wolter:1998a}
Wolter, F.:
{\it On logics with coimplication.}
Journal of Philosophical Logic {\bf 27} (1998) 353--387.
%
%
%
%
\bibitem{Wolter:Zakharyaschev:1997}
Wolter, F., Zakharyaschev, M.:
{\it The relation between intuitionistic and classical modal logics.}
Algebra and Logic {\bf 36} (1997) 73--92.
%
%
\end{thebibliography}
\end{document}